\documentclass[12pt,a4paper]{article}
\usepackage{natbib}
\usepackage{xr}
\usepackage{subfiles}

\RequirePackage{amsthm,amsmath,amsfonts,amssymb}
\usepackage{graphicx}
\usepackage{algpseudocode,algorithm,algorithmicx}
\usepackage{enumerate}
\usepackage{bm,bbm}
\usepackage{verbatim}

\newtheorem{theorem}{Theorem}[section]
\newtheorem{lemma}[theorem]{Lemma}
\theoremstyle{remark}
\newtheorem{definition}[theorem]{Definition}

\DeclareMathOperator*{\argmax}{arg\,max}

\newtheorem{corollary}{Corollary}
\newtheorem{proposition}{Proposition}
\newtheorem{remark}{Remark}

\algrenewcommand\algorithmicrequire{\textbf{Initialisation:}}
\algrenewcommand\algorithmicensure{\textbf{Output:}}

\usepackage{xcolor}
\usepackage{enumitem}
\newcommand{\hrefitem}[2]{\hyperref[#1]{#2}}
\usepackage[colorlinks,citecolor=blue,urlcolor=blue,filecolor=blue,backref=page]{hyperref}
\usepackage[noabbrev]{cleveref}

\newcommand{\of}[1]{{\left( #1 \right)}}
\newcommand{\bc}[1]{\left\lbrace #1 \right\rbrace} 
\newcommand{\bp}[1]{\left( #1 \right)} 
\newcommand{\bs}[1]{\left[ #1 \right]} 
\newcommand{\abs}[1]{\left\lvert #1 \right\rvert}
\newcommand{\absx}[1]{\lvert #1 \rvert}

\newcommand{\cD}{\mathcal{D}}

\newcommand{\cN}{\mathcal{N}}

\newcommand{\cX}{\mathcal{X}}

\newcommand{\PP}{\mathbb{P}}
\newcommand{\RR}{\mathbb{R}}

\newcommand{\NN}{\mathbb{N}}

\newcommand{\vertiii}[1]{{\left\vert\kern-0.25ex\left\vert\kern-0.25ex\left\vert #1 \right\vert\kern-0.25ex\right\vert\kern-0.25ex\right\vert}}

\newcommand{\intd}{\mathrm d}
\usepackage{dsfont}
\newcommand{\sgn}[1]{\operatorname{sgn}\!\left(#1\right)}
\newcommand{\1}[1]{\mathbbm 1_{#1}}

\newcommand{\EE}[2][]{\mathbb{E}_{#1} {\left[ #2\right]}}

\title{Annealed Leap-Point Sampler\\for Multimodal Target Distributions}
\author{Nicholas G. Tawn \and Matthew T. Moores \and Hugo Queniat \and Gareth O. Roberts}
\date{}

\begin{document}

\maketitle

\abstract{
In Bayesian statistics, exploring high-dimensional multimodal posterior distributions poses major challenges for existing MCMC approaches. This paper introduces the Annealed Leap-Point Sampler (ALPS), which augments the target distribution state space with modified annealed (cooled) distributions, in contrast to traditional tempering approaches. The coldest state is chosen such that its annealed density is well-approximated locally by a Laplace approximation. This allows for automated setup of a scalable mode-leaping independence sampler. ALPS requires an exploration component to search for the mode locations, which can either be run adaptively in parallel to improve these mode-jumping proposals, or else as a pre-computation step. A theoretical analysis shows that for a d-dimensional problem the coolest temperature level required only needs to be linear in dimension, $\mathcal{O}\left(d\right)$, implying that the number of iterations needed for ALPS to converge is $\mathcal{O}\left(d\right)$ (typically leading to overall complexity $\mathcal{O}\left(d^3\right)$ when computational cost per iteration is taken into account). ALPS is illustrated on several complex, multimodal distributions that arise from real-world applications. This includes a seemingly-unrelated regression (SUR) model of longitudinal data from U.S. manufacturing firms, as well as a spectral density model that is used in analytical chemistry for identification of molecular biomarkers.
}

\section{Introduction}\label{sec:intro}

The standard approach to sampling from a multivariate Bayesian posterior density is Markov Chain Monte Carlo (MCMC).
There is a wealth of established literature and  methodology enabling the construction of a Markov chain with prescribed invariant distribution   \citep[e.g.,][]{Robert2013}. Arguably the easiest way to construct such a chain is by using a localised proposal mechanism 
e.g., random walk Metropolis \citep{metropolis1953equation}, Metropolis-adjusted Langevin \citep{rossky1978brownian,roberts1996exponential}, or hybrid methods formed from concatenating local moves such as Hamiltonian Monte Carlo \citep{duane1987hybrid}.
Such proposals can then be accepted or rejected according to the Metropolis-Hastings acceptance probability \citep{hastings1970monte}, ensuring reversibility of the chain.

When the target exhibits multimodality, the majority of MCMC algorithms that use localised proposals will fail to explore the state space, leading to biased estimation of posterior expectations from the sample output \citep{roberts2001optimal,Mangoubi2018}. Indeed, when the target distribution is high-dimensional or the posterior arises from a complex Bayesian modeling context then it becomes very hard to know whether the posterior is multimodal. Therefore it is important to develop methods that are robust but also scale well to high-dimensional settings.

The most popular MCMC methods designed to explore multimodal distributions have traditionally been based on {\em tempering} approaches and building on the seminal contributions of \citet{geyer1991markov,marinari1992simulated,Geyer1995}. However, despite the success of these methods, tempering approaches often fail in high-dimensional settings, and where individual modes have different variances and/or display asymmetries \citep{woodard2009sufficient}.

This paper introduces a new algorithm, the Annealed Leap-Point Sampler (ALPS), that overcomes many of these disadvantages. It is designed to sample from a class of smooth densities on $\mathbb{R}^d$ that may exhibit multimodality. The ALPS construction uses 
an annealing approach, where contrary to traditional approaches it relies on raising the density to powers greater than one. It also adapts ideas from established methods in the literature designed to mitigate the worst effects of multimodality in multi-dimensional settings, such as mode-jumping proposals \citep{Tjelmeland2001} and weight-stabilised tempering \citep{Tawn2018a,GarethJeffNick} that incorporate non-local information about the target.

As well as presenting novel methodology, this paper introduces accompanying theoretical results that 
describe its computational complexity (under suitable regularity) and which underpin practical guidance on temperature selection for the algorithm. The conclusion is that the largest power 
$\beta_{\text{max}}$ in the sequence of annealing distributions should be selected so that it scales as $\mathcal{O}(d)$. Therefore, building on the theoretical results in \citet{Roberts2020} we are able to deduce that the number of iterations required should scale as $\mathcal{O}(d)$, typically leading to an overall
computational cost complexity of $\mathcal{O}\left(d^{3}\right)$ for ALPS when taking into account computational expense per iteration. See further discussion in Section \ref{sec:compcomplex}.
It should be noted that these complexity calculations do not include the cost of searching for the mode locations, an essential problem for all multi-modal MCMC approaches, and which is a challenging problem in its own right. We incorporate an exploration component in ALPS that searches for new modes, as well as maintaining a list of modes that it has already found. This can either be run as one or more parallel chains to adaptively improve the mode-jumping proposal distribution, or else as a pre-computation step separate from ALPS.

The merits of the ALPS algorithm are demonstrated empirically, firstly on a synthetic example where the modes are known, then on a Seemingly-Unrelated Regression (SUR) model for manufacturing investment data \citep{Grunfeld1958}, and lastly on a peak decomposition model for Raman spectroscopy \citep{Moores2016raman}. Bayesian analysis of spectroscopy was pioneered by \citet{Ritter1994} and \citet{van2001analysis}. It is a proven approach in analytical chemistry for many different types of spectroscopic data \citep[e.g., ][]{wang2008reversible,zhong2011bayesian,Harkonen2023}. As with SUR, this is a popular model for applications but its problems with multimodality have seldom been acknowledged or addressed.

SUR models were introduced by \cite{Zellner1962} and are widely used for panel data in econometrics. Many examples are provided by \cite{Srivastava1987} and \cite{Fiebig2001}. It has only relatively recently been established that the SUR model likelihood exhibits multimodality and that the current methodology for inference is not robust.
Our paper is the first time that multimodality has been demonstrated for the \citet{Grunfeld1958} data.
We show that ALPS can successfully overcome the difficulties in both the synthetic problem as well as the complex real-data applications of SUR and spectroscopy. 

The paper is structured as follows:
after giving background and motivation in \Cref{sec:motivation},
\Cref{sec:Components} provides a detailed description of each component  of ALPS; Section~\ref{sec:Procedure} covers their assembly into the complete method; Section~\ref{sec:theory}  presents the main theoretical contribution of the paper that is crucial to understanding the computational cost of the procedure; Section~\ref{sec:compcomplex} discusses the computational cost and how it scales with the dimension of the problem; Section~\ref{sec:Empiricalexamples} contains empirical studies; and Section~\ref{sec:conclusions} presents our conclusions and considerations for further work.
Detailed proofs, along with full R source code for the applications, are contained in the online supplementary material.

\section{Background and Motivation}
\label{sec:motivation}

Consider the problem of sampling from a $d$-dimensional smooth density function $\pi(\cdot)$, where it is either known or suspected that the target distribution is multimodal. Such problems arise frequently across many fields e.g., astrophysics \citep{feroz_multinest_2009}; biology \citep{Ballnus2017benchmarking} or signal processing \citep{Fong2019}. Additionally, multimodality can arise due to non-identifiability in statistical models, as commonly seen in neural networks where symmetries in parametrization induce multiple equivalent modes \citep{papamarkou_challenges_2022}.  These kinds of multimodal scenarios are well known to present significant challenges for even the most advanced MCMC algorithms, often resulting in poor mixing or incomplete exploration of the posterior distribution \citep[e.g., ][]{Mangoubi2018}.

There is a significant volume of existing research focused on overcoming the multimodal sampling problem, which can be roughly designated into one of two approaches. The most direct approach uses ``mode jumping'' MCMC proposals,  e.g.\ \cite{Tjelmeland2001}, \cite{multtryliu}, \cite{Tjelmeland2004}, \cite{baththesis09}, \cite{zhou2011multi,zhou2011random} and \cite {Pompe2018}.

However the most successful and widely adopted approach to the problem are based on state-space augmentation techniques.
These stem from two seminal, closely related but independently discovered algorithms:
 the parallel tempering (PT) algorithm \citep{geyer1991markov} and its closely-related cousin, the simulated tempering algorithm  \citep{marinari1992simulated}. The principle behind these methods is that a Markov chain can achieve global exploration when targeting a tempered version of the density ($\pi^\beta$ for some $0<\beta<1$). Therefore the strategy is to construct an MCMC algorithm on an augmented space, ${\mathbb{R}^d \times \Delta}$ where $\Delta=\{\beta_j, j=0, \ldots n\} \subset(0,1] $, with $\beta _0=1$ denoting the  temperature of interest. Since $\pi $ is assumed to be multi-modal, direct Markov chain transitions between its modes are likely to be very difficult. However, indirect paths between these modes can be achieved by moving to {\em hotter} temperatures (smaller $\beta $) where inter-mode transitions are feasible, and then moving temperature back to $\beta _0=1$.

The success of this simple strategy has inspired generalisations to a more general class of  augmentation methods, see for example \cite{Wang1990a}, \cite{geyer1991markov}, \cite{marinari1992simulated}, \cite{Neal1996}, \cite{kou2006discussion}, \cite{atchade2010wang}, \cite{2017arXiv170805239N}, \cite{Tawn2018} and \cite{syed2022non}.

Both classes of approaches have their strengths and weaknesses. Part of the motivation for the novel methodology presented in this paper is to design an algorithm that combines their strengths while mitigating as many of their weaknesses as possible.

Despite their many successes, there are significant negative aspects regarding the performance of tempering-based MCMC approaches which afflict these algorithms particulary in moderate to  high-dimensional settings, e.g.\ \cite{woodard2009conditions}, \cite{woodard2009sufficient} and \cite{Bhatnagar2016}. These problems are caused by weight degeneracy, whereby a mode which may have significant probability mass in the temperature of interest could become negligible at a different temperature.
There have been recent developments in \cite{Tawn2018a,Tawn2018} that mitigate these problems and thus accelerate the performance of PT.  However, these approaches lack robustness when the modal components of the target distribution exhibit significant asymmetry and heterogeneity of scale as illustrated in Figure~\ref{fig:ALPSintuition1}. 
\cite{GarethJeffNick} attempts to address the issues of heterogeneous scale, by applying tempering in a localised way. However, there remain questions over the high-dimensional practicality of the approach, especially regarding the Markov chain's convergence properties for the tempered target distributions that are essential to ensuring successful global exploration. 

Alternative approaches, e.g. \cite{Tjelmeland2001} and \cite{baththesis09}, attempt to approximate the target distribution via local Laplace approximations in order to design appropriate inter-modal MCMC moves. However, their mixing times increase exponentially in dimension \citep{Liu1996}. Performance is particularly problematic where the modal components exhibit skewness. Even in one dimension, the best Gaussian mixture approximation to the target distribution might not be a good fit, see for example Figure~\ref{fig:ALPSintuition1} at $\beta=1$.

Overcoming these problems and developing a practically applicable algorithm is the motivation behind this work. The most significant novelty in this work  is  that an  annealing  approach (i.e. powers\  $\beta>1$) is adopted to reduce the skew of any particular local component, making it appear approximately Gaussian. From the upper two plots in Figure \ref{fig:ALPSintuition1} we can clearly see the effect of reducing skew by annealing. This greatly facilitates the design of a Markov chain that has efficient inter-modal mixing properties (in practice we use a mode-hopping independence Metropolis-Hastings sampler).

\begin{figure}
  \centering
  \includegraphics[width=\linewidth]{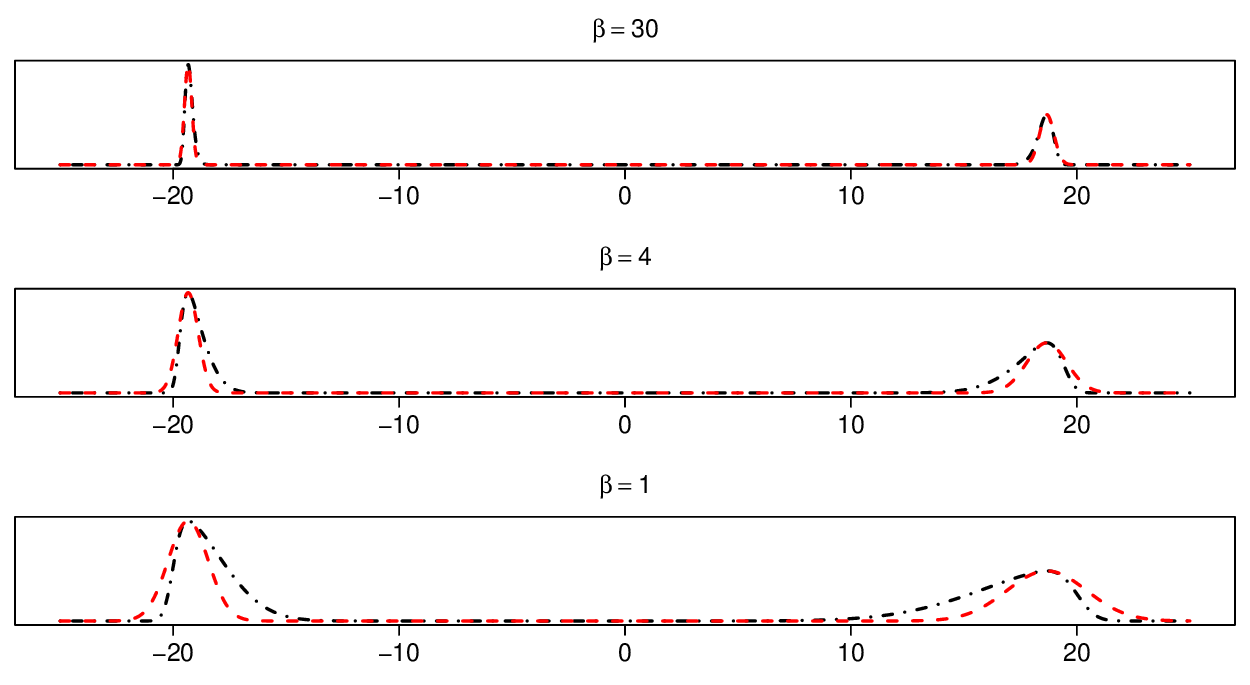}
  \caption{Annealing reduces the effect of skew. The dot-dashed line in the bottom panel shows a bimodal density with appreciable skewness; overlayed by the red dashed line representing the Gaussian mixture approximation derived through a local Laplace approximation. The two figures above this show annealed versions of the density and its respective approximation (at annealings $\beta =4$ and $\beta=30$ respectively); demonstrating that as the modes are annealed the deviation from local Gaussianity reduces. 
  }
  \label{fig:ALPSintuition1}
\end{figure}

We shall show that this ambitious plan can be accomplished by giving detailed methodology and supporting theory that demonstrates its robustness in high dimensional contexts. In particular, we demonstrate that the proposed annealing strategy achieves linear mixing-time scaling in dimension, addressing regimes in which standard tempering methods may exhibit exponentially slow convergence, e.g. \citet{woodard2009sufficient}.

This linear behaviour persists even for well-separated, skewed modes and heterogeneous covariance structures, ensuring efficient mixing as dimension increases. The resulting gain in Markov chain convergence shifts the complexity from exponential to polynomial, at the expense of a modest increase in per-iteration computational cost, preserving polynomial-time scalability.

\section{Components of ALPS}
\label{sec:Components}

ALPS consists of two components, \textit{Parallel Annealing} (ALPS-PA) which is an annealing ``parallel tempering'' scheme  and the \textit{Exploration Component} (ALPS-EC) which is a procedure aimed at exploring the whole state space to identify the locations of modes. These two processes can be run in parallel essentially creating an adaptive MCMC algorithm
\citep{roberts2009examples}.

As with any parallel tempering method, ALPS-PA  constructs a Markov chain on the space 
\def\state{{\cal X}}
$\mathbb{R}^{d(n+1)}=\prod_{j=0}^n\state _{\beta _{j}}$, where each $\state _{\beta _{j}}$ is an independent copy of $\mathbb{R}^{d}$. We denote the value of the chain at iteration $t$  by $X^t:= (
x_0^t,\ldots,x_n^t)$, which is targeting an invariant distribution
\begin{equation}
	\pi_{ALPS}(X) \propto 
 \prod_{j=0}^n \pi_{\beta_j}(x_j), \label{eq:ALPSinvar}
\end{equation}
where $\pi_{\beta_j}$ are an appropriate sequence of densities. Crucial to the success of ALPS-PA though will be
\begin{enumerate}[label={(\hrefitem{sec:3.\arabic*}{3.\arabic*})}, ref=3.\arabic*]
    \item the choice of the auxiliary densities $\pi_{\beta_j}$;
    \item the tempering schedule $\Delta= \{\beta_0, \beta_1, \ldots, \beta_{n}  \}$,  which determines a sequence of auxiliary probability densities;
    \item the choice of temperature swap moves that have accelerated mixing properties through the aid of a deterministic transformation between $\state_{\beta _j}$ and $\state_{\beta_{j\pm 1}}$;
    \item the designing of $P_{\beta_j}(\cdot,\cdot)$, the Markov chain transition mechanism for moves conditional on $\beta _j$, facilitating {\em within-mode} mixing;
    \item the collection of transition kernels $\{  Q_{\beta_0}(\cdot,\cdot),Q_{\beta_1}(\cdot,\cdot),\ldots,Q_{\beta_n}(\cdot,\cdot) \}$ designed to allow {\emph{inter-modal}} jumps;
    \item the mode finding exploration component of ALPS to discover the local maxima of the target distribution.
\end{enumerate}

The coming subsections will be devoted to describing and justifying the choices made in (3.1)--(3.6). 

\citet{Tjelmeland2001} and \citet{Pompe2018} showed that if the modes can be well-approximated by Gaussian distributions (typically derived from a Laplace approximation) then an independence sampler-type proposal based on a Gaussian or Gaussian mixture is very effective for mode-jumping. 

However, in high dimensional settings the challenge is that even small deviations (marginally) from Gaussianity lead to degenerate acceptance rates of such moves. \citet{Pompe2018} overcome this by fitting a more involved Gaussian mixture to the often complex modal structures. ALPS takes a different approach, it seeks to naturally and automatically adapt the target distribution into something that is increasingly well-approximated by a Gaussian mixture distribution with well-separated components i.e., a target where all modes appear locally Gaussian.

A central component in the design of the algorithm is the following Gaussian mixture distribution. 
Given a list of vectors  $M=\{\mu_j,\  1\le j\le m\}$, a list of symmetric positive-definite matrices $S=\{\Sigma_j,\ 1\le j\le m\}$, and a list of positive weights summing to unity $W=\{w_1, \ldots, w_m\}$, we set
\begin{align}
    \pi_{M,S,W}(x) = \sum_{j=1}^m w_j 
     \phi \left(x|\mu_j,{\Sigma_j}\right) \label{eq:GM}
     \end{align}
where $\phi(\cdot|\mu_j,\Sigma_j)$ denotes the pdf of a multivariate Gaussian with mean $\mu_j$ and covariance matrix $\Sigma_j $.

If the target distribution has modes that locally resemble Gaussian distributions, it can naturally be approximated effectively by a Gaussian mixture like the one in \eqref{eq:GM}, where $m$ would correspond to the number of local maxima. In such cases, provided the Gaussian mixture parameters can be derived, the problem simplifies considerably. Indeed, by combining localized MCMC moves with independence sampler moves drawn from the proposal distribution given in \eqref{eq:GM}, one can design a Markov chain that would yield very rapid convergence.

ALPS aims to exploit this construction with its very rapid convergence properties even in cases where the target has modes that deviate substantially from appearing Gaussian.

\subsection{Hessian-Adjusted Tempered (HAT) Distributions}
\label{sec:3.1}
Traditionally, tempering algorithms rely on augmentation strategies using powered densities where the entire density is raised to some power $\beta<1$. This has the effect of merging modes and thus facilitating transitions between different modal domains of attraction.

\citet{woodard2009sufficient} shows that when modes of a multimodal target exhibit heterogeneity, particularly with regards to their scales, then these power-based target distributions are a poor choice as they lead to torpid mixing of PT when the dimensionality of the problem grows. 
Unless the distributions are altered, annealing is subject to similar challenges of torpid mixing. The chains targeting the most annealed distributions become increasingly restricted to the dominant modes corresponding to global maxima, as in simulated annealing \citep{aarts1988simulated}.

The target distributions used in the PA component build upon on recent work in \citet{GarethJeffNick} which aims to (approximately) preserve mass that one may associate with each mode across different temperature levels. 

First, we introduce an approximate partition of the state space 
into domains (or basins) of attraction, in the sense of \citet{Tjelmeland2001,zhou2011multi}. In this context, the domain of attraction $\cD_j \subseteq \cX $ associated with mode $\mu_j$ is defined as the set of points from which the gradient flow dynamics (i.e., the solutions to the gradient-flow ODE) converge to $\mu_j$, as introduced in \citet{zhou2011multi}. Given the Gaussian mixture distribution $\pi_{M,S,W}$,  we derive an approximation of these domains, based on a modified Mahalanobis distance criterion (dependent on $M, S, W$, although we omit this explicit dependence for notational convenience).

Formally, the partitioning is provided by the following \textit{mode assignment function} at inverse temperature $\beta$,
\begin{equation}
A_{x, \beta} := \argmax_{j \in \{1,\dots,m\}} \left\{w_j\phi\left(x\mid\mu_j,\frac{\Sigma_j}{\beta}\right)\right\} \label{eq:ass}
\end{equation}
where $\phi(\cdot|\mu_j,\Sigma_j)$ denotes the pdf of a multivariate Gaussian with mean $\mu_j$ and covariance matrix $\Sigma_j$.

This assignment produces a partition consistent with the gradient-flow definition of domains of attraction when the target distribution is a Gaussian mixture. Through annealing, the modes approach local Gaussianity, and the approximation of these domains becomes increasingly accurate. Using this partition, we follow \citet[Def.~2]{GarethJeffNick} and restate the Hessian Adjusted Tempered (HAT) distributions, which we adopt as our auxiliary densities~$\pi_{\beta_j}$.

\begin{definition}[Hessian Adjusted Tempered (HAT) distributions]
\label{def:HAT}
Given an inverse-temperature level $\beta$ and its associated \textit{mode assignment function} $A_{x, \beta}$, the HAT distribution at inverse temperature $\beta$ is defined as
\begin{equation}
	\pi_{\beta}(x)  \propto
    \pi(x)^{\beta} \pi(\mu_{A_{x, \beta}})^{1-\beta}.
    \label{eq:robadj1}
\end{equation}
\end{definition}

Theoretical justification for these distributions can be found in the supplementary material in Appendix~\ref{sec:Weightstabilised}. In short, this form of the target distribution can mitigate some of the most undesirable weight stabilisation issues exhibited by traditional tempering approaches.

\subsection{Choice of inverse temperatures} 
\label{sec:3.2}
ALPS-PA will carry out an MCMC algorithm on the density $$
\Pi (x_0, \ldots ,x_n) \propto \prod_{i=0}^n \pi_{\beta_i} (x_i)
$$
where $\Delta = \{\beta_0, \beta_1, \ldots, \beta_n\}$ is a user-defined {\em temperature schedule}. By convention $\beta_0$ is taken to be $1$, the temperature of interest.  Unlike other tempering approaches, the other temperatures, $\Delta -\{\beta_0\}$,  are taken to be a collection of
colder temperatures, so that $\beta_0<\beta_1< \ldots < \beta_n$.

There are two important aspects to the choice of $\Delta$, the choice of $\beta_n $ and the spacing of the intermediate temperatures. For the former problem, we will derive theory in Section~\ref{sec:theory} which shows that we need to take $\beta_n=\mathcal{O}(d)$.

The spacings between consecutive temperatures in $\Delta $ must be chosen carefully. Appropriate choices maximize the round-trip rate, which measures how frequently complete cycles occur—beginning and ending at $\beta_0$, with an intermediate visit to $\beta_n$ \citep{syed2022non}. Optimizing this rate ensures rapid transfer of mixing information from the most annealed state to improve mixing at the target state. If spacings are too large, the acceptance rates for temperature swaps become excessively low. Conversely, excessively small spacings lead to slow mixing (and increased computational complexity) since there are too many intermediate temperature levels and it takes a long time to pass the information through the temperature schedule.

Appealing to the theoretical contributions of \cite{Tawn2018a}, \cite{Tawn2018} and \cite{GarethJeffNick}, it is approximately optimal to select the spacings consecutively with $\beta_{j+1}=\beta_{j} +\epsilon_j$ where $\epsilon_j = \mathcal{O}\left(d^{-1/2}\right)$. Crucially, the optimal choice for the spacings corresponds to temperature swap acceptance rates between consecutive temperatures of 0.234. A practitioner can use this to tune the setup of the algorithm; even pursuing an adaptive strategy similar to  \cite{Miasojedow2013}.

Implementing this scheme robustly would usually start with a geometrically-spaced schedule (optimal for all power densities as shown in \cite{atchade2011towards}), and then fine-tuned to approach the optimal acceptance rate at all temperatures. Details can be found in 
\cite{Tawn2018a,Tawn2018}.

\subsection{Design of temperature swap moves}
\label{sec:3.3}

Within ALPS-PA we shall make extensive use of the QuanTA transformation \citep[as originally introduced in][]{Tawn2018}.

\begin{definition}[QuanTA Transformation] \label{def:quantyy}
Given a tempering schedule $\Delta $ with $\beta_0<\ldots < \beta _n$ and lists $M,S,W$, then  define $T^+_{M,S,W}:\mathbb R^d \times
\{\beta _0, \ldots, \beta _{n-1}\}
\to \mathbb R ^d\times \{\beta _1, \ldots, \beta _{n}\}$ and $T^-_{M,S,W}:\mathbb R^d \times
\{\beta _1, \ldots, \beta _{n}\}
\to \mathbb R ^d\times \{\beta _0, \ldots, \beta _{n-1}\}$
by
\begin{subequations}\label{eq:QuTrans}
\begin{align}
T^+_{M,S,W}(x,\beta_i) &= \left(
\left(\frac{\beta_i}{\beta_{i+1}} \right)^{1/2} (x-\mu_{A_{x,\beta_i}}) + \mu_{A_{x,\beta_i}}, \ 
\beta_{i+1}
\right), \label{eq:QuTrans_plus} \\
T^-_{M,S,W}(x,\beta_i) &= \left(
\left(\frac{\beta_i}{\beta_{i-1}} \right)^{1/2} (x-\mu_{A_{x,\beta_i}}) + \mu_{A_{x,\beta_i}}, \ 
\beta_{i-1}
\right). \label{eq:QuTrans_minus}
\end{align}
\end{subequations}
\end{definition}
Note that, restricted to points for which the transformation remains in the same basin of attraction, the maps $T^+_{M,S,W}$ and $T^-_{M,S,W}$ are injective (though not surjective). Hence their inverses are well defined on their respective images within this reversibility-preserving region. If a proposed transformed point crosses into a different basin, this restriction fails and such proposals are immediately rejected in Step 3. below.

As in the standard framework of parallel tempering, pairs of adjacent temperature level components are chosen at random and for a swap move. Traditionally this simply involves proposing a swapping of the locations of the adjacent temperature components. \cite{Tawn2018} introduced a novel swap move that simultaneously proposes a deterministic transformation based on \eqref{eq:QuTrans} with the motivation of accelerating the mixing through the temperature space.

At time $t$, the procedure for a temperature swap move in ALPS follows a very similar strategy to that in \cite{Tawn2018}:
\begin{enumerate}
    \item Select a pair of temperatures $\beta_i,\beta_{i+1}$ uniformly at random
    \item Compute the transformations
    \[
    y_{i+1}=\left(T^+_{M,S,W}(x_i^t,\beta _i)\right)_1~~~\mbox{and}~~~y_{i}=\left(T^-_{M,S,W}(x_{i+1}^t,\beta _{i+1})\right)_1
    \]
    \item Check whether the transformations preserve basin membership which is required for reversibility. Reject the proposed swap if
    \[
    A_{x^t_i,\beta_i }\ne A_{y_{i+1},\beta_{i+1} }~~~\mbox{or}~~~A_{x^t_{i+1},\beta_{i+1} }\ne A_{y_{i},\beta_i }
    \]
    and set $x_{i+1}^{t+1}=x_{i+1}^t$ and $x_{i}^{t+1}=x_{i}^t$.
    \item If 3. does not lead to rejection then compute
    \[
    \alpha = \min \left\{ 1, \frac{\pi_{\beta_{i}}(y_i)\pi_{\beta_{i+1}}(y_{i+1})}{\pi_{\beta_{i}}(x_i^{t})\pi_{\beta_{i+1}}(x_{i+1}^{t})} \right\}.
    \]
    With probability $\alpha$ we accept a transformation-aided swap move and set
    \[
    x_{i+1}^{t+1}=y_{i+1}~~~\mbox~~~x_{i}^{t+1}=y_{i}\]
    else set
    \[
    x_{i+1}^{t+1}=x_{i+1}^t~~~\mbox~~~x_{i}^{t+1}=x_{i}^t.\]
\end{enumerate}

\cite{Tawn2018} show that, under mild smoothness conditions, for large inverse temperatures  $\beta$ the traditional parallel tempering scheme is optimal with spacings that are $\mathcal{O}(\beta)$, while this transformation-enhanced swapping scheme can achieve spacings that are up to $\mathcal{O}(\beta^{3/2})$ and $\mathcal{O}(\beta^{5/4})$ when the modes are symmetric and  asymmetric respectively in a neighbourhood of the mode point.

The transformation is more effective when the modes exhibit approximately Gaussian structure. For robustness in highly degenerate cases where QuanTA swaps may start to fail or, for instance, when modes are highly non-Gaussian, occasional traditional swap moves may optionally be included near the target temperature. In practice, we did not observe such behaviour, and empirical evidence suggests that this safeguard is not necessary.

\subsection{Choice of within-mode transitions}
\label{sec:3.4}

Algorithm moves that reposition particles within each temperature level are crucial for the algorithm’s success. The purpose of these moves is not to transition between modes—this is exclusively attempted at the highest inverse temperature $\beta_n$, see \citet{Roberts2020}. At all other temperatures, we only require that the sampler mixes adequately within each mode. Although this paper does not specify a particular within-temperature Markov chain proposal for ALPS, careful tuning of these localized moves remains crucial and will certainly be a function of the temperature level such that tuning would be needed for all temperature levels, as in Parallel Tempering. 

To increase robustness, we recommend employing sampling schemes that explicitly exploit local geometric properties which can differ significantly between modes, e.g.\ pre-conditioned Random Walk Metropolis. There are many established results for tuning the within-temperature marginal updates to the chain performed in PA component, e.g. \citet{roberts2001optimal}.

For our purposes, we have found a suitable approach to be the pre-conditioned variant of the multi-dimensional Random Walk Metropolis algorithm, in which the proposal covariance is determined by the negative inverse Hessian of the log-target evaluated at each mode. Specifically, at position $x$ and inverse temperature $\beta$, the Gaussian increment will have
\begin{align*}
    {\Sigma} = -\left[\nabla^2 \log \pi({\mu}_{A_{x,\beta}})\right]^{-1},
\end{align*}
as the proposal covariance. This strategy leverages local curvature information to enhance sampling efficiency, while avoiding the tuning challenges associated with more sophisticated algorithms such as Hamiltonian Monte Carlo \citep{neal_mcmc_2011}.

\subsection{Choice of mode-hopping mechanism.}
\label{sec:3.5}
Using the collections $M$, $S$, and $W$, we can construct a Gaussian mixture target distribution indexed by a temperature $\beta$ given by:
\begin{equation}
	q_{\beta}(x| M,S,W) :=  \sum_{j=1}^m {w}_j  \phi\left(x ~\Big|~\mu_{j},\frac{\Sigma_j}{\beta}\right).
	\label{eq:inddist}
\end{equation}
The transition kernel $Q_\beta$ will thus be constructed for inverse temperature $\beta $ as an independence sampler with proposal density $q_{\beta}(\cdot | M,S,W)$ and target distribution $\pi _{\beta }$. Given current state $x$ and proposed new state $y$ drawn from this proposal density, the move to $y$ is accepted with probability

\begin{equation}
	A_{IS}(x,y| M,S,W) := \mbox{min}\left( 1,\frac{\pi_\beta (y) q_{\beta}(x | M,S,W)}{\pi_\beta (x) q_{\beta}(y | M,S,W)} \right). \label{eq:indaccrate}
\end{equation}

This move will predominantly only be attempted at the highest inverse temperature $\beta_n$, as depicted in \cite{Roberts2020}.

The collections $M$, $S$ and $W$ will be estimated via Laplace approximations as detailed in Section~\ref{sec:3.6}. If adequate annealing has taken place so that the Laplace approximations for the local modes are reasonably accurate then  the weight approximations will closely match the regional masses in the $\beta_{\text{max}}$-level HAT target. In that case, the  independence sampler constructed above should be able to achieve high acceptance rates. 

Of course, any finite amount of annealing leaves some modal asymmetry and thus prevents acceptance rates being equal to 1. In general, one must anneal more as the dimensionality of the problem increases before the Laplace approximation becomes accurate. The question of how large $\beta_{\text{max}}$ should be with respect to the dimensionality of the problem to achieve sufficiently good performance is addressed in Section~\ref{sec:theory}. In practice, to choose such a value of $\beta_{\text{max}}$, any target acceptance rate \( a \in (0,1) \) at the cold level achieves the corresponding complexity guarantees described in Section~\ref{sec:compcomplex}. A natural choice is therefore \( a = 0.5 \), which can be increased if the resulting hopping rate is not sufficiently high for the problem at hand. Our empirical results do not indicate a clear preference for a particular value of \( a \); indeed, Section~\ref{sec:Empiricalexamples} presents experiments covering acceptance rates ranging from 0.25 to 0.9. Nevertheless, for robustness, a practitioner may choose a higher target acceptance rate, such as 0.9, to ensure frequent hopping. This does not incur computational burden under the chosen temperature schedule in Section~\ref{sec:3.3}.

\subsection{Design of mode-finding procedure}
\label{sec:3.6}
Each of the previously described mechanisms relies on the collections $M$, $S$, and $W$, which represent respectively the locations,  covariance structures, and weights of the modes.  To determine this information for each mode, one could employ a range of approaches. In ALPS, this task is achieved by the \textit{Exploration Component} (ALPS-EC) with an approach that appears to work adequately in the empirical studies.

ALPS-EC constructs a Markov chain that has an invariant distribution given by $[\pi(\cdot)]^{\beta_{\text{hot}}}$. If the inverse-temperature level, $\beta_{\text{hot}}$, is sufficiently low then, as in the standard PT algorithm, the Markov chain can explore the entire state space rapidly. After a pre-specified number of steps of the Markov chain, one performs a local optimisation (this work uses Quasi-Newton optimisation) initialised from the current location of the chain. If the optimum attained is a new mode point not previously encountered, one adds it to the collection of discovered modes.

In high-dimensional settings, or when modes are exceptionally close together, numerical instabilities can result in either false mode identification or redundant detection of modes. Both phenomena can degrade the algorithm’s performance. To mitigate this, we require that any newly identified mode be sufficiently separated from all previously detected modes. 

To assess this separation, we compare a proposed new mode $\mu^\star$ with associated local covariance estimate $\Sigma^\star$, to each previously discovered mode $\mu_k$ and its associated covariance estimate $\Sigma_k$ through the following pseudo-distance:
\begin{equation}
	D(\mu_k,\mu^{*}):= d^{-1}\max \left\{ (\mu_k-\mu^{*})^T \Sigma_k^{-1}(\mu_k-\mu^{*}), (\mu_k-\mu^{*})^T (\Sigma^{*})^{-1}(\mu_k-\mu^{*}) \right\}. \label{eq:alpseudo}
\end{equation}
This pseudo-distance is the maximum of two squared Mahalanobis distances, normalized by the dimension $d$, and therefore accounts for the local covariance structure around both the existing mode and the candidate mode, measuring the extent to which the two modes remain distinct once this local covariance structure is taken into account.
We then enforce that this pseudo-distance between the candidate mode and all existing modes exceeds a prescribed threshold, chosen according to a pre-specified quantile of the chi-squared distribution with $d$ degrees of freedom, with the normalization by $d$ applied correspondingly. This criterion helps prevent the algorithm from registering numerically indistinguishable or excessively similar modes as distinct.

Using the pre-defined tolerance, we check if the pseudo-distances from the proposed new mode $\mu^*$ to existing modes surpass this threshold. If so, we augment the estimated sets $\hat{M}$, $\hat{S}$, and $\hat{W}$ with the new mode $\mu^*$; these sets correspond to the current estimates for $M$, $S$, and $W$, as outlined below:
\begin{itemize}
	\item  $\hat{M}:=\{ \hat{\mu}_1 , \ldots, \hat{\mu}_{\hat{m}} \}$, the collection of mode points of the target distribution $\pi$, with $\hat{m}$ the number of unique mode points discovered so far;
	\item $\hat{S}:=\{ \hat{\Sigma}_1 , \ldots, \hat{\Sigma}_{\hat{m}} \}$,  where $\hat{\Sigma}_j= -\left[\nabla^2 \log (\pi(\hat{\mu}_j))\right]^{-1},~~j=1,\ldots, \hat{m}$;
   	\item $\hat{W}=\{\hat{w}_1,\ldots, \hat{w}_{\hat{m}} \}$, the collection of approximated weights \begin{equation}\hat{w_j}= \frac{\pi(\hat{\mu}_j) |\hat{\Sigma}_j|^{1/2}}{\sum_{k=1}^{\hat{m}} \pi(\hat{\mu}_k) |\hat{\Sigma}_k|^{1/2}}, ~~j=1,\ldots,\hat{m}.\label{eq:weighapp} \notag\end{equation} 
\end{itemize}

The hot-state mode finder outlined above provides a straightforward implementation of mode detection, and determining optimal EC temperature levels remains a nuanced challenge. Specifically, there is a trade off between using temperature levels that are too hot therefore the chain  drifts away from the modes, and too cold, leading to a chain that remains stuck in the local mode.

To go further and increase robustness, there are many ad-hoc modifications that can be considered: running a population of chains at different temperature levels; repeatedly refreshing the chain from already discovered modes; use different optimisation strategies; perform optimisation based on the curvature of the target; to only name a few.  

We do not claim that our approach to mode finding is optimal, nor do we provide theoretical guarantees. The process of identifying modes is computationally intensive and highly problem-dependent; no single approach can be expected to work best in all scenarios. Improving the performance of the mode-finding procedure is a focus of further work beyond the scope of this paper. In practice, it may be beneficial to adapt ALPS–EC to the characteristics of the problem at hand. For instance, in machine-learning settings, where objective functions are often evaluated using minibatches and are therefore noisy, stochastic gradient-based methods such as SGD or Adam \citep{kingma2017adammethodstochasticoptimization} are natural candidates for identifying regions of high density.

\section{ALPS : The procedure}
\label{sec:Procedure}

Bringing together the ingredients of Section~\ref{sec:Components}, we shall summarize the methodology underlying a complete iteration of the ALPS algorithm. 

Starting from the current sample \( X^t = (x_0^t, \dots, x_n^t) \) at iteration \( t \), we illustrate how the tools previously introduced in subsections (3.1)--(3.6) can be combined systematically to produce the subsequent sample \( X^{t+1} \).

Given a user-defined temperature schedule \( \Delta \), see \eqref{sec:3.2}, and current collections \(\{\hat{W},\hat{M},\hat{C}\}\), obtained through an online or prior running of ALPS-EC, see \eqref{sec:3.6}, ALPS-PA engages in the following:
\begin{enumerate}
    \item Conduct the within-temperature updates for all inverse temperatures \(\beta_j\);
\begin{enumerate}
    \item For \(0 \leq j \leq n-1\), draw \(x_j^{t+1} \sim P_{\beta_j}(x_j^t, .)\) for local exploration, see \eqref{sec:3.1} and \eqref{sec:3.4};
    \item For \(j=n\), run the mode-hopping independence sampler, \(x_n^{t+1} \sim Q(x_n^t, .)\), for global exploration, see \eqref{sec:3.5}.
\end{enumerate}
    \item Perform a predetermined number of \(s\) temperature swaps, through QuanTA-aided moves, as explained in \eqref{sec:3.3};
    \item Store the subsequent sample \(X^{t+1} = (x_0^{t+1}, \ldots, x_n^{t+1})\), go to step 1.
\end{enumerate}

\begin{remark} A comprehensive pseudocode version of ALPS, suitable for direct implementation, is presented in Appendix~\ref{sec:PseudoCode} of the online supplementary material. \end{remark}

\section{Theoretical Underpinnings of ALPS} 
\label{sec:theory}

This section introduces theoretical results addressing how the highest inverse temperature, $\beta_{\text{max}}$, must scale with dimension. The choice of $\beta_{\text{max}}$ is crucial, as the effectiveness of the ALPS algorithm fundamentally relies on ensuring the target distribution at the PA component’s highest annealed level is well approximated by a Gaussian mixture distribution.

In general, as the dimensionality of the problem increases,  $\beta_{\text{max}}$ needs to get larger to ensure that the local mode Laplace approximation is accurate. Hence choosing $\beta_{\text{max}}$ too small will mean that the PA component mixes poorly in the coldest temperature level.

Choosing $\beta_{\text{max}}$ extremely large to overcome this has a few issues. On a practical level, it  may result in  computational instability as modes collapse towards point masses. Furthermore, the HAT targets provide only approximate mass preservation, and care is therefore required when considering extremely large values of the inverse temperature; see \citet{GarethJeffNick}.
Finally, there is extra computational burden as one would need more (unnecessary) intermediate temperature levels in the PA component. 
Hence, one would like to choose $\beta_{\text{max}}$ sufficiently large to achieve non-degenerate acceptance rates but not so large that it makes round trips infrequent.

\subsection{Optimal Scaling of the Coldest Temperature Level}\label{sec:optCold}

\cite{Roberts2020} demonstrates that, at least under stylised conditions, if we can ensure that  jumping between modes is possible at the coldest temperature, then the entire algorithm (suitably time-scaled) converges to a process (Walsh Brownian Motion) which itself has a mixing time which is independent of $d$.
(While the exact result (Theorem 4) in that paper assumes that jumping between modes at the cold temperature always occurs, it is easy to see from the proof that it extends routinely to just requiring that jumps happen with probability which are bounded away from $0$.) Therefore this section will be dedicated to demonstrating that at a sufficiently cold inverse temperature (we need $\beta = \mathcal{O}(d)$) then the probability of accepting a move to a different mode in the cold temperature is bounded away from $0$.

For the sake of tractability the theoretical analysis will be focused on an approximation of the ALPS procedure. Suppose that the target is a $d$-dimensional mixture distribution where the $m$ components are comprised of location and scale transformations of a common distribution.

 Then, the approximation of the ALPS procedure that we analyse is a distribution on $\left\{ 1,2,\ldots,m \right \}\times\mathbb{R}^{d}$ where the extra dimension is acting as an augmentation that identifies which mixture component  the location is originating from. This is an extremely good approximation when the components are well separated which is the case at the coldest temperature levels that are being analysed here. 

\subsubsection{Assumptions on the Target Density}

Assuming a $d+1$-dimensional state space $\{ 1,2,\ldots, m \} \times \mathbb{R}^d$, the tempered target distributions at inverse temperature $\beta$ will be given by the iid product form
\begin{equation}
	\pi_\beta(k,x) \propto \sum_{j=1}^m w_j \left(\prod_{i=1}^d \left[\frac{1}{\sigma_{ij}} f\left(\frac{x_i -\mu_{ij}}{\sigma_{ij}}\right)\right]^\beta \left[\frac{1}{\sigma_{ij}} f\left(0\right)\right]^{1-\beta} \right)\mathbbm{1}_{\{k=j\}} . \label{eq:thmtargmod2}
\end{equation}
where $\sigma_{ij} \in \mathbb{R}_+$, $\mu_{ij}\in \mathbb{R}$  and $f$ is an un-normalised univariate density function such that:
\begin{enumerate}
	\item Assume that
	\begin{equation}
		f\in \mathcal C^5, ~~f(0)=1~~, f'(0)=0 ~~\text{and}~~\argmax_x f(x) =0 \label{eq:fbaseassump}
	\end{equation}
	with 0 being the unique global maximum.
	\item Defining $h(x):= \log f(x)$, assume 
	\begin{equation}
		H:= h''(0)<0 \label{eq:posdefmod}
	\end{equation}
	and that there exists $L\in \mathbb{R}_{+}$ such that
	\begin{equation}
		|h'''''(x)|<L. \label{eq:bddfourth}
	\end{equation}
	\item Polynomial tails such that there exist $\gamma, M \geq 1$ and $ K \in \mathbb{R}_{+}$ such that
	\begin{equation}
		\sup_{|x|>M} \frac{f(x)}{|x|^{-\gamma}} < K. \label{eq:polyT}
	\end{equation}
	\item \textit{Dutchman's Cap (Bac M\`{o}r) assumption}: \eqref{eq:posdefmod} implies that there exists $\delta_1>0$ such that for all $x \in B_{\delta_1}(0)$ then $f(x) \le \exp \left( - x^2 \frac{|H|}{4} \right)$; it is assumed that there exists a $\delta_1>\delta_2>0$ such that for all $x \in [-M,M]\backslash B_{\delta_2}(0)$ (where $M$ is specified in \eqref{eq:polyT}) then $f(x)\le  \exp \left( - \delta_2^2 \frac{|H|}{4} \right)$. Hence, for $x \in [-M,M]$ 
	\begin{equation}
  f(x) \le\begin{cases}
     \exp \left( - x^2 \frac{|H|}{4} \right), &  \text{if $|x|<\delta_2$}.\\
     \exp \left( - \delta_2^2 \frac{|H|}{4} \right), & \text{if $\delta_2\le |x| \le M$}.
  \end{cases}
	\label{eq:BacMor}
\end{equation}
\end{enumerate}
Notice that the tempered distribution in \eqref{eq:thmtargmod2} is well approximated by the HAT distributions when the modes become well separated, a regime that naturally arises as the inverse temperature $\beta$ increases. These assumptions are therefore natural and well suited to the asymptotic regime of interest.

Furthermore, the equivalent procedure in this theoretical analysis setting to the leap point independence sampler proposals is to propose from the Gaussian mixture distribution given by
\begin{equation}
	q_\beta(k,x) = \sum_{j=1}^m w_j \left(\prod_{i=1}^d \left[\frac{1}{\sigma_{ij}} \phi_\beta \left(\frac{x_i -\mu_{ij}}{\sigma_{ij}}\right)\right] \right) \mathbbm{1}_{\{k=j\}}. \label{eq:gaussianPropapoproxtheory}
\end{equation}
with $\phi_\beta$ denoting the pdf of a Gaussian of the form $N\left( 0, \left(\beta|H|\right)^{-1} \right)$.

Moreover, we denote the expected acceptance probability of a Mode Leap Point Independence sampler proposal at the coldest temperature level in this $d$-dimensional setting as
\begin{equation}
a(d) = \mathbb{E}(A_{IS}(x,y)), \label{eq:asymptotic_acceptance_independence_sampler}
\end{equation}
where the independence sampler acceptance probability, $A_{IS}$, is defined in \eqref{eq:indaccrate} and the expectation is taken with respect to $x\sim \pi_{\beta_{max}}$ as defined in \eqref{eq:thmtargmod2} and $y \sim q_{\beta_{max}}$ as defined in \eqref{eq:gaussianPropapoproxtheory}.

Under the above conditions, we now state the following result about the temperature required to ensure that  effective inter-modal mixing can occur.

\begin{theorem}[Dimensionality-scaling for the Coldest Temperature Level] \label{Thm:scaling}

Assume that the Mode Leap Point Independence sampler with proposal distribution given in \eqref{eq:gaussianPropapoproxtheory} is used to target the $(d+1)$-dimensional target distribution specified by \eqref{eq:thmtargmod2}, where the marginal iid components $f$ satisfy \eqref{eq:fbaseassump}, \eqref{eq:posdefmod}, \eqref{eq:bddfourth}, \eqref{eq:polyT} and \eqref{eq:BacMor}. 

If \begin{equation}
	\beta_{\text{max}} = \ell d  .  \label{eq:scalingnec_smooth}
\end{equation} for some $\ell \in \mathbb{R}_+$ then
\begin{equation}
    \lim _{d \rightarrow\infty} a(d) = 2 \Phi \left( -\frac{1}{\sqrt{2}} \sqrt{\frac{5h'''(0)^2}{12 \ell (-h''(0))^3}}   \right) >0 \nonumber
\end{equation}
where $h(x)=\log f(x)$ and $\Phi$ is the CDF of a standard Gaussian.
\end{theorem}
\begin{proof}
See Appendix~\ref{proof:regular}.
\end{proof}
\begin{remark}
\label{rem:whyneedcold}
Theorem~\ref{Thm:scaling} has been stated within the regime where the asymptotic acceptance probability is positive so that the algorithm can successfully move between modes. However its analysis as seen in the proofs also illustrate why simpler approaches which do not use annealed temperatures are flawed. Suppose we attempt independence sampler moves at the temperature of interest, then
\begin{equation}
	a(d)\approx 2 \Phi \left( -\sqrt{\frac{d}{2} }\sqrt{\frac{5 h'''(0)^2}{12  (-h''(0))^3}}   \right) \nonumber
\end{equation}
which is exponentially small in $d$ and therefore in even moderate dimensional problems, this approach won't work, even in this simplified IID target context.
\end{remark}

The limiting acceptance probability in Theorem~\ref{Thm:scaling} depends explicitly on the local asymmetry through the third derivative of the log-density at the mode.
Specifically, the quantity
\begin{align}
\gamma
\;:=\;
-\,\frac{h'''(0)}{\bigl(-h''(0)\bigr)^{3/2}}
\label{eq:skewness}
\end{align}
represents the \emph{local skewness} of the target density around its mode.
This standardized third derivative arises naturally in third-order Laplace and saddlepoint approximations and quantifies departures from local Gaussian behaviour; see, for example, \citet{danielsSaddlepointApproximationsStatistics1954}.
From a Bayesian perspective, such skewness occurs naturally in posteriors arising from asymmetric or non-Gaussian priors, for instance 
a Gamma prior combined with a Poisson likelihood. More generally, local skewness reflects an intrinsic asymmetry of the distribution in a neighbourhood of its mode, whereby probability mass accumulates unevenly on either side of the mode.

Theorem~\ref{Thm:scaling} shows that, even as the magnitude $|\gamma|$ increases and the local asymmetry becomes more pronounced, the Mode Leap Point Independence sampler retains a strictly positive asymptotic acceptance probability under linear scaling of the coldest inverse temperature level with dimension. Moreover, any target acceptance probability in $(0,1)$ can be achieved, under linear scaling $\beta_{\max} = \ell d$, by selecting a sufficiently large constant $\ell$, with higher skewness requiring a larger value of this constant while preserving the same linear scaling with dimension. This shows that ALPS can accommodate strong skewness without altering its mixing capabilities.

\subsection{Relaxing the regularity conditions}\label{sec:optColdRelaxed}
The previous section established the primary result under certain smoothness and regularity conditions on the target density.
Yet, it remains an open question whether these assumptions can be loosened.
While $C^2$ smoothness at the mode is necessary to ensure weak convergence to a Gaussian distribution, Theorem~\ref{Thm:scaling} highlights the central role played by the local skewness coefficient \eqref{eq:skewness} at the mode, in determining the performance of the algorithm.
In the favourable case where this coefficient vanishes, the local behaviour of the target is effectively symmetric, and the $\beta_{\max}=\mathcal{O}(d)$ scaling becomes more than sufficient, yielding an asymptotic acceptance probability satisfying $\lim_{d\to\infty} a(d)=1$.
This suggests that in the absence of local skewness, a substantially milder annealing regime—potentially sublinear in $d$—may already suffice to maintain a non-vanishing acceptance rate for the Mode Leap Point Independence sampler.
This naturally raises the question of what happens when skewness arises not from a smooth third derivative, but instead from piecewise-smooth behaviour, leading to a discontinuity of the third derivative at the mode.

To explore the consequences of relaxing the continuity of the third derivative, we turn once again to a tempered target $\pi_\beta$ in the form \eqref{eq:thmtargmod2}, relying on the original assumptions for $f$ except for the following relaxed ones.
\begin{enumerate}
    \item Assume that
\begin{equation}\label{eq:fnewbaseassump}
\begin{aligned}
f \in C^2(\mathbb{R}),\quad f &\in C^5((-\infty,0)) \cap C^5((0,\infty)), \\
f(0) = 1,\quad f'(0)&=0,\quad \argmax f(x)=0,
\end{aligned}
\end{equation}
with \(0\) being the unique global maximum.

    \item Furthermore, assume that\footnote{
We assume continuity at \(0\) up to the second derivative.
Higher-order derivatives are assumed continuous only on
\(\mathbb{R}_\pm\); the corresponding one-sided limits at \(0\)
are allowed to differ.
} \begin{align}
        \forall k \in \{3,4,5\}, \: f^{\bp{k}}\text{ extends continuously  both on }(-\infty,0]\text{ and  on }[0,\infty).
\label{eq:fextendedbycontinuity}
    \end{align}
\end{enumerate}
We now state the following result showing that, even under these relaxed conditions, the scaling of the inverse temperature required to ensure inter-modal mixing remains unchanged.

\begin{theorem}[Updated dimensionality-scaling for the Coldest Temperature Level]\label{Thm:scalingRelaxed}
Assume that the Mode Leap Point Independence sampler with proposal distribution given in \eqref{eq:gaussianPropapoproxtheory} is used to target the $(d+1)$-dimensional target distribution specified by \eqref{eq:thmtargmod2}, where the marginal iid components $f$ satisfy \eqref{eq:fnewbaseassump}, \eqref{eq:fextendedbycontinuity} \eqref{eq:posdefmod}, \eqref{eq:bddfourth}, \eqref{eq:polyT} and \eqref{eq:BacMor}.

If \begin{equation}
	\beta_{\text{max}} = \ell d  .  \label{eq:scalingnec_non_smooth}
\end{equation} for some $\ell \in \mathbb{R}_+$ then
\begin{align}
    \lim_{d \to \infty} a(d) = 2 \Phi\of{ - \frac{1}{\sqrt{2}} \sqrt{\frac{15 \bp{h^{\prime \prime \prime}\of{0_+}^2 + h^{\prime \prime \prime}\of{0_-}^2} - \frac{4}{\pi} \bp{h^{\prime \prime \prime}\of{0_+}- h^{\prime \prime \prime}\of{0_-}}^2}{72 \ell \bp{-h^{\prime \prime}\of{0}}^3} }} > 0 \nonumber \label{eq:ultimate}
\end{align}

where $h(x)=\log f(x)$ and $\Phi$ is the CDF of a standard Gaussian.
\end{theorem}

\begin{proof}
    See Appendix~\ref{proofRelaxed}.
\end{proof}
Theorem~\ref{Thm:scalingRelaxed} establishes that discontinuities in the third derivative of the target distribution at its modes do not prevent the Mode Leap Point Independence Sampler from achieving a positive acceptance rate in the asymptotic regime, provided $\beta_{\text{max}}$ scales as $\mathcal{O}(d)$.
In this piecewise-smooth setting, the asymptotic acceptance probability decomposes into a symmetric contribution, corresponding to the overall magnitude of local skewness on either side of the mode, and an asymmetric contribution induced by the jump in the third derivative, which captures the mismatch between the left and right local geometries.
In particular, when the third derivative is continuous at the mode, the skewness jump vanishes and the result reduces to the smooth-case acceptance of Theorem~\ref{Thm:scaling}.

Such piecewise-smooth behaviour arises in Bayesian settings where the log-posterior is constructed from components that are smooth on each side of the mode but differ in higher-order behaviour, for example due to asymmetric or piecewise-defined priors, truncated or constrained parameter spaces.
For a more in-depth investigation of ALPS under less restrictive smoothness and regularity assumptions, readers are referred to \citet[Sections~3 and~4]{queniat2024alps}.

\section{Computational Complexity of ALPS} 
\label{sec:compcomplex}

In the forthcoming Section~\ref{sec:Empiricalexamples}, we empirically compare the performance of ALPS with versions of the PT and mode leaping approaches on some challenging multimodal examples. Comparisons of computational efficiency are made between the algorithm's performances and the results are highly complimentary to ALPS.

With the scaling results presented in Section~\ref{sec:theory} it is possible to theoretically gauge the scalability of ALPS as the dimensionality of the problem grows, at least in stylised problems where explicit calculation and limits are accessible. To be transparent, the problem of finding modes in the full generality of a $d$-dimensional multimodal distribution is well-known to scale exponentially with $d$. This is since modes can, relatively, have exponentially  small (in $d$) basins of attraction \citep[e.g.,][]{Fong2019}. However the aim of ALPS is to achieve a robust and efficient method that explores the target distribution highly effectively, conditional upon finding the mode points. So we might hope to improve on this.

Indeed, the question of how the PA component scales with $d$ has been addressed in the accompanying paper \citep{Roberts2020}. The approach in \citet{Roberts2020} is to analyse the inverse-temperature process which converges within certain stylised sequences of target density as $d\to \infty $ to a limiting skew-Brownian motion process. This is attained through appropriate time re-scaling of the process as the dimension grows in the setting where the modes are from an exponential power family. 
The resulting complexity can be read off from the time re-scaling required to reach an ergodic non-trivial limit.

We refer the reader to Section 5 and in particular Corollaries 2 and 3 of \citet{Roberts2020} for a detailed description of these results (though note the typo in Corollary 3 which should be clear on reading the previous paragraph). For our purposes here, we note that these results suggest that the full version of  ALPS has a mixing time that scales as $\mathcal{O}(d)$ and that if the QuanTA transformation-aided swap moves are not deployed (in favour of standard parallel tempering swap moves) then the mixing time scales as $\mathcal{O}(d[\log (d)]^2)$.

It is noteworthy that the  scalability of the full version of ALPS is identical in order to that of the random walk Metropolis algorithm in a uni-modal problem. As such, the inter-modal mixing is coming at no extra mixing time cost to now be able to achieve global exploration of the Markov chain. Of course additional cost is incurred computationally and this is now made explicit.

In this setting (and approximately outside it) the optimal temperature schedule is geometrically spaced \citep{Tawn2018a}. More precisely, for the exponential power family setting analysed in \citet{Tawn2018a}, the optimally spaced inverse-temperatures satisfy
\[
    \beta_i = \beta_{i-1}\left(1+\sqrt{c/d}\right),
\]
for some constant \(c>0\). Thus adjacent inverse-temperatures are separated on a relative scale of order \(d^{-1/2}\). Consequently, the number of temperature levels required to reach a cold state, whose inverse-temperature is \(\mathcal{O}(d)\) by Theorem~\ref{Thm:scaling}, using optimally spaced inverse-temperatures is \(\mathcal{O}(d^{1/2})\). Using the arguments from \cite{Roberts2020}, this implies that the mixing time of the inverse-temperature is \(\mathcal{O}(d)\).

The leading cost for updating each temperature level per iteration of the algorithm comes from a $d\times d$-dimensional matrix multiplication with a $d$-dimensional column vector in the computation of the mode allocation; this has cost $\mathcal{O}(d^2)$. 

Combining the cost per iteration with the mixing time cost results from \cite{Roberts2020} suggest that the scalability of ALPS with dimension, $d$ is given by:
\[ \mathcal{O}(d^3) .\]

This complexity is a substantial improvement  on that of standard simulated and parallel tempering, which are demonstrated in
\cite{woodard2009conditions} and \cite{woodard2009sufficient} 
to be no better than exponential in $d$. Of course, this cost for ALPS can only be achieved once the mode points are at least approximately known, and in full generality this is a hard problem in itself.

\section{Applications of ALPS} 
\label{sec:Empiricalexamples}

This section contains three empirical studies for ALPS. The first is a synthetic example involving a complex, twenty-dimensional target distribution featuring a mixture of highly-skewed components that have differing scales. Due to its synthetic nature we can compare the performance of ALPS against benchmark algorithms, since the number of modes and their locations are known. The second example demonstrates the practicality of ALPS on a real-world data application from an SUR model in economics \citep{Zellner1962}. It is only relatively recently that SUR models have  been identified as having issues with multimodality \citep{Drton2004, Sundberg2010}. We show that ALPS performs well on the bimodal example from \cite{Drton2004} as well as on a fifteen-dimensional financial dataset \citep{Grunfeld1958} that was not previously known to exhibit multimodality. 
Finally, the third example illustrates another application of ALPS on a real-world data problem, with a spectral density model for analytical chemistry \citep{Ritter1994,Moores2016raman}.
Full R source code for all applications is included in the web-based supporting materials.

\subsection{A Synthetic 13-Modal Benchmark Distribution}
\label{sec:synthetic}
\subsubsection{Example 1: The 20-Dimensional Case}
\label{sec:Ex1}

The target distribution is given by a mixture of thirteen well-separated, twenty-dimensional skew-normal distributions with heterogeneous scaling of the modes and homogeneous skew. Its density function corresponds to the convex combination
\begin{equation*}
    \pi(x) \propto \sum_{k=1}^{13}  \omega_k \prod_{j=1}^{20}  \phi\left( \frac{x_j - (\mu_{k})_j }{\sigma_k}\right)\Phi\left( \alpha\frac{x_j - (\mu_{k})_j }{\sigma_k}\right)
\end{equation*}
where the mixing weights $\omega_k > 0\ \forall k \in \{1,\ldots,13\}$, skewness parameter $\alpha =10$, and where $\phi(\cdot)$ and $\Phi(\cdot)$ denote the density and CDF of a standard normal distribution, respectively. The mode locations $\vec\mu_k$ are distributed in a similar manner to the 13-mode example in \cite{Tjelmeland2001}, except in a 20-dimensional hypercube instead of only 2 dimensions. This example has been designed to illustrate the problem with torpid mixing of parallel tempering  \citep{woodard2009conditions} due to the heterogeneity of the scales $\sigma_k$ of the mixture components.

We compare the performance of ALPS against the gold standard of optimised parallel tempering and also the version of ALPS which doesn't anneal, i.e.\ discovers modes with the hot state mode finder and uses an independence sampler based on the Laplace approximation to the modes. In this section we use the shorthand LAIS (Laplace approximation independence sampler) to refer to this third algorithm. Although the LAIS employs mode-jumping proposals as in \citet{Tjelmeland2001}, the latter algorithm performs two optimisation steps at each iteration to find the nearest modes, rather than running an Exploration Component as a parallel chain. This reduces the computational cost of LAIS considerably, without changing its convergence properties. The strong skewness of the mixture components will prove to be critical for the acceptance rate of LAIS, since the error in the Laplace approximation increases with dimension, as explained by Remark~\ref{rem:whyneedcold} in Section~\ref{sec:optCold}.

The basic configurations of the three algorithms, the different levels of tuning required, the sample output from the chains, and the comparison of elapsed run times are described as follows.

The LAIS only uses a single temperature level which is the target level. For ALPS, 7 temperature levels were used for the annealing part of the algorithm with a coldest inverse temperature $\beta_{\text{max}}= 4096$ and a cooling schedule given by powers in the vector  $(1  ,  4  , 16  , 64 , 256, 1024 , 4096)$. The maximum level of cooling was excessively cold as non-degenerate acceptances for mode jumping moves were achieved at warmer temperatures (in agreement with the results of Section~\ref{sec:theory}). However, because mixing across temperatures proves extremely quick at colder temperatures and evaluating the synthetic target is inexpensive, it is possible to use much lower temperatures to get higher acceptance rates for the mode-jumping moves. In this example the mode jumping acceptance rates were around 0.85.

In contrast, PT required 23 temperature levels which were geometrically spaced with common ratio of 0.6. This gave the suggested optimal temperature swap rates of 0.234 between consecutive temperature  level pairs along with a temperature that appeared hot enough to explore the entire state space. 
Note that the minimum power $\beta_{min} = 0.6^{22} \approx 1.3 \times 10^{-5}$ used for PT is slightly hotter than the inverse temperature used for mode-finding in ALPS, $\beta_{\text{hot}} = 2.5 \times 10^{-5}$, in order to make a fair comparison between the two algorithms.
Even so, the PT algorithm fails to find all 13 modes even after 500,000 iterations. This is due to the differences in scale $\sigma_k$ between the mixture components, as predicted by \cite{woodard2009sufficient}. Due to torpid mixing, it could take several million more iterations for PT to converge.

In all three algorithms random walk Metropolis was utilised for the within temperature moves; in the case of ALPS and LAIS the random walk Metropolis moves used pre-conditioned  covariance structure informed by the covariance matrix estimated at the local mode point by the Laplace approximation. In all cases the within temperature moves were tuned to have the suggested optimal acceptance rate of 0.234.

The most difficult aspect of tuning ALPS was with respect to the Exploration Component for mode-finding. This was also highlighted in the real-data example in the following section. As mentioned previously, finding ``narrow modes'' is never guaranteed in the finite run of the algorithm if their location is unknown \citep{Fong2019}. In both LAIS and ALPS where mode-finding is required, it was found that an inverse temperature of $\beta_{\text{hot}} = 0.000025$ worked very well and rapidly found the modes; typically within a few thousand iterations of the algorithm. 
However, in this example the number of modes was known and so the mode finding could be optimised for ALPS and LAIS; something not possible in a real data example where one may continue mode searching throughout the duration or even doing some pre-computation dedicated to mode searches (as was the case in the following real-data example in Section~\ref{sec:SURexample}).

Each of the three algorithms was run long enough to generate 500,000 samples at the target temperature (including the ``burn-in'' samples from the transient phase).
For the purpose of comparison, the chains were all initiated from the first mode location, $\mu_1$ in each instance.
The EC took less than 3,000 iterations to find all 13 modes. Since this chain only searches for modes at every fourth iteration, this corresponds to 12,000 iterations of ALPS or LAIS.
Based on the trace plots, we discarded the first 70,000 iterations as burn-in. However, note that ALPS is the only algorithm that has actually converged to the target distribution, $\pi(\cdot)$. This means that the burn-in required for LAIS or PT is at least 500,000 iterations and possibly much more. The elapsed run time for ALPS was 3.2 hours, while it took 2.5 hours for LAIS and 2.3 hours for PT. The run time of LAIS was dominated by the computational cost of the EC chain, even though this was a small proportion of the total number of iterations. This is due to the large number of likelihood evaluations required for each optimisation step. 
The difference between ALPS and LAIS is due to the multiple cold temperature levels in ALPS. Although PT has the fastest run time of the 3 algorithms for a fixed number of iterations, it should again be stressed that PT has not yet converged at that point.

\begin{figure}
  \includegraphics[width=0.55\linewidth]{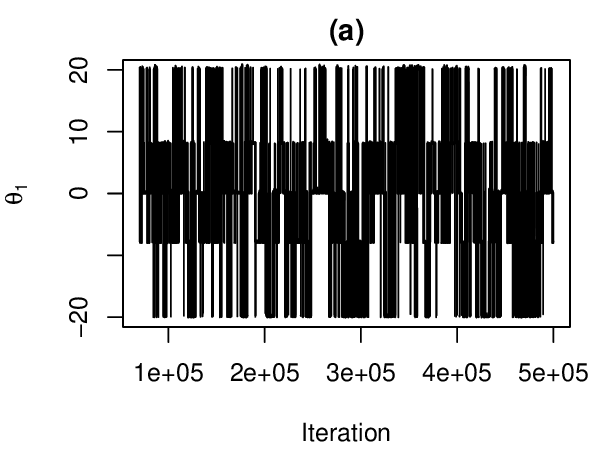}
\includegraphics[width=0.4\linewidth]{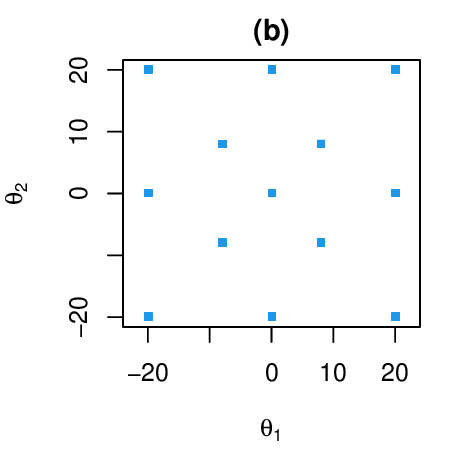}
  \includegraphics[width=0.55\linewidth]{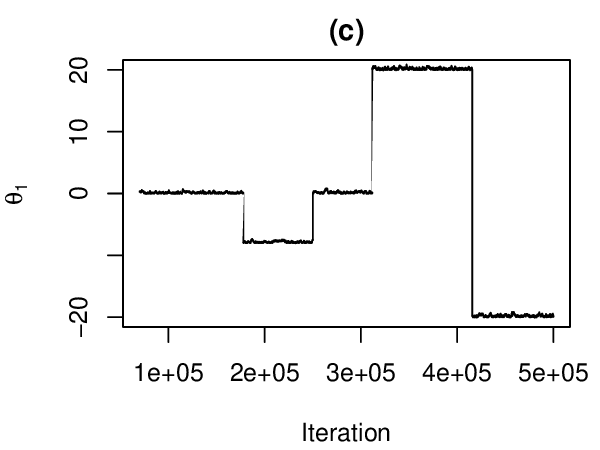}
\includegraphics[width=0.4\linewidth]{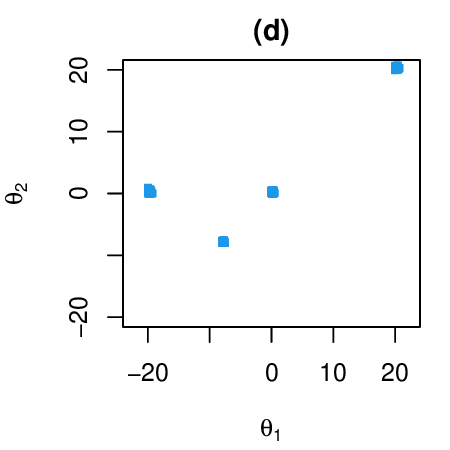}
  \includegraphics[width=0.55\linewidth]{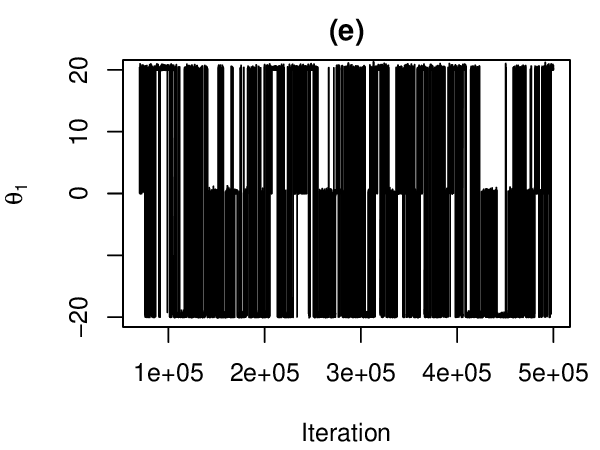}
\includegraphics[width=0.4\linewidth]{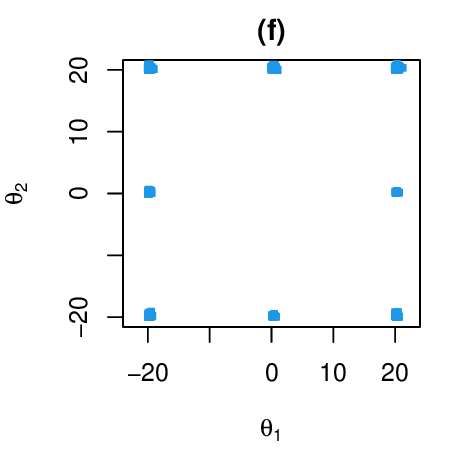}

  \caption{For the algorithms ALPS (top), LAIS (middle) and PT (bottom) respectively, the trace plots (left) of the first component of the Markov chain
  as well as marginal density plots (right) of the first 2 components at the target temperature, $\beta=1$. It can be seen from the trace plot (a) that ALPS mixes well and (b) shows that it visits all 13 modes. In contrast, (c) shows that LAIS very rarely jumps between modes and (d) only visits 4 of the modes. From (e) PT appears to be mixing well, but (f) shows that it has only visited 8 modes.
}
  \label{fig:TracePlotComp}
\end{figure}

Figure~\ref{fig:TracePlotComp} shows trace plots (on the left side) for the first component of the Markov chain at $\beta=1$ for each of the three algorithms after a burn-in is removed.
Corresponding marginal density plots for the first two components are shown alongside.
A successful outcome would be a plot where the trace line jumps between the modes of the marginal component, centred at values -20, -8, 0, 8 and 20, as shown in (a) for ALPS. Although $d=20$, the first two coordinates of the 13 modes are unique, thus the mode locations can be distinguished in the marginal plots shown on the right side of Figure~\ref{fig:TracePlotComp}.
This enables us to verify whether the chain is actually visiting all 13 modes.
For example, (b) shows  ALPS was the only algorithm to successfully jump between all of the mode points.

Even though LAIS could propose jumps between all 13 modes, the empirical acceptance rate at the target temperature $\beta=1$ is around $1.4 \times 10^{-5}$, which means only 5 proposals are accepted out of 350,000 attempts (one of these 5 jumps is during the burn-in phase). Due to the skewness of the mixture components, the error in the Laplace approximation is too large for LAIS to work effectively. This is evident in (c), while (d) shows that LAIS has only visited 4 of the modes. In contrast, the acceptance rate for mode-jumping proposals in ALPS was 0.85 at the coldest temperature level.

The trace plot for PT (e) appears to show that it is mixing well, but actually it is only visiting modes at -20, 0 or 20, not at -8 or 8. This demonstrates the danger of relying on one-dimensional trace plots to assess convergence in high-dimensional problems. In a real data example where the true modes were unknown, we might erroneously conclude that PT has successfully converged. We have artificially constructed this example so that we can see all of the modes in the two-dimensional marginal plot, (f). Here it is evident that PT has only visited the 8 modes with $\sigma_k = 0.3$ and has been unable to find the 5 smaller modes with $\sigma_k = 0.2$. Although the difference in scale might appear small, when $d=20$ the relative volume of the domain of attraction is approximately $0.2^{20}$ or $10^{-14}$ for the smaller components and $0.3^{20}$ or $3.5 \times 10^{-11}$ for the larger ones, a difference of several orders of magnitude.

In addition to the simulation study presented here, we have also included a performance comparison of ALPS, LAIS and PT across 10 repeated runs of each algorithm in Appendix~\ref{sec:Ex3}. This example has 4 modes instead of 13 and we run each algorithm for 200,000 iterations, discarding the first 20,000 as burn-in.

\subsubsection{Scaling with Dimension at the Cold Level}
\label{sec:scalingexample}

We now empirically validate the dimensional scaling result of Theorem~\ref{Thm:scaling} for the coldest temperature level using the synthetic multimodal example introduced above. The aim of this experiment is to assess whether the asymptotic acceptance rate derived in Theorem~\ref{Thm:scaling} is verified in high-dimensional settings.

Given a prescribed target acceptance rate $a \in (0,1)$ at the cold level, Theorem~\ref{Thm:scaling} implies that the linear scaling
\begin{align}
\beta_{\text{max}} = \ell(a)\, d \label{eq:optimal_scaling}
\end{align}
yields a non-degenerate asymptotic acceptance rate. Inverting the limiting expression in Theorem~\ref{Thm:scaling}, the constant $\ell(a)$ in \eqref{eq:optimal_scaling} is given explicitly by
\begin{align}
\ell(a)
=
\frac{5 h'''(0)^2}
{24 \, (-h''(0))^3 \,
\left[\Phi^{-1}(a/2)\right]^2 }, \label{eq:ell_optim}
\end{align}
where $h(x) = \log f(x)$ and $\Phi$ denotes the standard Gaussian distribution function. 
In this synthetic example, the function $f$ is fully specified, so that \eqref{eq:ell_optim} can be computed explicitly.

Using this expression, we compute $\beta_{\text{max}}$ for increasing dimensions and apply the Mode Leaping Independence Sampler to the same synthetic mixture model, keeping the number of modes fixed. In this experiment, the skewness parameter is set to $\alpha = 5.0$. For each dimension, we report the empirical acceptance rate of the mode-jumping proposals.

The results are shown in Figure~\ref{fig:cold_scaling}. As the dimension increases, the empirical acceptance rates approach the prescribed target value, providing empirical support for the linear scaling $\beta_{\text{max}} = \ell(a)\, d$, thereby verifying the asymptotic result stated in Theorem~\ref{Thm:scaling}.

\begin{figure}[t]
    \centering
    \includegraphics[width=\linewidth]{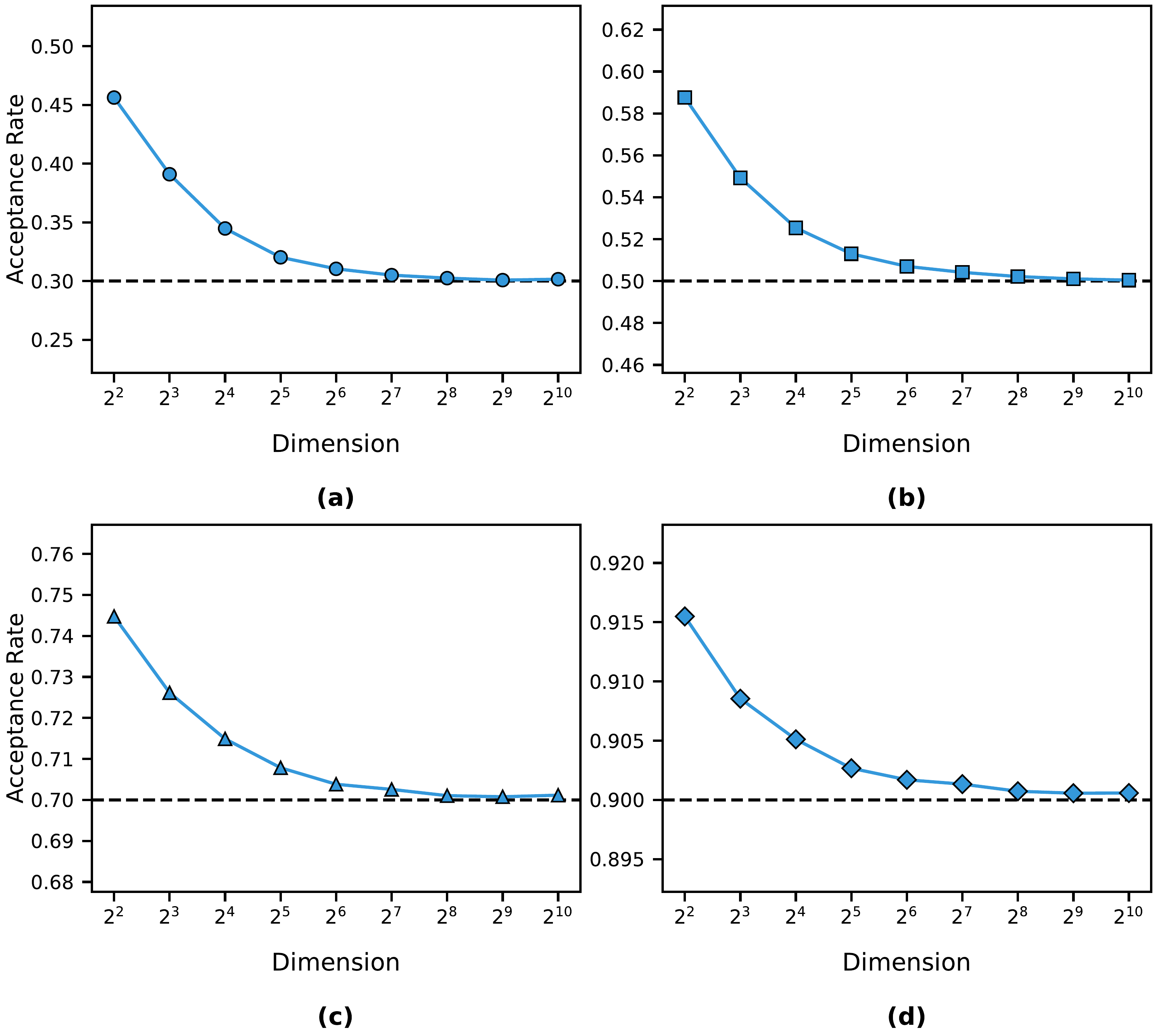}
    \caption{
    Empirical acceptance rates of the mode-jumping proposals at the coldest temperature level as a function of dimension, averaged over $128$ runs of $16384$ iterations.
    For each dimension $d$, the cold inverse temperature is set to $\beta_{\text{max}} = \ell(a)\, d$, where $\ell(a)$ is computed according to \eqref{eq:ell_optim}. Target acceptance rates $a$ are set to (a) 0.3, (b) 0.5, (c) 0.7 and (d) 0.9. }
    \label{fig:cold_scaling}
\end{figure}

\subsection{Example 2: Seemingly-Unrelated Regression Model}
\label{sec:SURexample}

Multi-response models \citep[ch. 14]{Searle2017} arise from datasets that have more than one dependent variable, where these variables are correlated with each other. An important example is Seemingly-Unrelated Regression (SUR or SURE), which was introduced by \cite{Zellner1962}. A wide variety of different SUR models are described by \cite{Srivastava1987} and the review of \cite{Fiebig2001}. These models involve a system of $M$ linear regression equations, one for each response, $m = 1, \dots, M$:
\begin{equation}
\vec{y}_m = \mathbf{X}_m \vec\theta_m + \vec\epsilon_m ,
\end{equation}
where $\vec{y}_m$ is a column vector of $N_m$ observations for the $m$th response, $\vec\theta_m$ is a vector of $J_m$ regression coefficients, $\mathbf{X}_m$ is an $N_m \times J_m$ design matrix of covariates, and $\vec\epsilon_m$ is a vector of residual errors. For notational convenience, we assume that $N_1 = N_2 = \dots = N_M = N$ and likewise that $J_1 = J_2 = \dots = J_M = J$, although in full generality
that is not always the case. See Appendix~\ref{appendix:SUR} for further introductory information on the SUR model.

\citet{Drton2004} provided a specific example of a dataset where the SUR model exhibits two local modes, which we will examine in more detail in the following section.

\subsubsection{Bimodal Example}

In the simplest, bivariate case with $M=2$ responses and only $J=1$ covariates per response, the model  can be written as:
$$\left[\begin{array}{c}
\vec{y}_1 \\ 
\vec{y}_2
\end{array}\right] = \left[\begin{array}{cc}
\vec{x}_1 & \vec{0} \\
\vec{0} & \vec{x}_2
\end{array}\right] 
\left[\begin{array}{c}
\theta_1 \\
\theta_2
\end{array}\right]  + \left[\begin{array}{c}
\vec\epsilon_1 \\
\vec\epsilon_2
\end{array}\right] .
$$
\citet{Drton2004} featured an example dataset with $N=8$ observations. The profile likelihood, derived in Appendix~\ref{appendix:SUR}, is illustrated in Figure~\ref{fig:profLikeSUR}. There are two modes at $\boldsymbol\theta =$ (0.78, 1.54) and (2.76, 2.50) and a saddle point at (1.62, 2.03). The iterative algorithm is initialised at the ordinary least squares estimate, $\hat{\boldsymbol\theta} =$ (1.25, 1.78). Using the R package `systemfit' \citep{Henningsen2007} with a tolerance level of $10^{-6}$, the algorithm of \citet{Zellner1962} converges to the first mode after 25 iterations. However, the profile log-likelihood at this mode is -27.72, whereas at the other mode it is -27.35. This algorithm has therefore failed to find the global maximum.

\begin{figure}[h]
  \centering
  \begin{tabular}[b]{c}
    \includegraphics[width=.4\textwidth]{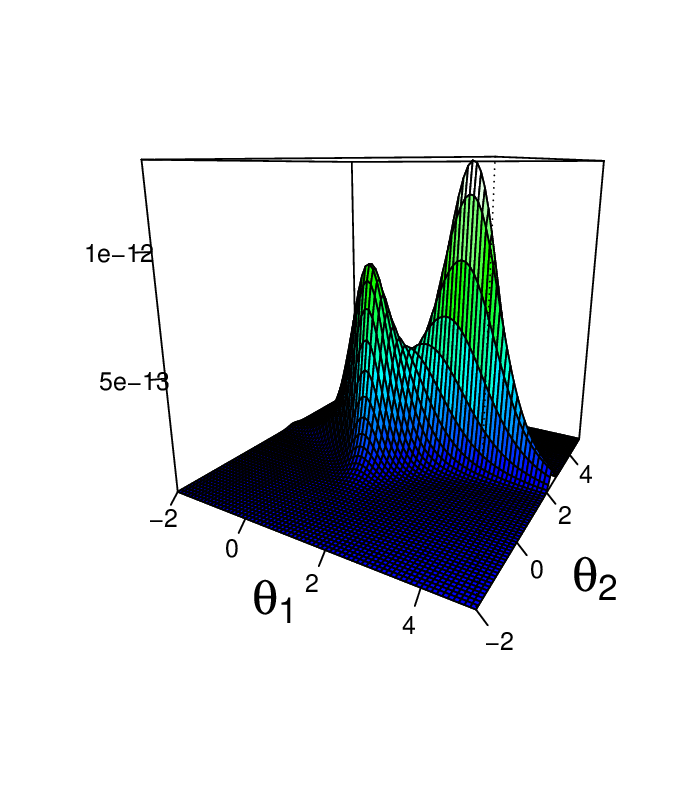} \\
    \small \textsf{(a)}
  \end{tabular} \qquad
  \begin{tabular}[b]{c}
    \includegraphics[width=.4\textwidth]{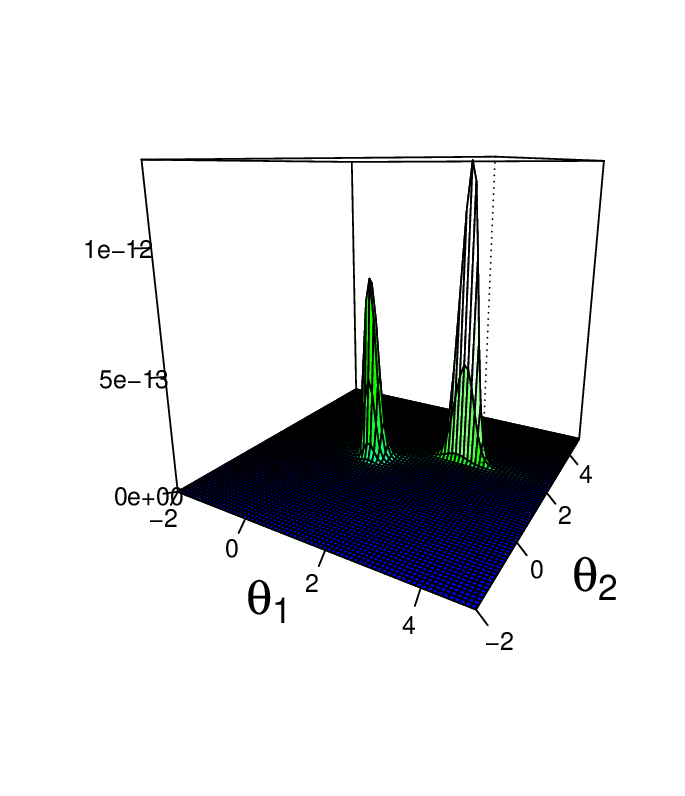} \\
    \small \textsf{(b)}
  \end{tabular}
  \caption{(a)~3D perspective plot of the profile likelihood over the covariance matrix for the bivariate SUR model, adapted from \citet[Fig. 1]{Drton2004}; (b)~Annealed likelihood at the coldest temperature level used in ALPS.}
  \label{fig:profLikeSUR}
\end{figure}

\begin{figure}
        \centering
  \begin{tabular}[b]{c}
                \includegraphics[width=0.4\linewidth]{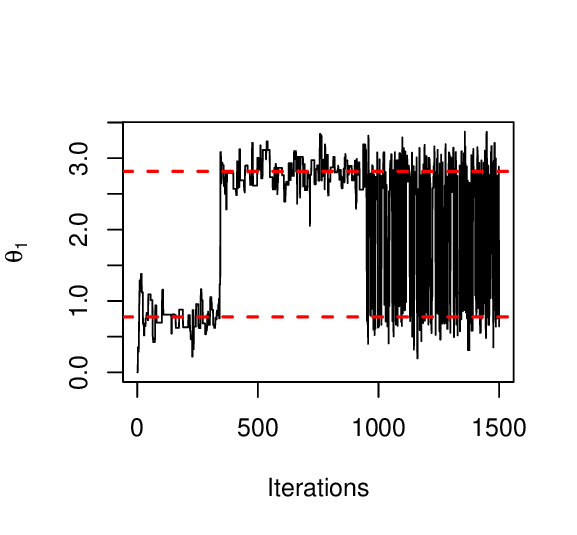} \\
    \small \textsf{(a) Traceplot for $\theta_1$.}      
  \end{tabular}%
\qquad
  \begin{tabular}[b]{c}
                \includegraphics[width=0.4\linewidth]{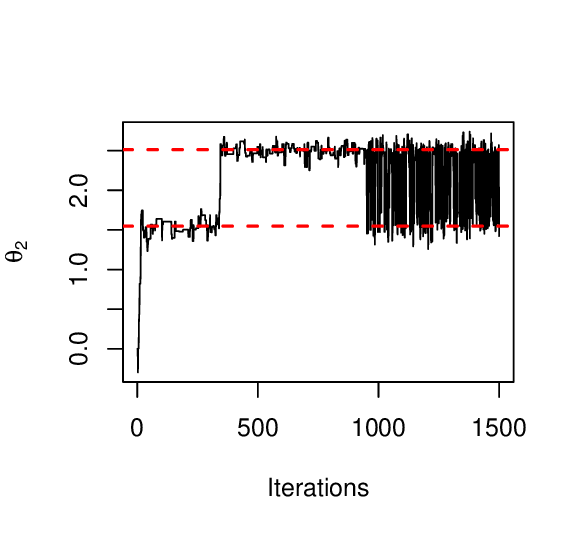} \\
    \small \textsf{(b) Traceplot for $\theta_2$.} 
  \end{tabular}%
\qquad
  \begin{tabular}[b]{c}
                \includegraphics[width=0.4\linewidth]{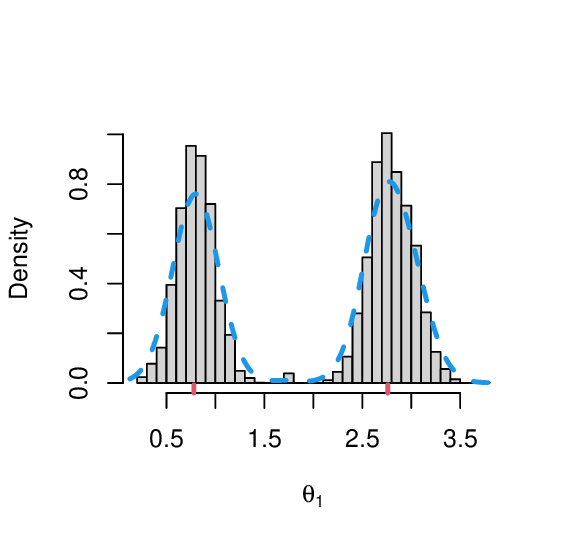} \\
    \small \textsf{(c) Histogram for $\theta_1$.} 
  \end{tabular}%
\qquad
  \begin{tabular}[b]{c}
                \includegraphics[width=0.4\linewidth]{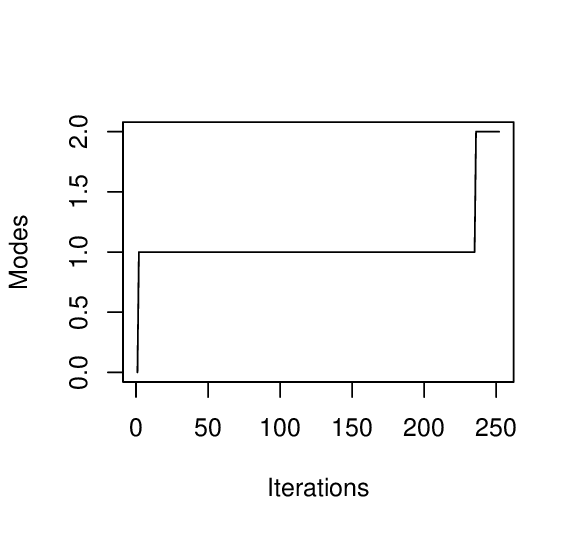} \\
    \small \textsf{(d) Exploration Component.}  
  \end{tabular}%
\caption{Results for ALPS at the target temperature with the bivariate SUR model from \cite{Drton2004}. We see that upon discovery of the two modes, the ALPS-PA successfully enables inter-modal exploration between the two modes, whose locations marginally are indicated by the red dashed lines. Note that the ratio of iterations between the hot-state mode finder and the other chains is 1:4, so 236 iterations of the Exploration Component correspond to 944 iterations of the within-temperature chains. This example highlights that before discovery of both the modes the within moves fail to explore the state space successfully.}
\label{f:ndvi}
\end{figure}
We run ALPS using the same profile log-likelihood function, with  $\beta_{\text{max}} = 10$ and a geometric cooling schedule composed of 
three inverse-temperature levels $(1, \sqrt{10}, 10)$ as well as the hot state 0.5. The Exploration Component discovers both modes after 236 iterations, as shown in Figure~\ref{f:ndvi}~(d), and ALPS takes  2.6 seconds for 10,000 iterations in total. The MCMC trace plots for the two parameters at the coldest temperature are shown in Figures \ref{f:ndvi}~(a) and \ref{f:ndvi}~(b), respectively, and a histogram for $\theta_1$ is shown in Figure \ref{f:ndvi}~(c). The jump rate for mode-hopping proposals was 0.9 at the coldest temperature, with acceptance rates of 0.067 for the EC, 0.265 for within-temperature moves at the intermediate temperature and 0.331 at the target temperature. The temperature-swap acceptance rates between consecutive levels were $(0.508,0.847)$.

\subsubsection{Investment Demand}
\cite{Zellner1962} illustrated his method using an investment equation with $J=3$ terms,
$$
\vec{y}_m = \vec{x}_{1,m} \theta_{1,m} + \vec{x}_{2,m} \theta_{2,m} + \vec{x}_{3,m} \theta_{3,m} +  \vec\epsilon_m ,
$$
where $y_{i,m}$ is the gross investment by firm $m$ during the $i$th year, $\vec{x}_{1,m}$ is a vector of ones (so that $\theta_{1,m}$ is a firm-specific intercept), $\vec{x}_{2,m}$ is the market value of the firm, and $\vec{x}_{3,m}$ is its capital stock. The dataset of U.S. manufacturing firms was originally published by \cite{Grunfeld1958} and has since received considerable attention in the econometrics literature, as reviewed by \cite{Kleiber2010}. We will consider the first $N=15$ years of data, from 1935 to 1949, for $M=5$ firms (General Motors, Chrysler, General Electric, Westinghouse, and US Steel), so the parameter space has 15 dimensions in total. These data are available in the \textsf{R} package \textsf{systemfit} \citep{Henningsen2007}.

The iterative algorithm of \cite{Zellner1962} took 52 iterations to converge to a local mode, with a sum of squared residuals (SSR) of 216,943  and profile log-likelihood of -263.7. However, \citet{Drton2004} suggests that this problem exhibits multimodality and so we attempt to use ALPS to sample from the posterior distribution.

Using the EC, three modes were discovered by the EC chain as shown in Figure~\ref{f:surGreene} (d), although the estimated weight suggests the third mode has negligible probability mass. The first mode discovered by the EC part of ALPS matches the estimate from \textsf{systemfit}, while the other two modes have profile log-likelihoods of -264.9 and -329.1, respectively. 

This problem is particularly ill-posed due to parameter unidentifiability (which leads to the multimodality in the posterior). In this problem, despite attempts to reparameterize, there remain issues with the modes of the posterior distribution being located on long, thin ridges that seem to decay with significantly heavier tails than Gaussian. This means that to obtain a good Laplace approximation to the mode the ALPS procedure would need to use very large inverse-temperature values. However, when this was attempted it led to severe computational issues with regards to machine precision, likely due to the condition numbers of the Hessians at the mode points, illustrating the ill-posed nature of the problem.  To overcome this a truncated version of the annealed HAT targets was utilised for annealed temperatures with $\beta>1$. These are a small modification to the HAT targets and they are defined and discussed in detail in Appendix~\ref{sec:Truncated}.

We ran ALPS with four inverse temperature-levels $(1, 4, 16, 32)$. The MCMC trace plot for the ninth parameter (General Electric capital stock) at the coldest temperature is shown in Figure~\ref{f:surGreene} (b). Traceplots for the marginal chains 15 parameters at the target temperature $\beta=1$ and the coldest inverse temperature $\beta_{max}= 32$ are provided in Appendix~\ref{appendix:Grunfeld}. The jump rate for mode-hopping proposals was 0.517. The temperature-swap acceptance rates between consecutive levels were $(0.264,0.560, 0.681)$.
For robustness of within-temperature moves, we used a mixture of pre-conditioned and standard random walk moves with acceptance probabilities between 0.37 and 0.42. The ALPS-PA component generating 200,000 samples took under 40 seconds.

ALPS is not ideally suited to this example. The modes are not especially well-separated and the modes lie on narrow ridges which lead to computational instability. Additional Hessian regularization, for example using the SoftAbs map of \citet{Betancourt2013}, could improve numerical stability in this setting. Nevertheless, the fact that the algorithm still draws samples from multiple modes suggests that the approach remains robust in settings where reliable local Laplace approximations are challenging to obtain.
\begin{figure}\centering
  \begin{tabular}[b]{c}
\includegraphics[width=0.4\linewidth]{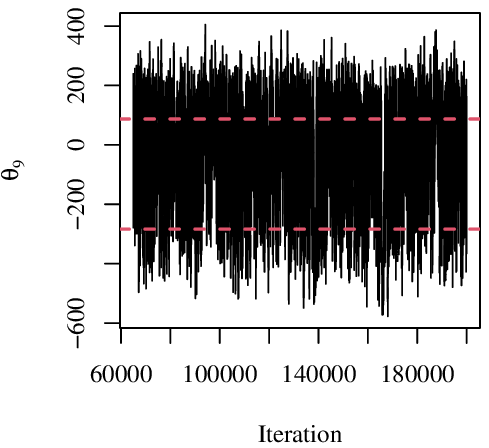} \\
    \small \textsf{(a) Traceplot for $\theta_9$ at target temp.} 
  \end{tabular}%
\qquad
  \begin{tabular}[b]{c}
\includegraphics[width=0.4\linewidth]{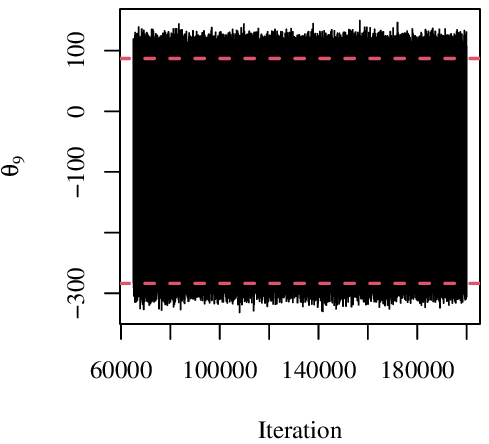} \\
    \small \textsf{(b) Traceplot for $\theta_9$ at coldest temp.} 
  \end{tabular}%
\qquad
  \begin{tabular}[b]{c}
                \includegraphics[width=0.4\linewidth]{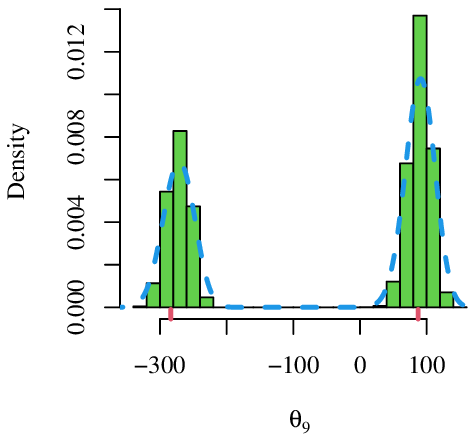} \\
    \small \textsf{(c) Histogram for $\theta_9$ at coldest temp.} 
  \end{tabular}%
\qquad
  \begin{tabular}[b]{c}
\includegraphics[width=0.4\linewidth]{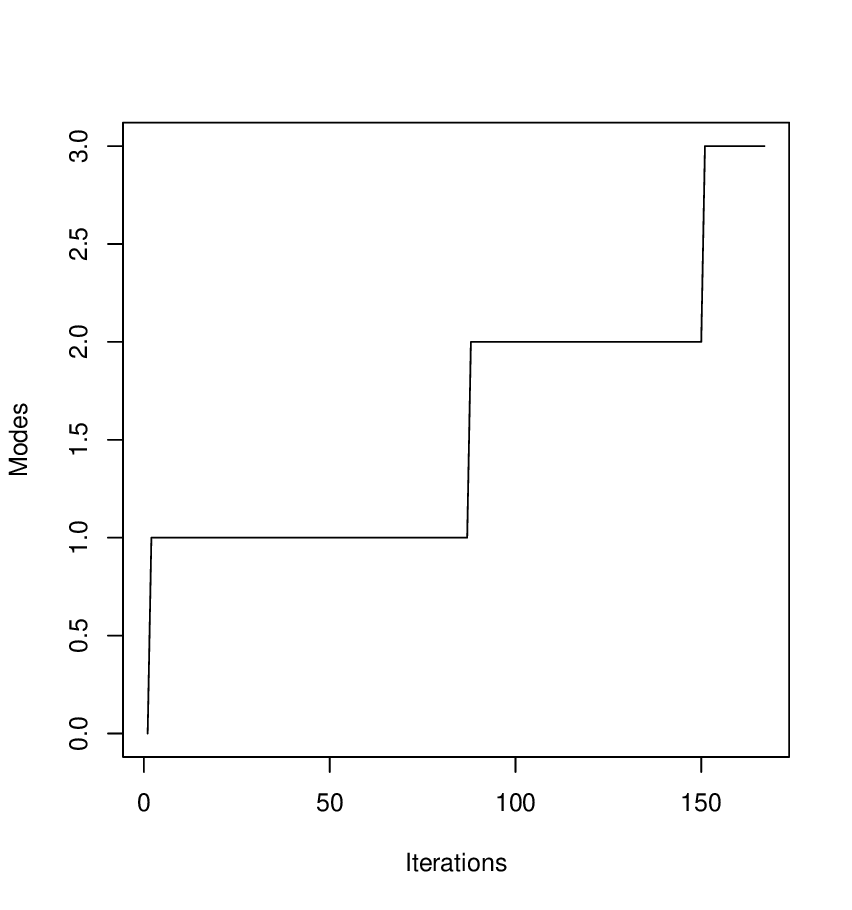} \\
    \small \textsf{(d) Exploration Component.} 
  \end{tabular}%
\caption{Results for ALPS with the SUR dataset from \cite{Grunfeld1958} with $N=15$ years of observations for $M=5$ manufacturing firms, with $J=3$ covariates per firm. The red, dashed lines show the locations of two of the modes (the third mode has negligible probability mass). The EC takes a little over 150 iterations to find the modes as shown in (d). (b)-(c) show that in the one-dimensional marginal ALPS is making successful intermodal jumps at the coldest inverse temperature $\beta_{max}= 32$, while (a) confirms that this mixing information is successfully transferred to the target level.}
\label{f:surGreene}
\end{figure}
\subsection{Example 3: Spectral Density Model}
\label{sec:RamanExample}

\citet{Ritter1994} was the first to introduce a Bayesian approach for peak decomposition of electron spectroscopy. This has become a popular method in analytical chemistry for many different applications, including X-ray and $\gamma$-ray spectroscopy \citep{van2001analysis} and mass spectrometry \citep{wang2008reversible}. In this example, we will use ALPS to fit a peak decomposition model for Raman spectroscopy \citep{Moores2016raman}. The model represents spectroscopic data $\vec{y} \in \mathbb{R}^n$ as a conical combination of $M$ basis functions:
\begin{equation}\label{eq:spectrum}
\vec{y} = \sum_{m=1}^M \alpha_m f_m(\vec{x};\, \ell_m, \psi_m)  + \vec\epsilon ,
\end{equation}
where the observed spectrum has been discretised at wavenumbers $x_1,\dots,x_n \in \mathcal{X} \subset \mathbb{R}$. In this example $\mathcal{X}$ is the interval from 100 to 600 cm$^{-1}$ and the $n=101$ wavenumbers are evenly spaced, 5 cm$^{-1}$ apart. We assume additive {\em i.i.d.} Gaussian noise, $\epsilon_i \sim \mathcal{N}(0, \sigma^2)$. Each of the peaks is described by three parameters: its amplitude $\alpha_m > 0$, location $\ell_m \in \mathcal{X}$, and scale $\psi_m > 0$. The dimension of the parameter space is therefore $3M$. We assume that each spectral density function $f_m(\cdot)$ is either Lorentzian,
\begin{equation}\label{eq:Lorentzian}
f_L(x;\, \ell_m, \psi_m) \propto \frac{\psi_m^2}{(x - \ell_m)^2 + \psi_m^2} ,
\end{equation}
or Gaussian,
\begin{equation}\label{eq:Gaussian}
f_G(x;\, \ell_m, \psi_m) \propto \exp\left\{ - \frac{(x - \ell_m)^2}{2 \psi_m^2} \right\} .
\end{equation}
Our choice of basis functions $f_L(\cdot)$ and $f_G(\cdot)$ is motivated by the physical data-generating process of Raman spectroscopy, as described by \citet{Diem2015} and \citet{Moores2016raman}. 

For this example, $\vec{y}$ is a Raman spectrum of orthoclase feldspar (KAlSi$_3$O$_8$) shown in Fig.~\ref{f:Orthoclase}~(a). There are $M=4$ peaks in this spectrum, two Gaussian and two Lorentzian. The peak locations $\ell_1,\dots,\ell_4$ are slightly different to those reported by \citet{freeman2008characterization} and therefore need to be inferred from the data. The scale and amplitude parameters of these peaks are also unknown. The other peaks reported by \citet{freeman2008characterization} are below the limit of detection in this sample.

\begin{figure}\centering
  \begin{tabular}[b]{c}
\includegraphics[width=0.42\linewidth]{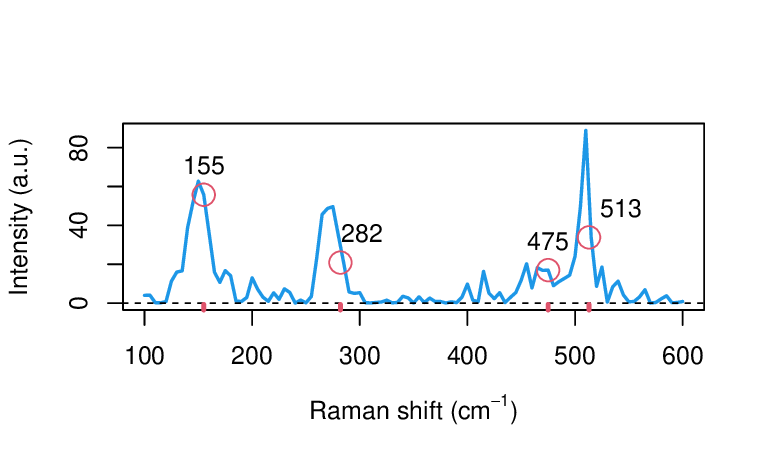} \\
    \parbox[t][2\baselineskip][c]{0.4\linewidth}{\centering\small (a) Raman spectrum of orthoclase feldspar (\textsf{KAlSi$_3$O$
          +_8$}).}   
  \end{tabular}%
\qquad
  \begin{tabular}[b]{c}
  \includegraphics[width=0.42\linewidth]{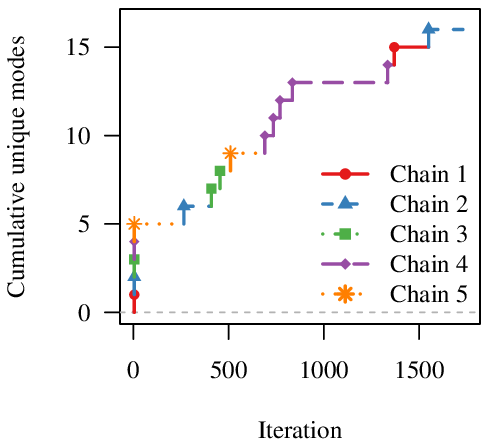} \\
    \parbox[t][2\baselineskip][c]{0.4\linewidth}{\centering\small (b) Exploration component (ALPS-EC).}
  \end{tabular}%
\qquad
  \begin{tabular}[b]{c}
\includegraphics[width=0.42\linewidth]{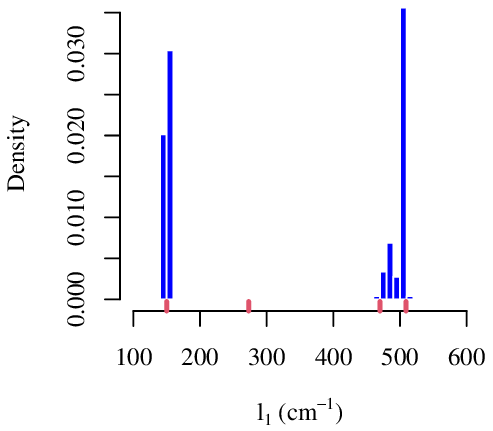} \\
    \parbox{0.4\linewidth}{\centering\small (c) Peak location $\ell_1$.}
  \end{tabular}%
\qquad
  \begin{tabular}[b]{c}
\includegraphics[width=0.42\linewidth]{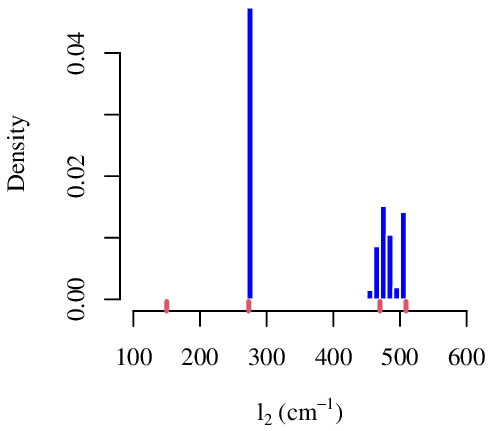} \\
    \parbox{0.4\linewidth}{\centering\small (d) Peak location $\ell_2$.}
  \end{tabular}%
\qquad
  \begin{tabular}[b]{c}
\includegraphics[width=0.42\linewidth]{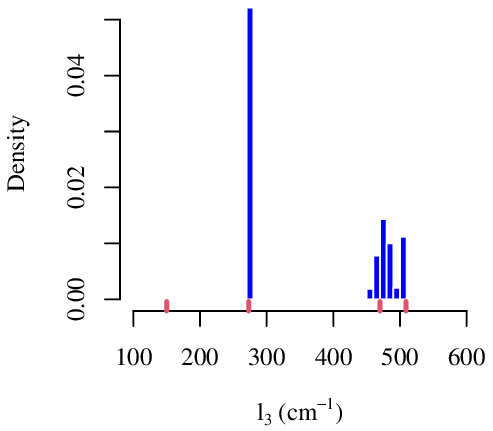} \\
    \parbox{0.4\linewidth}{\centering\small (e) Peak location $\ell_3$.}
  \end{tabular}%
\qquad
  \begin{tabular}[b]{c}
\includegraphics[width=0.42\linewidth]{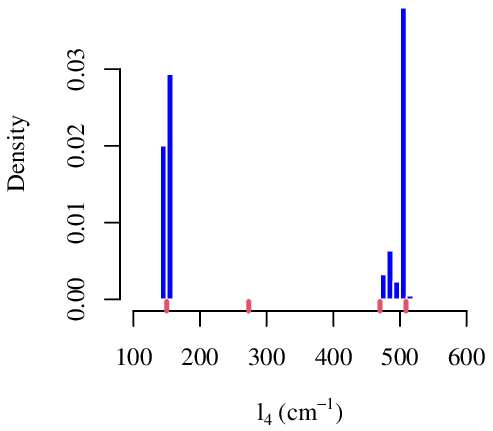} \\
    \parbox{0.4\linewidth}{\centering\small (f) Peak location $\ell_4$.}
  \end{tabular}%

\vspace{0.5\baselineskip}
\caption{Raman spectrum of orthoclase (a) showing the 4 peak locations reported by \citet{freeman2008characterization}; (b) the 16 modes found by the parallel ALPS-EC; and (c)-(f) histograms of the ALPS-PA output for the peak location parameters at the target temperature $\beta=1$.}\label{f:Orthoclase}
\end{figure}

To fit the model in Eq.~\eqref{eq:spectrum} with ALPS, we transform the parameters to the unconstrained space $\mathbb{R}^d$ using bijective functions: $a_m = \log\{\alpha_m\}$; $s_m = \log\{\psi_m\}$; and 
\begin{equation}
    l_m = \text{logit}\left\{ \frac{\ell_m - x_1}{x_n - x_1}\right\},
\end{equation}
where in this example the minimum wavenumber $x_1 = 100$ cm$^{-1}$ and the maximum $x_n = 600$ cm$^{-1}$. This results in a parameter vector $\vec\theta = (l_1, s_1, a_1, \dots, l_M, s_M, a_M)^\top \in \mathbb{R}^{12}$.

\begin{table}[ht]
\centering
\begin{tabular}{rrrrrrrr}
  \hline
 & Weight & $\log(\mathcal{L})$ & $\log(|\Sigma|)$ & $\ell_1$ & $\ell_2$ & $\ell_3$ & $\ell_4$ \\ 
  \hline
1 & 0.319 & -285.2 & -74.54 & 150 & 272 & 477 & 509 \\ 
  2 & 0.120 & -285.2 & -76.48 & 150 & 272 & 509 & 484 \\ 
  3 & 0.001 & -290.0 & -76.48 & 150 & 486 & 509 & 272 \\ 
  4 & 0.001 & -290.0 & -76.48 & 150 & 509 & 486 & 272 \\ 
  5 & 0.319 & -285.2 & -74.54 & 150 & 477 & 272 & 509 \\ 
  6 & 0.120 & -285.2 & -76.48 & 484 & 272 & 509 & 150 \\ 
  7 & 0.120 & -285.2 & -76.48 & 150 & 509 & 272 & 484 \\ 
   \hline
\end{tabular}
\caption{Modes found by ALPS-EC for the spectral density model. Weight is a Laplace approximation to the relative weight of each mode, which involves the log-likelihood $\log(\mathcal{L})$ and the log-determinant $\log(|\Sigma|))$. $\ell_1,\dots,\ell_4$ are the location parameters of the four peaks.} 
\label{table:spcModes}
\end{table}

 We ran 8 parallel independent chains of ALPS-EC, using the \textsf{R} package \textsf{future} \citep{Bengtsson2021guture} to run each chain in its own concurrent thread. These chains took 1550 iterations to find 16 unique modes, as shown in Fig.~\ref{f:Orthoclase}~(b).  Running on a computer with 8 CPU cores, it took less than 110 seconds.

 However, only 7 of these modes appear to have non-negligible mass and these are presented in Table~\ref{table:spcModes} along with their estimated wights via Laplace approximation. 

In a typical label-switching regime such as a mixture model \citep{jasra2005markov} we would expect the likelihood to be invariant to permutation of the labels, but for this model it is not necessarily the case. That is because of the different spectral density functions $f_L(\cdot)$ \eqref{eq:Lorentzian} and $f_G(\cdot)$ \eqref{eq:Gaussian}. For example, $\ell_1$ and $\ell_4$ are both Lorentzian peaks, so those peak locations can be swapped without affecting the likelihood, as in modes 6 and 7. However, if $\ell_3$ and $\ell_4$ are swapped, as in modes 1 and 2, that does change the relative weight. 

 To aid mode finding in this setting one could further exploit permutation proposals \citep{fruhwirth2001markov} to jump between modes, but this would be problem-specific: the user would need to encode the symmetric structure of the parameter space, which might not be known in advance. 

ALPS-PA took 3.6 minutes for 80,000 iterations with six inverse-temperature levels $(1, 1.6, 2.56, 4.096, 6.55, 10.45)$. The jump rate for mode-hopping proposals was 0.885 at the coldest inverse temperature $\beta_{\max} = 10.45$. For robustness, we used a mixture of pre-conditioned and standard random walk moves with acceptance probabilities between 0.22 and 0.36. The temperature-swapping acceptance rates were 0.412 between the first two temperatures and 0.969 between the coldest temperatures. 

Marginal histograms of the  location parameters $\ell_1,\dots,\ell_4$ at the target temperature are shown in Figures~\ref{f:Orthoclase}(c)--(f). The long-range dependence in the peak at 150 cm$^{-1}$, characteristic of Lorentzian line shape, is evident from the high probability of $\ell_1$ and $\ell_4$. Similarly, the Gaussian line shape of the peak at 273 cm$^{-1}$ is clearly identified by the model. These first two peaks are well-separated from the others, whereas the peaks at 470 cm$^{-1}$ and 509 cm$^{-1}$ are overlapping with each other. In analytical chemistry, the peak at 470 cm$^{-1}$ is known as a shoulder. Reliable inference for peaks like this, which are common in Raman spectroscopy, is a difficult problem. All four location parameters assign some probability to this peak, since the data are too noisy to clearly identify whether it is Gaussian or Lorentzian. There is also some ambiguity in the line shape of the peak at 509 cm$^{-1}$, but the model assigns higher probability for Lorentzian than Gaussian.

\section{Conclusion and Further Work} 
\label{sec:conclusions}

This paper introduces a novel algorithm, ALPS, that is designed to sample effectively and scalably from multimodal distributions that have support on $\mathbb{R}^d$ and are $C^2$ smooth. The main methodological contribution of this paper is that this problem can be solved by (localised) annealing rather than tempering and that this new approach mitigates many of the issues due to mode heterogeneity that plague traditional methods for these problems. Also, ALPS exploits recent advances in the literature that both enable the algorithm to be stable in high dimensions as well as providing substantially accelerated performance.

The accompanying theoretical contribution of the paper shows that the coldest temperature required for the procedure will be $\mathcal{O}(d)$ under suitable regularity conditions. This gives both guidance on tuning the algorithm and demonstrates polynomial computational complexity cost of $\mathcal{O}(d^3)$.

We have applied ALPS to two real-world examples of multimodal distributions, a Seemingly-Unrelated Regression (SUR) model of investment demand \citep{Zellner1962} and a peak decomposition model of Raman spectroscopy \citep{Moores2016raman}. There are many more potential applications for this algorithm to sample from continuous distributions where multimodality is known or suspected to be present, such as dynamical models of biological systems \citep{Ballnus2017benchmarking}, differential equation models \citep{obrien2025dragons}, and other curved exponential families \citep{Sundberg2010}.

There are a number of further considerations regarding adding robustness and practicality to the ALPS procedure. Many of the ideas for exploration are already alluded to in the paper: improving the effectiveness of the EC part by running more than one exploration chain and having multiple temperature levels; parallel implementation of the PA and EC parts of ALPS; and introducing weightings to regions to stabilise the inter-modal masses when the HAT targets don't sufficiently preserve the weights. Another possibility would be to incorporate the cold state level into a simulated tempering framework, interpolating between the cold Laplace approximation and the target distribution, and serving, per se, as a prior or reference distribution as in \citet{surjanovic2023paralleltemperingvariationalreference}. This would allow perfect regeneration of the Markov chain at each visit to the cold level thanks to the i.i.d. sampling of the Laplace approximation. Adapting the ideas of ALPS to facilitate an automatic and scalable reversible jump MCMC algorithm is the immediate focus for further work.\\

\appendix

\section*{Acknowledgements}
We thank Dr.~Sigurd Assing (University of Warwick) and Fernando Zepeda Herrera (University of Warwick) for their valuable contributions and the editor and
referees for their insightful feedback. Gareth O.~Roberts has been supported by the UKRI grant EP/Y014650/1 as part of the ERC Synergy project OCEAN, EPSRC grants Bayes for Health (R018561), CoSInES (R034710), PINCODE (EP/X028119/1), ProbAI (EP/Y028783/1) and (EP/V009478/1).

\setcounter{equation}{28}
\appendix

\section{Proofs}\label{sec:appendix}
\subsection{Traditional set of assumptions}
\label{proof:regular}
For all proofs that follow, we shall, for notational convenience, denote $\beta_{\text{max}}$ simply by $\beta$ as that is the only temperature considered in this analysis.

The added complexity of the setup with the introduction of an augmented allocation component given in \eqref{eq:thmtargmod2} allows for significant simplification in the theoretical analysis.

Suppose that $(k,x) \sim \pi_\beta(\cdot,\cdot)$.  Suppose that the independence sampler move based on \eqref{eq:gaussianPropapoproxtheory} proposes the point $(u,y) \sim q_\beta (\cdot, \cdot)$. Then the acceptance ratio  after some basic cancellation of terms would be given by 
\begin{equation}
    \frac{\pi_\beta(u,y)q_\beta(k,x)}{\pi_\beta(k,x)q_\beta(u,y)}=\prod_{i=1}^d \frac{\left[f\left(\frac{y_i -\mu_{iu}}{\sigma_{iu}}\right)\right]^\beta \left[\phi_\beta \left(\frac{x_i -\mu_{ik}}{\sigma_{ik}}\right)\right]}{ \left[f\left(\frac{x_i -\mu_{ik}}{\sigma_{ik}}\right)\right]^\beta \left[\phi_\beta \left(\frac{y_i -\mu_{iu}}{\sigma_{iu}}\right)\right]}.
\end{equation}

Due to the data augmentation component, wlog for the evaluation of the acceptance ratio we may assume that $\mu_{ij}=0$ for all $i \in \{1,\ldots,d\}$ and $j\in \{1,\ldots, m\}$, i.e. the modes in all components are centred at $0 \in \mathbb{R}^d$.

Furthermore, by a similar argument, one can see from simple calculations that the acceptance ratio is unchanged by a rescaling of the regions conditional on the components provided the components of the independence sampler proposal undergo identical rescaling. As such we can assume wlog that $\sigma_{ij}=1$ for all $i \in \{1,\ldots,d\}$ and $j\in \{1,\ldots, m\}$.

This dramatically simplifies the problem. In fact, wlog the problem reduces to studying the behaviour of an independence sampler move on just a $d$-dimensional target distribution on $\mathbb{R}^d$ (without the component allocation) given by
\begin{equation}
    \pi_\beta(x) \propto \prod_{i=1}^d f^\beta(x_i) \label{eq:simplified_version_target}
\end{equation}
with independence sampler proposals generated by 
\begin{equation}
q_\beta(x)=\prod_{i=1}^d \phi_\beta(x_i). \label{eq:simplified_version_proposal}
\end{equation}

To this end, suppose that the cold state component of the ALPS Markov chain is at a position $x$ and that a Gaussian independence sampler proposal, $y$, is made as in the ALPS scheme; then the probability of acceptance for the proposal is given by
\begin{equation}
	A_{IS}(x,y) = \min \left\{ 1,  e^{B(d)}  \right\} \label{eq:probability_of_acceptance_independent_sampler}
\end{equation}
where
\begin{equation}
	e^{B(d)} = \prod_{i=1}^d \frac{f^{\beta}(y_i)\phi_\beta(x_i)}{\phi_\beta(y_i)f^{\beta}(x_i)} \label{eq:exponential_acceptance_ratio}
\end{equation}
where $x_i \sim f^\beta$, $y_i \sim \phi_\beta$ with $\phi_\beta$ denoting the pdf of a Gaussian of the form $N\left( 0, \left(\beta|H|\right)^{-1} \right)$ which is the Gaussian approximation that the ALPS algorithm implements about the mode in this case.\\
The proof of Theorem~\ref{Thm:scaling}, deferred to this appendix, relies on an asymptotic analysis of \ref{eq:probability_of_acceptance_independent_sampler} in the high-dimensional regime, as formalized in \ref{eq:asymptotic_acceptance_independence_sampler}. Prior to the proof, we establish a number of essential preliminary results.
\subsubsection{Auxiliary Results}
For notational convenience herein denote
\[
B_\epsilon (L) := \{x \in \mathbb{R}: |x-L|<\epsilon   \}
\]
and formally define the normalisation constant.
\begin{definition}[Normalisation at level $\beta$] \label{def:nbet}
For the un-normalised marginal component of the product target specified in Theorem~\ref{Thm:scaling}, i.e.\ $f(\cdot)$, then define the normalisation constant
\begin{equation}
	N(\beta) :=  \int [f(x) ]^ \beta dx \label{eq:Nbeta}
\end{equation}
\end{definition}
\begin{definition}[Asymptotic Gaussian standard deviation]
With $f(\cdot)$ as in Theorem~\ref{Thm:scaling} with $H=h''(0)$, define
\begin{equation}
	\sigma(\beta) = \frac{1}{\sqrt{\beta |H|}}. \label{def:sigbet}
\end{equation}
\end{definition}
The following few propositions work on integrals that use the un-normalised density. It will be clear when the normalisation constant comes in. In fact, it will be shown that $N(\beta)= \mathcal{O}(\beta^{1/2})$.
%
\begin{proposition}[Exponential decay of tail integrals] \label{prop:Tailexp}
With $f(\cdot)$ as in Theorem~\ref{Thm:scaling} and specifically for the polynomial tail condition given in \eqref{eq:polyT} then for each $k \geq 0$ there exist constants $C_1>0~\text{and}~C_2 \ge 0$ such that
\begin{equation}
	\limsup_{\beta \rightarrow \infty} \left|e^{C_1\beta}\int_{\mathbb{R}\backslash B_M(0)} z^k  f^\beta (z) dz \right| = C_2 <\infty \nonumber
\end{equation}
\end{proposition}
\begin{proof}
With $K, \gamma ~\text{and}~ M$ all specified in \eqref{eq:polyT} (and wlog $K<M$) and assuming $\beta$ sufficiently large
\begin{eqnarray*}
\left| \int_M^{\infty} z^k  f^\beta (z) dz  \right| &\le& \int_M^{\infty} |z|^k  |z|^{-\gamma\beta}\frac{f^\beta (z)}{|z|^{-\gamma\beta}} dz  \\
&\le&  K^\beta \int_M^{\infty} |z|^{k-\gamma\beta}   dz  \\
&\le& \bigg(\frac{K}{M^\gamma}\bigg)^\beta \frac{M^{k+1}}{\left(\gamma\beta-(k+1)\right)} \\
&=&  \frac{M^{k+1}}{\left(\gamma\beta-(k+1)\right)} e^{-\beta |\log(K/M)| }
\end{eqnarray*}
and thus the conclusion of the proposition follows trivially.
\end{proof}

Herein, for notational convenience, we denote
\[
I_k(a,b):=\int_a^b z^k  f^\beta (z) dz
\]

\begin{proposition}[Intermediate integral exponential decay] \label{prop:expinter}
For $r \le 1/3$ and $f(\cdot)$ as in Theorem~\ref{Thm:scaling}, specifically the Dutchman's cap assumption in \eqref{eq:BacMor}, then for each $k \ge 0$ there exists a constant $D_1>0$ such that
\begin{equation}
	\limsup_{\beta \rightarrow \infty} \left|e^{D_1\beta^{(1-2r)}} \left( I_k(-M, -M\beta^{-r})+ I_k(M\beta^{-r}, M)\right)\right| = 0 <\infty 
\end{equation}
\end{proposition}
\begin{proof}
Take $\delta_2$ as specified in the Dutchman's cap assumption in \eqref{eq:BacMor} and recall we have
\begin{equation}
  f(z) \le\begin{cases}
     \exp \left( - z^2 \frac{|H|}{4} \right), &  \text{if $|z|<\delta_2$}.\\
     \exp \left( - \delta_2^2 \frac{|H|}{4} \right), & \text{if $\delta_2\le |z| \le M$}. \label{recall:bacmor}
  \end{cases}
\end{equation}
Assume that $\beta$ is sufficiently large so that $M\beta^{-r}<\delta_2$. By splitting the integral
\[
I_k(M\beta^{-r}, M) = I_k(M\beta^{-r}, \delta_2)+I_k(\delta_2, M).
\]
Focusing on the first term and using \eqref{recall:bacmor} and noting that $1-2r>0$
\begin{eqnarray}
|I_k(M\beta^{-r}, \delta_2)| &=& \int_{M\beta^{-r}}^{\delta_2} z^k f^\beta (z)  dz
\leq \int_{M\beta^{-r}}^{\delta_2} z^k  \exp \left( - z^2 \beta \frac{|H|}{4} \right) dz \nonumber\\
&\leq& \int_{M\beta^{-r}}^{\delta_2} \delta_2^k \exp \left( - M^2 \beta^{(1-2r)} \frac{|H|}{4} \right) dz \nonumber\\
&=& \delta_2^k\left[\delta_2 - M\beta^{-r}\right]  \exp \left( - M^2 \beta^{(1-2r)} \frac{|H|}{4} \right). \label{eq:firstbrekinterinteg}
\end{eqnarray}

Then for the second term again using \eqref{recall:bacmor}  then
\begin{eqnarray}
|I_k(\delta_2,M)| &=& \int_{\delta_2}^M z^k f^\beta (z) dz
\leq \int_{\delta_2}^M M^k  \exp \left( - \delta_2^2 \beta \frac{|H|}{4} \right) dz \nonumber\\
&\leq& (M-\delta_2)M^k    \exp \left( - \delta_2^2 \beta \frac{|H|}{4} \right) .\label{eq:secbrekinterinteg}
\end{eqnarray}
Combining \eqref{eq:firstbrekinterinteg}, \eqref{eq:secbrekinterinteg} and an identical procedural application to the lower integrand, $I_k(-M,-M\beta^{-r})$, then the proposition is proved.
\end{proof}

\begin{lemma}[Limiting integrals] \label{Lemma:Keyints}
For $f(\cdot)$ in Theorem~\ref{Thm:scaling} 
\begin{equation}
  \lim_{\beta\rightarrow \infty }\left\{ \beta^{g(k)}\sqrt{\frac{\beta |H|}{2\pi}}\int_{\mathbb{R}} z^k f^\beta(z) dz\right\}=\begin{cases}
     \frac{\mathbb{E} (S^{k+3})h'''(0)}{6|H|^{(k+3)/2}}, &  \text{if $k$ is odd}.\\
     \frac{\mathbb{E} (S^{k})}{|H|^{k/2}}, & \text{if $k$ is even}
  \end{cases} \label{eq:crucialintegrals}
\end{equation}
where $S$ is distributed according to a standard normal distribution, i.e.\ $N(0,1)$, and 
\begin{equation*}
  g(k)=\begin{cases}
    \frac{k+1}{2} , &  \text{if $k$ is odd}\\
     \frac{k}{2} , & \text{if $k$ is even}.
  \end{cases}
\end{equation*}
\end{lemma}
\begin{proof}
By Propositions~\ref{prop:Tailexp} and \ref{prop:expinter} then one can derive that for $r=1/3$ 
\begin{equation}
	I_k(-\infty,\infty) = I_k(-M\beta^{-r}, M\beta^{-r}) +R(\beta) \label{eq:keytailkiller}
\end{equation}
where $R(\beta)=\mathcal{O} \left(\exp \left( -D \beta^{(1-2r)} \right)\right)$ for some constant $D>0$ i.e.\ exponentially decaying in $\beta$. Consequently, we focus on $I_k(-M\beta^{-r}, M\beta^{-r})$.

By Taylor expansion with Taylor remainder $|\xi_z|<|z|$
\begin{eqnarray}
f^\beta (z)&=&\exp \left( \beta \left[ z^2 \frac{H}{2} + z^3 \frac{h'''(0)}{6} +z^4 \frac{h''''(\xi_z)}{24}   \right]  \right) \nonumber\\
&=& \sqrt{\frac{2\pi}{\beta |H|}} \phi_\beta(z) \exp \left( \beta z^3 \frac{h'''(0)}{6}\right) \exp \left(\beta z^4 \frac{h''''(\xi_z)}{24}  \right) \label{eq:integTay111}
\end{eqnarray}
where $\phi_\beta(\cdot)$ denotes the pdf of a $N(0, \sigma(\beta)^2)$. 

By further Taylor expansion 
\begin{eqnarray}
\exp \left( \beta z^3 \frac{h'''(0)}{6}\right) &=&  1 + \beta z^3 \frac{h'''(0)}{6} + \beta^2 z^6 \frac{h'''(0)^2}{72} e^{\xi^1_{(\beta,z)}}  \label{eq:cubicTay} \\
\exp \left(\beta z^4 \frac{h''''(\xi_z)}{24}  \right) &=&     1 + \beta z^4 \frac{h''''(\xi_z)}{24}e^{\xi^2_{(\beta,z)}}       \label{eq:quarticTay}
\end{eqnarray}
where $\xi^1_{(\beta,z)}$ and $\xi^2_{(\beta,z)}$ are Taylor remainders such that
\begin{eqnarray}
|\xi^1_{(\beta,z)}| &\le&  \left| \beta z^3 \frac{h'''(0)}{6} \right| \label{eq:cubicremainder}\\
|\xi^2_{(\beta,z)}| &\le&  \left| \beta z^4 \frac{h''''(\xi_z)}{24} \right| . \label{eq:quarticremainder}
\end{eqnarray}
Combining \eqref{eq:integTay111}, \eqref{eq:cubicTay} and \eqref{eq:quarticTay} 
\begin{equation}
	I_k(-M\beta^{-r}, M\beta^{-r}) = \sqrt{\frac{2\pi}{-\beta H}} \sum_{m=1}^6 T_m(k) \label{eq:Tsum}
\end{equation}
where
\begin{eqnarray*}
T_1(k) &=& \int_{-M\beta^{-r}}^{M\beta^{-r}} z^k \phi_\beta(z) dz \\
T_2(k) &=& \int_{-M\beta^{-r}}^{M\beta^{-r}} \beta z^{k+3} \frac{h'''(0)}{6} \phi_\beta(z) dz \\
T_3 (k)&=& \int_{-M\beta^{-r}}^{M\beta^{-r}} \beta^2 z^{k+6} \left(\frac{h'''(0)^2}{72}\right) e^{\xi^1_{(\beta,z)}}  \phi_\beta(z) dz \\
T_4(k) &=& \int_{-M\beta^{-r}}^{M\beta^{-r}} \beta z^{k+4} \left(\frac{h''''(\xi_z)}{24}\right)e^{\xi^2_{(\beta,z)}} \phi_\beta(z) dz \\
T_5(k) &=& \int_{-M\beta^{-r}}^{M\beta^{-r}} \beta^2 z^{k+7} \left(\frac{h'''(0)}{6}\right)\left(\frac{h''''(\xi_z)}{24}\right)e^{\xi^2_{(\beta,z)}} \phi_\beta(z) dz \\
T_6(k) &=& \int_{-M\beta^{-r}}^{M\beta^{-r}} \beta^3 z^{k+10} \left(\frac{h'''(0)^2}{72}\right) \left(\frac{h''''(\xi_z)}{24}\right)e^{\xi^1_{(\beta,z)}}e^{\xi^2_{(\beta,z)}} \phi_\beta(z) dz. \\
\end{eqnarray*}
Consider the case that $k=3$ and all others follow by identical methodology. For $k=3$ it is clear that $T_1(3)=0$ so consider $T_3(3)$ with \eqref{eq:cubicremainder} and using the fact that $e^x \le e^{|x|}$ and $r=1/3$
\begin{eqnarray}
|T_3(3)| &\le& \int_{-M\beta^{-r}}^{M\beta^{-r}} \beta^2 |z|^{9} \left(\frac{h'''(0)^2}{72}\right) \exp \left[\beta |z|^{3} \left|\frac{h'''(0)}{6}\right| \right]  \phi_\beta(z) dz \nonumber \\
&\le&  \beta^2 \left(\frac{h'''(0)^2}{72}\right) \exp \left[M^{3} \left|\frac{h'''(0)}{6}\right| \right]  \int_{-M\beta^{-r}}^{M\beta^{-r}} |z|^{9} \phi_\beta(z) dz  \nonumber \\
&\le&  \beta^2 \left(\frac{h'''(0)^2}{72}\right) \exp \left[M^{3} \left|\frac{h'''(0)}{6}\right| \right]  \int_{\mathbb{R}} |z|^{9} \phi_\beta(z) dz   \nonumber \\
&=& C_3^3 \beta^2 \int_{\mathbb{R}} |z|^{9} \phi_\beta(z) dz \label{eq:kis3term}
\end{eqnarray}
where $C^3_3 \in \mathbb{R}_{+}$ is a constant.

Similarly, there exists constants $C_4^3,C_5^3,C_6^3$ such that
\begin{eqnarray}
|T_4(3)| &\le&  C_4^3 \beta \int_{\mathbb{R}} |z|^{7} \phi_\beta(z) dz \label{eq:kis4term} \\
|T_5(3)| &\le&  C_5^3 \beta^2 \int_{\mathbb{R}} |z|^{10} \phi_\beta(z) dz \label{eq:kis5term} \\
|T_6(3)| &\le&  C_6^3 \beta^3 \int_{\mathbb{R}} |z|^{13} \phi_\beta(z) dz. \label{eq:kis6term} 
\end{eqnarray}
Routine calculations using integration by parts reveals that
\begin{equation}
  \int_{\mathbb{R}} |z|^{m} \phi_\beta(z) dz \le \begin{cases}
     \frac{96}{\sqrt{2\pi}(\beta |H|)^{7/2}}, &  \text{if $m=7$}\\
     \frac{1536}{\sqrt{2\pi}(\beta |H|)^{9/2}}, &  \text{if $m=9$}\\
		\frac{945}{(\beta |H|)^{5}}, &  \text{if $m=10$}\\
		\frac{92160}{\sqrt{2\pi}(\beta |H|)^{13/2}}, &  \text{if $m=13$}\\
  \end{cases}\label{eq:kis3powerparts}
\end{equation}
Combining \eqref{eq:kis3term}, \eqref{eq:kis4term}, \eqref{eq:kis5term}, \eqref{eq:kis6term} and \eqref{eq:kis3powerparts} shows that for $m \in \{3,4,5,6\}$
\begin{equation}
	|T_m(3)| \le \mathcal{O} \left( \frac{1}{\beta^{5/2}} \right). \label{eq:kis3higherorder}
\end{equation}
Hence the crucial term is $T_2(3)$,
\begin{eqnarray*}
T_2(3) &=& \int_{-M\beta^{-r}}^{M\beta^{-r}} \beta z^{6} \frac{h'''(0)}{6} \phi_\beta(z) dz  \\
&=&\frac{h'''(0)}{6} \beta \int_{\mathbb{R}}  z^{6} \phi_\beta(z) dz + R_2(3)
\end{eqnarray*}
where using identical methodology to Propositions~\ref{prop:Tailexp} and \ref{prop:expinter} one can show that for a constant $F>0$ then $R_2(3) = \mathcal{O} \left( \beta^{5/2} e^{-F\beta^{2/3}} \right)$ and so trivially $R_2(3) = \mathcal{O} \left( \beta^{-5/2} \right)$. Hence, 
\begin{equation}
	T_2(3) = \frac{5 h'''(0)}{2 \beta^2 |H|^3} +  \mathcal{O} \left( \beta^{-5/2} \right). \label{eq:crucialT2term}
\end{equation}
 Thus combining the results of \eqref{eq:keytailkiller}, \eqref{eq:Tsum}, \eqref{eq:kis3higherorder} and \eqref{eq:crucialT2term} then 
\begin{equation}
	\lim_{\beta \rightarrow \infty}\left\{ \beta^2 \sqrt{\frac{\beta |H|}{2\pi}}\int_{\mathbb{R}} z^k f^\beta(z) dz\right\} = \frac{5 h'''(0)}{2 |H|^3}. \label{eq:crucialT2term1}
\end{equation}
This proves the lemma for k=3 and identical methodology all other values of $k \in \mathbb{N}_{0}$ follow.
\end{proof}

For the purposes of proving the main result the following corollary concatenates the key results from Lemma~\ref{Lemma:Keyints}.

\begin{corollary}[Key moment results] \label{cor:keyintsexplained}
From Lemma~\ref{Lemma:Keyints} and the details of the proof it can be concluded that:
\begin{eqnarray}
N(\beta) &=&  \sqrt{\frac{2\pi}{-\beta H}} + \mathcal{O}\left(\beta^{-1}\right) \label{eq:noramliseterm}\\
\mu_\beta:=\int_{\mathbb{R}} z^3 \frac{f^\beta(z)}{N(\beta)} dz &=&  \frac{5 h'''(0)}{2 \beta^2|H|^3} + \mathcal{O}\left(\beta^{-5/2}\right) \label{eq:expterm}\\
\sigma^2_\beta:= \int_{\mathbb{R}} z^6 \frac{f^\beta(z)}{N(\beta)} dz &=&   \frac{15 }{\beta^3 |H|^3} + \mathcal{O}\left(\beta^{-7/2}\right) \label{eq:varterm}
\end{eqnarray}
\end{corollary}
\begin{proof}
All methodology and details from Lemma~\ref{Lemma:Keyints} above.
\end{proof}

\subsubsection{Main result}
Now with the results established in Corollary~\ref{cor:keyintsexplained} we are in a position to establish the result of Theorem~\ref{Thm:scaling}.

\begin{proof}[Proof of Theorem~\ref{Thm:scaling}]
Recall the expected acceptance probability of a Mode Leap Point Independence sampler proposal at the coldest temperature level in this $d$-dimensional setting as
\begin{equation}
a(d) = \mathbb{E}(A_{IS}(x,y)) = \mathbb{E}(\min \left\{ 1,  e^{B(d)}  \right\}),
\end{equation}
where the independence sampler acceptance probability, $A_{IS}$, is defined in \eqref{eq:probability_of_acceptance_independent_sampler} and the expectation is taken with respect to $x\sim \pi_{\beta}$ as defined in \eqref{eq:simplified_version_target} and $y \sim q_\beta$ as defined in \eqref{eq:simplified_version_proposal}.\\
As stated in \eqref{eq:exponential_acceptance_ratio}, the logged acceptance ratio $B$ is given by 
\begin{equation}
	B(d)= \sum_{i=1}^d \left[ \beta h( y_i) - \log \phi_\beta (y_i)   \right]- \left[ \beta h( x_i) - \log \phi_\beta(x_i)   \right]. \label{eq:loggedA}
\end{equation}
By Taylor expansion about the point 0 (which is the mode of $f(\cdot)$ and so $h'(0)=0$) there exists $\xi_x$ where $0<|\xi_{x}|<|x|$ such that
\begin{eqnarray}
h(
x) =h(0)+ \frac{ x^2}{2} H + \frac{ x^3}{6} h'''(0)+  \frac{ x^4}{24} h^{''''}(0)+ \frac{ x^5}{120} h'''''(\xi_{x}). \label{eq:TAy1}
\end{eqnarray}
Thus \eqref{eq:loggedA} becomes
\begin{eqnarray}
B(d) &=& \sum_{i=1}^d  \Bigg[ \frac{\beta  h'''(0)}{6}(y_i^3-x_i^3) +  \frac{\beta h''''(0)}{24}\left(y_i^4-x_i^4 \right) \nonumber\\
&&~~~~~~~~ +\frac{\beta }{120}\left(y_i^5 h'''''(\xi_{y_i})-x_i^5 h'''''(\xi_{x_i})\right) \Bigg]. \nonumber  
\end{eqnarray}

\begin{lemma}
Under the setting of Theorem~\ref{Thm:scaling}, then if $\beta= \ell d$ as in \eqref{eq:scalingnec_smooth} then the term 
\begin{equation}
	W(d):= \sum_{i=1}^d \frac{\beta  h'''(0)}{6}(y_i^3-x_i^3)
\end{equation}
is such that as $d \rightarrow \infty$
\begin{equation}
	W(d) \Rightarrow N\left(-\frac{5h'''(0)^2}{12 \ell |H|^3} , \frac{5h'''(0)^2}{6 \ell |H|^3} \right)
\end{equation}
where the convergence is weak convergence. \label{lem:asympnorm}
\end{lemma}
\begin{proof}
Since the distribution of the $x_i$ and $y_i$ is dependent on  $d$ then we are in the situation of triangular arrays. As such we aim to use the Lindeberg-Feller theorem, see e.g.\ \citet{Durrett2010}. Define for each $d \in \mathbb{N}$ and for $i \in \{1,\ldots,d\}$
\[
X_{(d,i)} :=   \left( \frac{\beta  h'''(0)}{6} \right) \left( x_i^3-\mu_\beta\right).
\]
\textbf{\textit{Lindeberg-Feller Conditions:}} To apply the result of the Lindeberg-Feller theorem from \citet{Durrett2010} to the $X_{(d,i)} $'s, it is sufficient to establish the following two conditions as $d \rightarrow \infty$
\begin{enumerate}[label=\roman*)]
	\item There exists $\sigma^2 \in \mathbb{R}_{+}$ such that
	\begin{equation}
		\sum_{i=1}^d  \mathbb{E}(X_{(d,i)}^2) \rightarrow \sigma^2 >0 \label{eq:LF1}
	\end{equation}
	\item For all $\epsilon>0$ and with $A_\epsilon^d := \{ |X_{(d,i)}| > \epsilon \}$
	\begin{equation}
		\lim_{d\rightarrow \infty}\sum_{i=1}^d  \mathbb{E}(|X_{(d,i)}|^2 \mathbbm{1}_{A_\epsilon}) =0 \label{eq:LF2}
	\end{equation}
\end{enumerate}
It is simple to show that L-F condition (i) in \eqref{eq:LF1} holds, since by \eqref{eq:expterm} and \eqref{eq:varterm}
\begin{equation}
	\mathbb{E}\left(X_{(d,i)}^2\right) = \frac{5h'''(0)^2}{12   \ell d |H|^3} +\mathcal{O}\left(d^{-3/2}\right) \nonumber
\end{equation}
and so it is immediate that 
\begin{equation}
	\sum_{i=1}^d  \mathbb{E}(X_{(d,i)}^2) \rightarrow \frac{5h'''(0)^2}{12   \ell |H|^3} >0
\end{equation}

To show that L-F condition (ii) in \eqref{eq:LF2} holds is more involved. To this end let $\epsilon>0$
\begin{eqnarray}
&& 	\mathbb{E}(|X_{(d,i)}|^2 \mathbbm{1}_{A_\epsilon^d}) \nonumber \\  &=&  \left( \frac{\beta  h'''(0)}{6} \right)^2\left[\mathbb{E}\left(x_i^6 \mathbbm{1}_{A_\epsilon^d}\right) - 2 \mu_\beta \mathbb{E}\left(x_i^3 \mathbbm{1}_{A_\epsilon^d}\right) +  \mu_\beta^2 \mathbb{P}(A_\epsilon^d) \right]  \label{eq:LF2expansion}\\
	&\le& \left|\frac{\beta  h'''(0)}{6} \right|^2\left[\mathbb{E}\left(x_i^6 \mathbbm{1}_{A_\epsilon^d}\right) + 2 |\mu_\beta| \mathbb{E}\left(|x_i|^3 \mathbbm{1}_{A_\epsilon^d}\right) +  \mu_\beta^2 \mathbb{P}(A_\epsilon^d) \right] \label{eq:LF2expansionineq}
\end{eqnarray}
Using identical methodology to that used in proving Lemma~\ref{Lemma:Keyints} with only a trivial extension one can show that
\[
\mathbb{E}\left(|x_i|^3 \mathbbm{1}_{A_\epsilon^d}\right) \le \mathcal{O} \left( d^{-3/2} \right)
\]
and so with $\mu_\beta=\mathcal{O}\left( d^{-2} \right) $ given in \eqref{eq:expterm} then
\begin{eqnarray*}
| \mu_\beta| \mathbb{E}\left(|x_i|^3 \mathbbm{1}_{A_\epsilon^d}\right) &\le&  \mathcal{O}\left( d^{-7/2} \right)\\
| \mu_\beta^2 \mathbb{P}(A_\epsilon^d) | &\le& |\mu_\beta^2| = \mathcal{O}\left( d^{-4} \right) 
\end{eqnarray*}
and thus we can conclude that 
\begin{equation}
	\lim_{d\rightarrow\infty} \sum_{i=1}^d \left|\frac{\beta  h'''(0)}{6} \right|^2 \left[2 \mu_\beta \mathbb{E}\left(x_i^3 \mathbbm{1}_{A_\epsilon^d}\right) +  \mu_\beta^2 \mathbb{P}(A_\epsilon^d) \right] =0. \label{eq:LF2therest}
\end{equation}
Consequently we must analyse the behaviour of $\mathbb{E}\left(x^6 \mathbbm{1}_{A_\epsilon^d}\right)$. We begin by noting that $A_\epsilon^d$ can be equivalently written as 
\begin{eqnarray}
A_\epsilon^d &=& \left\{ x > \left(\mu_\beta + \frac{6\epsilon}{\beta |h'''(0)|}  \right)^{1/3} \right\} \cup \left\{ x < \left(  \mu_\beta- \frac{6\epsilon}{\beta |h'''(0)|}\right)^{1/3}\right\} \nonumber \\
&:=& F_{+\epsilon}^d \cup F_{-\epsilon}^d . \label{eq:breakdownset}
\end{eqnarray}
With this notation,
\begin{equation}
	\mathbb{E}\left(x^6 \mathbbm{1}_{A_\epsilon^d}\right)=\mathbb{E}\left(x^6 \mathbbm{1}_{F_{+\epsilon}^d }\right)+\mathbb{E}\left(x^6 \mathbbm{1}_{F_{-\epsilon}^d }\right). \label{eq:LFiisplit}
\end{equation}
By \eqref{eq:expterm} we have that $\mu_\beta=\mathcal{O}(d^{-2})$ and so from the equivalent formulation of $A_\epsilon^d$ in \eqref{eq:breakdownset} and for $d$ sufficiently large there exists a constant $J>0$ such that
\begin{eqnarray}
 F_{+\epsilon}^d &=&\left\{ x > \left(\mu_\beta + \frac{6\epsilon}{\beta |h'''(0)|}  \right)^{1/3} \right\} \nonumber \\
&\subseteq& \left\{ x > \frac{J}{d^{1/3} } \right\}.
\end{eqnarray}
Focusing on the first term on the RHS of \eqref{eq:LFiisplit} and with $M$ from \eqref{eq:polyT}
\begin{equation}
	\mathbb{E}\left(x^6 \mathbbm{1}_{F_{+\epsilon}^d }\right) \le \frac{1}{N(\beta)} \left[I_6(J d^{-1/3},M) + I_6(M,\infty) \right]. \label{eq:furthersplitLFii}
\end{equation}
By Propositions~\ref{prop:Tailexp} and \ref{prop:expinter} then there exists a constant $W_1>0$ such that the RHS of \eqref{eq:furthersplitLFii} is $\mathcal{O}\left(  e^{-W_1d^{1/3}} \right)$. By identical methodology one can show that there exists $W_2>0$ such that
\begin{equation*}
	\mathbb{E}\left(x^6 \mathbbm{1}_{F_{-\epsilon}^d }\right) \le \mathcal{O}\left(  e^{-W_2d^{1/3}} \right)
\end{equation*}
and so with $W_3 := \min\{W_1,W_2\}$ then
\begin{equation}
	\mathbb{E}\left(x^6 \mathbbm{1}_{A_\epsilon^d}\right) = \mathcal{O}\left(  e^{-W_3d^{1/3}} \right). \label{eq:score}
\end{equation}
Consequently by \eqref{eq:score} then
\begin{equation}
	\lim_{d\rightarrow\infty} \sum_{i=1}^d \left|\frac{\beta  h'''(0)}{6} \right|^2 \mathbb{E}\left(x_i^6 \mathbbm{1}_{A_\epsilon^d}\right) =0. \label{eq:LF2thesix}
\end{equation}
So combining \eqref{eq:LF2therest} and \eqref{eq:LF2thesix} proves that the second Lindeberg-Feller condition in \eqref{eq:LF2} holds.

With \eqref{eq:LF1} and \eqref{eq:LF2} holding  for the triangular array of $X_{(d,i)}$'s then the Lindeberg-Feller theorem implies that that as $d\rightarrow \infty$
\begin{equation}
	-\sum_{i=1}^d X_{(d,i)} \Rightarrow N\left(0 , \frac{15h'''(0)^2}{36 \ell |H|^3} \right) .\nonumber
\end{equation}
However, by using the result of \eqref{eq:expterm} that gives an explicit value to the $\mu_\beta$ and the definition of the $X_{(d,i)}$'s then
\begin{equation}
	-\sum_{i=1}^d \frac{\beta  h'''(0)}{6} x_i^3 \Rightarrow N\left(-\frac{15h'''(0)^2}{36 \ell |H|^3} , \frac{5h'''(0)^2}{12 \ell |H|^3} \right). \label{eq:xasymp}
\end{equation}

For a given $d$, the $y_i \sim N(0, (\beta|H|)^{-1})$, thus they too fall into a triangular arrays setup with
\begin{eqnarray}
\mathbb{E} ( y_i^3 ) &=& 0  \label{eq:meanyi}  \\
\text{Var}(y_i^3) &=& \mathbb{E} ( y_i^6 ) = \frac{15}{\beta^3 |H|^3}  \label{eq:varyi}
\end{eqnarray}
and so setting
\[
Y_{(d,i)} = \left( \frac{\beta  h'''(0)}{6} \right) y_i^3
\]
one can follow a simplified version of the above Lindeberg-Feller procedure for the $X_{(d,i)}$'s to establish that as $d\rightarrow \infty$
\begin{equation}
	\sum_{i=1}^d \frac{\beta  h'''(0)}{6} y_i^3 \Rightarrow N\left(0 , \frac{5h'''(0)^2}{12 \ell (-h''(0))^3} \right) \label{eq:yasymp}
\end{equation}
Using the fact that the $x_i$'s and $y_i$'s are independent along with the weak convergences established in \eqref{eq:xasymp} and \eqref{eq:yasymp} the result of Lemma~\ref{lem:asympnorm} is established.
\end{proof}

\begin{lemma}\label{lem:decay}
Under the setting of Theorem~\ref{Thm:scaling}, then if $\beta= \ell d$ as in \eqref{eq:scalingnec_smooth} then as $d \rightarrow \infty$
\begin{equation}
	T_1(d) : = \sum_{i=1}^d \frac{\beta h''''(0)}{24}\left(y_i^4 -x_i^4 \right)  \rightarrow 0 ~~\text{in}~~\mathbb{P}.
\end{equation}
\end{lemma}
\begin{proof}
Fix $\epsilon>0$. Using identical methodology to that used in Lemma~\ref{Lemma:Keyints} then it can be shown that 
\begin{eqnarray}
	\mathbb{E}(x_i^4) &=& \frac{3}{\left(\beta |H| \right)^{2}} + \mathcal{O}\left(\beta^{-3}\right) \label{eq:4thx} \\
	\mathbb{E}(x_i^8) &=& \frac{105}{\left(\beta |H| \right)^{4}} + \mathcal{O}\left(\beta^{-5}\right) \label{eq:8thx}	
\end{eqnarray}
and from standard calculations using the Gaussianity of the $y_i$'s
\begin{eqnarray}
	\mathbb{E}(y_i^4) &=& \frac{3}{\left(\beta |H| \right)^{2}} \label{eq:4thy} \\
		\mathbb{E}(y_i^8) &=& \frac{105}{\left(\beta |H| \right)^{4}}. \label{eq:8thy}
\end{eqnarray}
Defining $\mu_{T_1(d)}:= \mathbb{E}(T_1(d))$ then by \eqref{eq:4thx} and \eqref{eq:4thy} as $d\rightarrow \infty$
\begin{equation}
	\mu_{T_1(d)}\rightarrow 0. 
\end{equation}
Thus there exists $U>0$ such that for all $d>U$ then $|\mu_{T_1(d)}|<\frac{\epsilon}{2}$. Furthermore, by \eqref{eq:4thx}, \eqref{eq:4thy}, \eqref{eq:8thx}	and \eqref{eq:8thy} then
\begin{equation}
	\text{Var}(T_1(d)) \rightarrow 0 ~~ \text{as}~~ d \rightarrow 0. \label{eq:degenvar}
\end{equation}
For $d>U$  and using Chebyshev's inequality and utilising \eqref{eq:degenvar}
\begin{equation}
	\mathbb{P}\left(|T_1(d)|>\epsilon \right)\le \mathbb{P}\left(|T_1(d)-\mu_{T_1(d)}| \ge \frac{\epsilon}{2}\right) \le \frac{4\text{Var}(T_1(d))}{\epsilon^2} \rightarrow 0 ~~\text{as}~~ d \rightarrow 0 \nonumber
\end{equation}
and so the result of Lemma~\ref{lem:decay} holds.
\end{proof}

\begin{lemma}\label{lem:decay3}
Under the setting of Theorem~\ref{Thm:scaling}, then if $\beta= \ell d$ as in \eqref{eq:scalingnec_smooth} then 
\begin{equation}
	T_2(d) : =  \sum_{i=1}^d \frac{\beta }{120}\left(y_i^5 h'''''(\xi_{y_i})-x_i^5 h'''''(\xi_{x_i})\right)  \label{eq:T(d)inproof}
\end{equation}
is such that $T_2(d) \rightarrow 0$ in $\mathbb{P}$ as $d\rightarrow \infty$. 
\end{lemma}
\begin{proof}
Using the assumption given in \eqref{eq:bddfourth} 
\begin{eqnarray}
|T_2(d)| \le \sum_{i=1}^d L\frac{\beta }{120}\left(|y_i|^5 +|x_i|^5 \right)=:T_2^r(d) >0 \label{eq:bddfourthineqonT}
\end{eqnarray}
and so if one can show that $T_2^r(d) \rightarrow 0$ in $\mathbb{P}$ then the required result in \eqref{eq:T(d)inproof} follows. Note that 
\begin{equation}
	\mathbb{E}(|y_i|^5) = \frac{16}{\sqrt{2\pi}\left(\beta |H|\right)^{5/2}}  \label{eq:fifthord}
\end{equation}
and using identical methodology to the proof of Lemma~\ref{Lemma:Keyints} then it can be shown that for some constant $E \in \mathbb{R}$
\begin{equation}
	\mathbb{E}(|x_i|^5) = \frac{E}{d^{5/2}} + \mathcal{O}\left(  d^{-3} \right) \label{eq:fifthordxis}
\end{equation}
Let $\epsilon>0$, then by Markov inequality and then using \eqref{eq:fifthord} and \eqref{eq:fifthordxis}
\begin{eqnarray}
\mathbb{P}(T_2^{r}(d)>\epsilon) &\le& \frac{L \ell d}{24 \epsilon} \sum_{i=1}^d\left[\mathbb{E}(|y_i|^5)+\mathbb{E}(|x_i|^5)\right] \nonumber \\
&=& \frac{L \ell d^2}{24 \epsilon} \left[\mathbb{E}(|y_1|^5)+\mathbb{E}(|x_1|^5)\right] \nonumber \\
&\rightarrow& 0~~~\text{as}~~~ d \rightarrow \infty. \nonumber
\end{eqnarray}
Thus the result of Lemma~\ref{lem:decay3} holds.
\end{proof}

Taking the results of Lemmata~\ref{lem:asympnorm}, \ref{lem:decay} and \ref{lem:decay3} and trivial application of Slutsky's Theorem, then we have that as $d \rightarrow \infty$ then 
\begin{equation}
	B(d) = W(d)+T_1(d)+T_2(d)\Rightarrow B \sim N\left(-\frac{5h'''(0)^2}{12 \ell (-h''(0))^3} , \frac{5h'''(0)^2}{6 \ell (-h''(0))^3} \right) \label{eqaymB}
\end{equation}
where the convergence is weak convergence. 

We are interested in the behaviour of $\mathbb{E} \left[ \min\left\{1,e^{B(d)} \right\} \right]$ in the limit as $d \rightarrow \infty$. Since the function, $\min\left\{1,e^{B(d)} \right\}$, is a bounded continuous function of $B(d)$ then by weak convergence of $B(d)$ established above in  \eqref{eqaymB}, then
\[    
    \lim _{d\rightarrow \infty} \mathbb{E} \left[ \min\left\{1,e^{B(d)} \right\} \right] = \mathbb{E} \left[ \min\left\{1,e^{B} \right\} \right].
\]
\begin{proposition} \label{prop:Gaussres}
Suppose that  $G\sim N(-\frac{\sigma^2}{2}, \sigma^2)$,
\begin{equation}
\mathbb{E}\left(1\wedge e^G \right)= 2 \Phi\left(-\frac{\sigma}{2}\right) .\nonumber
\end{equation}
\end{proposition}
\begin{proof}
A routine calculation (see e.g.,\ \cite{roberts1997weak}).
\end{proof}
Applying Proposition~\ref{prop:Gaussres} to the asymptotic form of $B$ in \eqref{eqaymB} immediately gives the required result and completes the proof of Theorem~\ref{Thm:scaling}.
\end{proof}

\subsection{Relaxed assumptions}
\label{proofRelaxed}
This section extends the previous proof to a slightly more general setting, obtained by relaxing the assumptions; see Theorem~\ref{Thm:scalingRelaxed}.
\subsubsection{Technical Lemma}
First, we define the step functions that will be used when dealing with left and right derivatives.
\begin{definition}
\label{step}
    (Step functions). We denote by $\mathcal S_0$ the class of step functions, such that for $g \in \mathcal S_0$,
\begin{align}
    \forall x \neq 0, \: g\of{x} = g\of{0_-} \1{\bp{-\infty, 0}} + g\of{0_+} \1{\bp{0, + \infty}} \label{step_equation}
\end{align}
\end{definition}
Then, the following very useful technical lemma allows us to compute quantities of interest with respect to the density considered in the relaxed scenario:
\begin{lemma}[Limiting integrals]
\label{lemma}
For $f(\cdot)$ as in Theorem~\ref{Thm:scalingRelaxed} and $g \in \mathcal S_0$,
\begin{align}
\int_{\RR} g\of{z} z^{k} f\of{z}^{\beta} \intd z
&=
\sqrt{\frac{2\pi}{\beta \absx{H}}}
\Biggl[
\frac{g\of{0_+}+\bp{-1}^k g\of{0_-}}
     {2\bp{\beta \absx{H}}^{k/2}}
\EE{\absx{S}^k}
\notag \\
&\quad
+
\frac{
g\of{0_+}h^{\prime\prime\prime}\of{0_+}
+\bp{-1}^{k+1}g\of{0_-}h^{\prime\prime\prime}\of{0_-}
}
{12\beta^{(k+1)/2}\absx{H}^{(k+3)/2}}
\EE{\absx{S}^{k+3}}
\notag \\
&\quad
+
O\of{\beta^{-(k+2)/2}}
\Biggr].
\end{align}
    where $S$ is distributed according to a standard normal distribution, i.e. $N(0,1)$.
\end{lemma}

\begin{proof}
    
For notational convenience, we denote :
\begin{align*}
    h_* := \max\of{\absx{h^{\prime \prime \prime}\of{0_-}}, \absx{h^{\prime \prime \prime}\of{0_+}}}
\end{align*}

\begin{proposition}
\label{proposition:tail}

(Exponential decay of tail integrals). With $f(\cdot)$ as defined earlier and specifically for the polynomial tail condition given in \eqref{eq:polyT} and for any $g \in \mathcal S_0$, then for each $k \geq 0$, there exists constants $C_{1}>0$ and $C_{2} \geq 0$ such that
\end{proposition}

\begin{align}
\limsup _{\beta \rightarrow \infty}\left|e^{C_{1} \beta} \int_{\mathbb{R} \backslash B_{M}(0)} g\of{z} z^{k} f^{\beta}(z) \intd z\right|=C_{2}<\infty \label{exp_decay}
\end{align}
\begin{proof}
 With $K, \gamma$ and $M$ all specified in \eqref{eq:polyT} (and wlog $K<M^\gamma$) and assuming $\beta$ sufficiently large,
\begin{align*}
    \left|\int_{M}^{\infty} g\of{z} z^{k} f^{\beta}(z) \intd z\right| & \leq \int_{M}^{\infty} \absx{g\of{z}} |z|^{k}|z|^{-\gamma \beta} \frac{f^{\beta}(z)}{|z|^{-\gamma \beta}} \intd z \\
& \leq K^{\beta} \absx{g\of{0_+}}\int_{M}^{\infty}|z|^{k-\gamma \beta} \intd z \\
& \leq  \absx{g\of{0_+}} \bp{\frac{K}{M^\gamma}}^\beta \frac{M^{k+1}}{(\gamma \beta-(k+1))} \\
& =\frac{M^{k+1}  \absx{g\of{0_+}}}{(\gamma \beta-(k+1))} e^{-\beta|\log (K / M^\gamma)|}
\end{align*}
and thus the conclusion follows trivially.
\end{proof}
Moving further, once again for notational purposes, we denote for any $g \in \mathcal S_0$ - the class of step functions, see \cref{step} :
\begin{align*}
    I_{k}^{g}\of{a,b} :=  \int_a^b g\of{z} z^k f^\beta\of{z} \intd z
\end{align*}
\begin{proposition}
\label{proposition:inter}
    (Intermediate integral exponential decay). For $r \leq 1 / 3$ and $f(\cdot)$ presented earlier, specifically the Dutchman's cap assumption in \eqref{eq:BacMor} and any $g \in \mathcal S_0$, then for each $k \geq 0$,there exists a constant $D_{1}>0$ such that
\begin{align}
\limsup _{\beta \rightarrow \infty}\left|e^{D_{1} \beta^{(1-2 r)}}\left(I_{k}^g\left(-M,-M \beta^{-r}\right)+I_{k}^g\left(M \beta^{-r}, M\right)\right)\right|=0<\infty \label{inter_decay}
\end{align}
\end{proposition}
\begin{proof}
    Take $\delta_{2}$ as specified in the Dutchman's cap assumption in \eqref{eq:BacMor}. \\
Assume that $\beta$ is sufficiently large so that $M \beta^{-r}<\delta_{2}$. By splitting the integral

$$
I_{k}^{p}\left(M \beta^{-r}, M\right)=I_{k}^{g}\left(M \beta^{-r}, \delta_{2}\right)+I_{k}^{g}\left(\delta_{2}, M\right)
$$

Focusing on the first term and using \eqref{eq:BacMor} and noting that $1-2 r>0$

\begin{align}
\left|I_{k}\left(M \beta^{-r}, \delta_{2}\right)\right| & \leq \int_{M \beta^{-r}}^{\delta_{2}} \absx{g\of{z}} z^{k} f^{\beta}(z) \intd z \leq \absx{g\of{0_+}} \int_{M \beta^{-r}}^{\delta_{2}} z^{k} \exp \left(-z^{2} \beta \frac{|H|}{4}\right) \intd z \notag \\
& \leq \absx{g\of{0_+}} \int_{M \beta^{-r}}^{\delta_{2}} \delta_{2}^{k} \exp \left(-M^{2} \beta^{(1-2 r)} \frac{|H|}{4}\right) \intd z \notag \\
& =  \absx{g\of{0_+}} \delta_{2}^{k}\left[\delta_{2}-M \beta^{-r}\right] \exp \left(-M^{2} \beta^{(1-2 r)} \frac{|H|}{4}\right) \label{inter_1}
\end{align}

Then for the second term again using \eqref{eq:BacMor}, we have

\begin{align}
\left|I_{k}\left(\delta_{2}, M\right)\right| &\leq \int_{\delta_{2}}^{M}  \absx{g\of{z}} z^{k} f^{\beta}(z) \intd z \leq  \absx{g\of{0_+}} \frac{\sqrt{\beta}}{J} \int_{\delta_{2}}^{M} M^{k} \exp \left(-\delta_{2}^{2} \beta \frac{|H|}{4}\right) \intd z \notag \\
& \leq \absx{g\of{0_+}} \left(M-\delta_{2}\right) M^{k} \exp \left(-\delta_{2}^{2} \beta \frac{|H|}{4}\right) \label{inter_2}
\end{align}

Combining \eqref{inter_1}, \eqref{inter_2} and an identical procedural application to the lower integrand, $I_{k}^p \left(-M,-M \beta^{-r}\right)$, then the proposition is proved.
\end{proof}
Coming back to \Cref{lemma}, by \cref{proposition:tail,proposition:inter} then one can derive that for $r=1 / 3$

\begin{align}
I_{k}^g(-\infty, \infty)=I_{k}^g\left(-M \beta^{-r}, M \beta^{-r}\right)+R(\beta) \label{asymptotic_split}
\end{align}
where $R(\beta)=\mathcal{O}\left(\exp \left(-D \beta^{(1-2 r)}\right)\right)$ for some constant $D>0$, i.e. exponentially decaying in $\beta$. Consequently, we focus on $I_{k}^p\left(-M \beta^{-r}, M \beta^{-r}\right)$.

By Taylor expansion with Taylor remainder $\left|\xi_{z}\right|<|z|$

\begin{align}
f^{\beta}(z) & =\exp \left(\beta\left[  z^{2} \frac{h^{\prime \prime}\of{0_{\sgn{z}}}}{2}+z^{3} \frac{h^{\prime \prime \prime}(0_{\sgn{z}})}{6} + z^4 \frac{h^{\prime \prime \prime \prime}(\xi_z)}{24}  \right]\right) \notag \\
& =\sqrt{\frac{2 \pi}{\beta \absx{h^{\prime \prime }\of{0_{\sgn{z}}}}}} \phi_{\beta \sgn{z}}(z)  \exp \left(\beta z^{3} \frac{h^{\prime \prime \prime}(0_{\sgn{z}})}{6}\right) \exp\of{\beta  z^4 \frac{h^{\prime \prime \prime \prime}(\xi_z)}{24}} \label{taylor_f}
\end{align}

where $\phi_{\beta \sgn{z}}(\cdot)$ denotes the pdf of a $\cN\left(0, \bp{\beta  \absx{h^{\prime \prime }\of{0_{\sgn{z}}}}}^{-1} \right)$.

By further Taylor expansion,
\begin{align}
\exp \left(\beta z^{3} \frac{h^{\prime \prime \prime}(\xi_z)}{6}\right) & =1+\beta z^{3} \frac{h^{\prime \prime \prime}(0_{\sgn{z}})}{6}+\beta^{2} z^{6} \frac{h^{\prime \prime \prime}(0_{\sgn{z}})^{2}}{72} e^{\xi^1_{(\beta, z)}}\label{taylor_bis} \\
 \exp\of{\beta  z^4 \frac{h^{\prime \prime \prime \prime}(\xi_z)}{24}} &= 1 + \beta z^4 \frac{h^{\prime \prime \prime \prime}(\xi_z)}{24} e^{\xi^2_{(\beta, z)}} \label{taylor_bis2}
\end{align}

where $\xi_{(\beta, z)}^{k}$ are Taylor remainders such that
\begin{align}
& \left|\xi_{(\beta, z)}^{1}\right| \leq\left|\beta z^{3} \frac{h^{\prime \prime \prime}(0_{\sgn{z}})}{6}\right|  \label{remainder} \\
&  \left|\xi_{(\beta, z)}^{2}\right| \leq\left|\beta z^{4} \frac{h^{\prime \prime  \prime\prime}(\xi_z)}{24}\right|  \label{remainder2}
\end{align}

Combining \eqref{taylor_f} and \eqref{taylor_bis}

\begin{align}
I_{k}^g\left(-M \beta^{-r}, M \beta^{-r}\right)= \sum_{m=1}^{6} T_{m}(k,g) \label{decomposition}
\end{align}
where 
\begin{align*}
    T_{1}(k,g) & = \sqrt{\frac{2 \pi}{\beta}} \int_{-M \beta^{-r}}^{M \beta^{-r}}  \sqrt{\frac{1}{\absx{h^{\prime \prime}\of{0_{\sgn{z}}}}}} g\of{z}z^{k} \phi_{\beta \sgn{z}}(z) \intd z \\
T_{2}(k,g) & = \sqrt{\frac{2 \pi}{\beta}}\int_{-M \beta^{-r}}^{M \beta^{-r}}  \sqrt{\frac{1}{\absx{h^{\prime \prime}\of{0_{\sgn{z}}}}}} g\of{z} \beta z^{k+3} \frac{h^{\prime \prime \prime }(0_{\sgn{z}})}{6} \phi_{\beta \sgn{z}}(z) \intd z \\
T_{3}(k,g) & = \sqrt{\frac{2 \pi}{\beta}}\int_{-M \beta^{-r}}^{M \beta^{-r}}  \sqrt{\frac{1}{\absx{h^{\prime \prime}\of{0_{\sgn{z}}}}}} g\of{z}\beta^2 z^{k+6} \frac{h^{\prime \prime \prime }(0_{\sgn{z}})^2}{72}  e^{\xi_{(\beta, z)}^{1}} \phi_{\beta \sgn{z}}(z) \intd z \\
T_{4}(k,g) & = \sqrt{\frac{2 \pi}{\beta}} \int_{-M \beta^{-r}}^{M \beta^{-r}} \sqrt{\frac{1}{\absx{h^{\prime \prime}\of{0_{\sgn{z}}}}}} \beta z^{k+4} \frac{h^{\prime \prime \prime \prime }(\xi_z)}{24}  e^{\xi_{(\beta, z)}^{2}} \phi_{\beta \sgn{z}}(z) \intd z \\
T_{5}(k,g) & = \sqrt{\frac{2 \pi}{\beta}}\int_{-M \beta^{-r}}^{M \beta^{-r}} \sqrt{\frac{1}{\absx{h^{\prime \prime}\of{0_{\sgn{z}}}}}} \beta^2 z^{k+7} \frac{h^{\prime \prime \prime \prime }(\xi_z) h^{\prime \prime \prime }(0_{\sgn{z}})}{48}  e^{\xi_{(\beta, z)}^{2}} \phi_{\beta \sgn{z}}(z) \intd z \\
T_{6}(k,g) & = \sqrt{\frac{2 \pi}{\beta}}\int_{-M \beta^{-r}}^{M \beta^{-r}} \sqrt{\frac{1}{\absx{h^{\prime \prime}\of{0_{\sgn{z}}}}}} \beta^3 z^{k+10} \frac{h^{\prime \prime \prime \prime }(\xi_z) h^{\prime \prime \prime }(0_{\sgn{z}})^2}{144}  e^{\xi_{(\beta, z)}^{2}} e^{\xi_{(\beta, z)}^{1}} \phi_{\beta \sgn{z}}(z) \intd z 
\end{align*}

For the higher order terms, we observe the following, with $g_* = \max\of{\absx{g\of{x}}}$ :
\begin{align}
    \absx{T_{3}(k,g)} &= \abs{\sqrt{\frac{2 \pi}{\beta}}\int_{-M \beta^{-r}}^{M \beta^{-r}}  \sqrt{\frac{1}{\absx{h^{\prime \prime}\of{0_{\sgn{z}}}}}} g\of{z}\beta^2 z^{k+6} \frac{h^{\prime \prime \prime }(0_{\sgn{z}})^2}{72}  e^{\xi_{(\beta, z)}^{1}} \phi_{\beta \sgn{z}}(z) \intd z} \notag\\
    &\leq \sqrt{\frac{2 \pi}{\beta \absx{H}}}\beta^2 \frac{h_*^2 g_*}{72}  \int_{-M \beta^{-r}}^{M \beta^{-r}} \absx{z}^{k+6} \exp\of{\beta z^3 \frac{h_*}{6}} \phi_{\beta \sgn{z}}(z) \intd z \notag \\
    &\leq  \sqrt{\frac{2 \pi}{\beta \absx{H}}} \beta^2 \frac{h_*^2 g_*}{72} \exp\of{M^3 \frac{h_*}{6}} \int_\RR \absx{z}^{k+6} \phi_{\beta \sgn{z}}(z) \intd z \notag \\
    &=  C_3^k \beta^{-(k+3)/2} \int_\RR \absx{z}^{k+6} \cN\of{z;0,1} \intd z \label{m=3}
\end{align}
Following analog methodologies, one can show the existence of $C_4^k, C_5^k, C_6^k \ge 0$ such that :
\begin{align}
    \absx{T_{4}(k,g)} \leq C_4^k \beta^{-(k+3)/2} \int_\RR \absx{z}^{k+4} \cN\of{z;0,1} \intd z \label{m=4}\\
   \absx{T_{5}(k,g)} \leq C_5^k \beta^{-(k+4)/2} \int_\RR \absx{z}^{k+7} \cN\of{z;0,1} \intd z \label{m=5} \\ 
   \absx{T_{6}(k,g)} \leq C_6^k \beta^{-(k+5)/2} \int_\RR \absx{z}^{k+10} \cN\of{z;0,1} \intd z \label{m=6}
\end{align}
Hence, combining \eqref{m=3}, \eqref{m=4}, \eqref{m=5} and \eqref{m=6} shows that for any $m \in \{3,4,5,6\}$,
\begin{align}
    \absx{T_{m}(k,g)} \leq O\of{\frac{1}{\beta^{(k+3)/2}}} \label{t3}
\end{align}

For the two remaining terms, we have
\begin{align}
T_{1}(k,g)
&=
\sqrt{\frac{2\pi}{\beta}}
\int_{-M\beta^{-r}}^{M\beta^{-r}}
\sqrt{\frac{1}{\absx{h^{\prime\prime}\of{0_{\sgn{z}}}}}}
\, g\of{z} z^{k} \phi_{\beta\sgn{z}}(z)
\intd z
\notag \\
&=
\sqrt{\frac{\pi}{2\beta}}
\Biggl[
\frac{g\of{0_+}}{\sqrt{\absx{h^{\prime\prime}\of{0_+}}}}
\int_{\RR} \absx{z}^{k}\phi_{\beta +}(z)\intd z
\notag \\
&\qquad\qquad
+
\bp{-1}^{k}
\frac{g\of{0_-}}{\sqrt{\absx{h^{\prime\prime}\of{0_-}}}}
\int_{\RR} \absx{z}^{k}\phi_{\beta -}(z)\intd z
\Biggr]
+
R_1\of{k,g}
\notag \\
&=
\sqrt{\frac{\pi}{2\beta^{k+1}}}
\Biggl[
\frac{g\of{0_+}}
     {\sqrt{\absx{h^{\prime\prime}\of{0_+}}^{k+1}}}
+
\bp{-1}^{k}
\frac{g\of{0_-}}
     {\sqrt{\absx{h^{\prime\prime}\of{0_-}}^{k+1}}}
\Biggr]
\int_{\RR}\absx{z}^{k}\cN\of{z;0,1}\intd z
+
R_1\of{k,g}.
\label{T1}
\end{align}
Then, we have :
\begin{align}
T_{2}(k,g)
&=
\sqrt{\frac{2\pi}{\beta}}
\int_{-M\beta^{-r}}^{M\beta^{-r}}
\sqrt{\frac{1}{\absx{h^{\prime\prime}\of{0_{\sgn{z}}}}}}
\, g\of{z}\,
\beta z^{k+3}
\frac{h^{\prime\prime\prime}\of{0_{\sgn{z}}}}{6}
\phi_{\beta\sgn{z}}(z)
\intd z
\notag \\
&=
\sqrt{\frac{2\pi}{\beta}}
\Biggl[
\frac{
g\of{0_+}h^{\prime\prime\prime}\of{0_+}
}
{12\sqrt{\absx{h^{\prime\prime}\of{0_+}}}}
\int_{\RR}\beta\absx{z}^{k+3}\phi_{\beta+}(z)\intd z
\notag \\
&\qquad\qquad
+
\bp{-1}^{k+1}
\frac{
g\of{0_-}h^{\prime\prime\prime}\of{0_-}
}
{12\sqrt{\absx{h^{\prime\prime}\of{0_-}}}}
\int_{\RR}\beta\absx{z}^{k+3}\phi_{\beta-}(z)\intd z
\Biggr]
+
R_2\of{k,g}
\notag \\
&=
\sqrt{\frac{2\pi}{\beta^{k+2}}}
\Biggl[
\frac{
g\of{0_+}h^{\prime\prime\prime}\of{0_+}
}
{12\absx{h^{\prime\prime}\of{0_+}}^{(k+4)/2}}
+
\bp{-1}^{k+1}
\frac{
g\of{0_-}h^{\prime\prime\prime}\of{0_-}
}
{12\absx{h^{\prime\prime}\of{0_-}}^{(k+4)/2}}
\Biggr]
\notag \\
&\qquad\qquad
\times
\int_{\RR}\absx{z}^{k+3}\cN\of{z;0,1}\intd z
+
R_2\of{k,g}.
\label{T2}
\end{align}

By analogous reasoning to the proofs for propositions \cref{proposition:tail} and \cref{proposition:inter}, one can show that there exist $F_1, F_2 >0$ such that $R_1\of{k,g} = O\of{\beta^{(k+3)/2} e^{-F_1 \beta^{2/3}}}$, $R_2\of{k,g} = O\of{\beta^{(k+3)/2} e^{-F_2 \beta^{2/3}}}$  and trivially $R_1\of{k,g} + R_2\of{k,g} = O\of{\beta^{-(k+3)/2}}$.

Combining the results from \eqref{t3}, \eqref{T1},\eqref{T2}, we get the anticipated result.
\end{proof}

\clearpage
\subsubsection{Asymptotic Behaviour of the Logged Acceptance Ratio}
First, the previous technical lemma allows us to derive some crucial asymptotic quantities for the asymptotic behaviour.

\begin{corollary}[Moments]
Using the technical \Cref{lemma}  we derive the following quantities and their asymptotic behaviour,
\begin{align}
N\of{\beta}
&=
\sqrt{\frac{2\pi}{\beta\absx{H}}}
\Biggl[
1
+
\frac{
h^{\prime\prime\prime}\of{0_+}
-
h^{\prime\prime\prime}\of{0_-}
}{12}
\sqrt{\frac{8}{\pi\beta\absx{H}^3}}
+
O\of{\beta^{-1}}
\Biggr],
\label{N}
\\[0.5em]
\mu_\beta
&:=
\int_{\RR}
h^{\prime\prime\prime}\of{0_{\sgn{z}}}\,
z^3\,
\frac{f\of{z}^{\beta}}{N\of{\beta}}
\intd z
\notag \\
&=
\sqrt{\frac{8}{\pi}}\,
\frac{
h^{\prime\prime\prime}\of{0_+}
-
h^{\prime\prime\prime}\of{0_-}
}
{2(\beta\absx{H})^{3/2}}
\notag \\
&\quad
+
\frac{1}{12\beta^2\absx{H}^{3}}
\Biggl[
15
\Bigl(
h^{\prime\prime\prime}\of{0_+}^2
+
h^{\prime\prime\prime}\of{0_-}^2
\Bigr)
\notag \\
&\qquad\qquad
-
\frac{4}{\pi}
\Bigl(
h^{\prime\prime\prime}\of{0_+}
-
h^{\prime\prime\prime}\of{0_-}
\Bigr)^2
\Biggr]
+
O\of{\beta^{-5/2}},
\label{mu}
\\[0.5em]
\sigma_\beta
&:=
\int_{\RR}
h^{\prime\prime\prime}\of{0_{\sgn{z}}}^{2}\,
z^6\,
\frac{f\of{z}^{\beta}}{N\of{\beta}}
\intd z
\notag \\
&=
\frac{15}{(\beta\absx{H})^3}
\frac{
h^{\prime\prime\prime}\of{0_+}^{2}
+
h^{\prime\prime\prime}\of{0_-}^{2}
}{2}
+
O\of{\beta^{-7/2}}.
\label{sigma_beta}
\end{align}
\end{corollary}

Furthermore, we can derive the Taylor Expansion of $h$ in the neighbourhood of $0$ :
\begin{align}
    h\of{z} = - \frac{\absx{H}}{2} z^2 + \frac{h^{\prime \prime \prime}\of{0_{\sgn{z}}}}{6} z^3 + \frac{h^{\prime \prime \prime \prime}\of{0_{\sgn{z}}}}{24} z^4 + \frac{h^{\prime \prime \prime \prime \prime}\of{\xi_z}}{120} z^5 \label{Taylor}
\end{align}
where $0 < \absx{\xi_z} < \absx{z}$.

The new logged acceptance ratio \eqref{eq:loggedA} then gives us :
\begin{align}
    B\of{d} = \sum_{i=1}^d  \Biggl[&\frac{\beta}{6} \bp{y_i^3 h^{\prime \prime \prime}\of{0_{\sgn{y_i}}}- x_i^3 h^{\prime \prime \prime}\of{0_{\sgn{x_i}}}} + \frac{\beta}{24}\bp{y_i^4 h^{\prime \prime \prime \prime}\of{0_{\sgn{y_i}}}- x_i^4 h^{\prime \prime \prime \prime}\of{0_{\sgn{x_i}}}} \notag \\
    &+ \frac{\beta}{120} \bp{y_i^5 h^{\prime \prime \prime \prime \prime}\of{\xi_{y_i}}- x_i^5  h^{\prime \prime \prime \prime \prime}\of{\xi_{x_i}}}\Biggr] \label{taylor_expansion_acceptance}
\end{align}

We decompose this sum in three distinct terms, according to the magnitude of each component :
\begin{align}
    W\of{d} &:= \sum_{i=1}^d \frac{\beta}{6} \bp{y_i^3 h^{\prime \prime \prime}\of{0_{\sgn{y_i}}}- x_i^3 h^{\prime \prime \prime}\of{0_{\sgn{x_i}}}} \label{w} \\
    T_1\of{d} &:= \sum_{i=1}^d \frac{\beta}{24}\bp{y_i^4 h^{\prime \prime \prime \prime}\of{0_{\sgn{y_i}}}- x_i^4 h^{\prime \prime \prime \prime}\of{0_{\sgn{x_i}}}} \label{t1} \\
    T_2\of{d} &:= \sum_{i=1}^d \frac{\beta}{120} \bp{y_i^5 h^{\prime \prime \prime \prime \prime}\of{\xi_{y_i}}- x_i^5  h^{\prime \prime \prime \prime \prime}\of{\xi_{x_i}}} \label{t2}
\end{align}
Let's set 
\begin{align}
        \sigma^2 = \frac{15 \bp{h^{\prime \prime \prime}\of{0_+}^2 + h^{\prime \prime \prime}\of{0_-}^2} - \frac{4}{\pi} \bp{h^{\prime \prime \prime}\of{0_+}- h^{\prime \prime \prime}\of{0_-}}^2}{36 \ell \absx{H}^3} \label{sigma}
\end{align}
Then, we have the three following lemmas to characterize the asymptotic behaviour of the terms we just introduced, which we will prove one by one,
\begin{lemma}
    \label{lemmaW}
    With $f(\cdot)$ as in Theorem~\ref{Thm:scalingRelaxed}, 
    \begin{align}
         W\of{d} \overset{d \to \infty}{\Longrightarrow} \cN\of{- \frac{\sigma^2}{2}, \sigma^2} \label{asymptotic_W}
    \end{align}
\end{lemma}

\begin{lemma}
    \label{lemmaT1}
    With $f(\cdot)$ as in Theorem~\ref{Thm:scalingRelaxed}, 
    \begin{align}
    T_1\of{d}  \overset{\PP}{\longrightarrow} 0 \label{prob_conv_1}
\end{align}
\end{lemma}
\begin{lemma}
    \label{lemmaT2}
    With $f(\cdot)$ as in Theorem~\ref{Thm:scalingRelaxed},  
    \begin{align}
    T_2\of{d}  \overset{\PP}{\longrightarrow} 0 \label{prob_conv_2}
\end{align}
\end{lemma}

\begin{proof}[Proof of \cref{lemmaW}]
     We divide the proof between the current state $\mathbf x$ and the proposed state $\mathbf y$.
\subsubsection*{Behaviour of the $x_i$'s}
    Define the following collection of random variables for $d \in \NN$ and $1\leq i \leq d$ :
\begin{align*}
    X_{\bp{d,i}} := \frac{\beta}{6} \bp{x_i^3 h^{\prime\prime\prime}\of{0_{\sgn{x_i}}} - \mu_\beta}
\end{align*}
\textbf{Lindeberg-Feller Conditions}: To apply the result of the Lindeberg-Feller theorem, the Central Limit Theorem for triangular arrays, to the $X_{(d, i)}$ 's, we need to prove that the following two conditions hold as $d \rightarrow \infty$

\begin{enumerate}[label=\roman*)]
    \item There exists $\sigma^{2} \in \mathbb{R}_{+}$such that

\begin{align}
\sum_{i=1}^{d} \mathbb{E}\left(X_{(d, i)}^{2}\right) \rightarrow \sigma^{2}>0 \label{variance}
\end{align}
\item For all $\varepsilon>0$ and with $A_{\varepsilon}^{d}:=\left\{\left|X_{(d, i)}\right|>\varepsilon\right\}$

\begin{align}
\lim _{d \rightarrow \infty} \sum_{i=1}^{d} \EE{\absx{X_{(d, i)}}^{2} \1{A_{\varepsilon}^d}}=0 \label{lindeberg}
\end{align}
\end{enumerate}

The first L-F condition (i) is much easier to show as, we have 
\begin{align*}
    \EE{X_{\bp{d,i}}^2} &= \frac{\beta^2}{36} \bp{\sigma_\beta -\mu_\beta^2} \\
    &= \frac{\beta^2}{72} \frac{1}{\bp{\beta \absx{H}}^3} \bs{ 15 \bp{h^{\prime \prime \prime}\of{0^+}^2 + h^{\prime \prime \prime}\of{0^-}^2} - \frac{4}{\pi} \bp{h^{\prime \prime \prime}\of{0^+} - h^{\prime \prime \prime}\of{0^-}}^2  } + O\of{\beta^{-3/2}} \\
    &= \frac{1}{72\beta \absx{H}^3} \bs{ 15 \bp{h^{\prime \prime \prime}\of{0^+}^2 + h^{\prime \prime \prime}\of{0^-}^2} - \frac{4}{\pi} \bp{h^{\prime \prime \prime}\of{0^+} - h^{\prime \prime \prime}\of{0^-}}^2  } + O\of{\beta^{-3/2}}
\end{align*}
Now recall that $\beta = l d$, it is immediate that 
\begin{align}
     \sum_{i=1}^d \EE{X_{\bp{d,i}}^2} \underset{d \rightarrow +\infty}{\rightarrow} \frac{1}{72 l \absx{H}^3} \bs{ 15 \bp{h^{\prime \prime \prime}\of{0^+}^2 + h^{\prime \prime \prime}\of{0^-}^2} - \frac{4}{\pi} \bp{h^{\prime \prime \prime}\of{0^+} - h^{\prime \prime \prime}\of{0^-}}^2  }
\end{align}
For the second condition, let $\varepsilon > 0$, then we have :
\begin{align}
    \EE{\absx{X_{\bp{d,i}}}^2 \1{A_\varepsilon^d}} &= \frac{\beta^2}{36} \bs{\EE{x_i^6 h^{\prime\prime\prime}\of{0_{\sgn{x_i}}}^2 \1{A_\varepsilon^d}} -  2 \mu_\beta \EE{x_i^3 h^{\prime\prime\prime}\of{0_{\sgn{x_i}}} \1{A_\varepsilon^d}} + \mu_\beta^2 \PP\of{A_\varepsilon^d}} \label{LF}
\end{align}
First, we note that the set $A_\varepsilon^d$ can be divided as 
\begin{align}
    A_\varepsilon^d& = \bc{ x_i^3 h^{\prime\prime\prime}\of{0_{\sgn{x_i}}} > \mu_\beta + \frac{6 \varepsilon}{\beta}} \cup \bc{ x_i^3 h^{\prime\prime\prime}\of{0_{\sgn{x_i}}} < \mu_\beta - \frac{6 \varepsilon}{\beta}} \notag \\
    &= F_{+\varepsilon}^d \cup F_{-\varepsilon}^d \label{set_decompositon}
\end{align}
Knowing that $\mu_\beta = O\of{d^{-3/2}}$, for $d$ sufficiently large, there exists constants $J_1, J_2 >0$ such that :
\begin{align}
    F_{+\varepsilon}^d &= \bc{ x_i^3 h^{\prime\prime\prime}\of{0_{\sgn{x_i}}} > \mu_\beta + \frac{6 \varepsilon}{\beta}} \notag
    \\ &\subseteq \bc{x_i^3  h^{\prime\prime\prime}\of{0_{\sgn{x_i}}} > \frac{J_1}{d}} \notag
    \\ &\subseteq \bc{\absx{x_i}^3 \absx{ h^{\prime\prime\prime}\of{0_{\sgn{x_i}}}}> \frac{J_1}{d}} \notag
    \\ &\subseteq \bc{\absx{x_i} > \bp{\frac{J_1}{h_*d}}^{1/3}} \label{f+}
    \\ F_{-\varepsilon}^d &= \bc{ x_i^3  h^{\prime\prime\prime}\of{0_{\sgn{x_i}}} < \mu_\beta - \frac{6 \varepsilon}{\beta}} \notag \\
    &\subseteq \bc{x_i^3   h^{\prime\prime\prime}\of{0_{\sgn{x_i}}} < - \frac{J_2}{d}} \notag \\
    &\subseteq \bc{ - x_i^3   h^{\prime\prime\prime}\of{0_{\sgn{x_i}}} >  \frac{J_2}{d}} \notag \\
    &\subseteq \bc{\absx{x_i}^3 \absx{ h^{\prime\prime\prime}\of{0_{\sgn{x_i}}}} > \frac{J_2}{d}} \notag \\
    &\subseteq \bc{\absx{x_i} > \bp{\frac{J_2}{h_*d}}^{1/3}} \label{f-}
\end{align}
Now, by combining \eqref{set_decompositon}, \eqref{f+} and \eqref{f-}, as well as setting $J = \bp{\frac{\min\of{J_1,J_2}}{h_*}}^{1/3}$, we get :
\begin{align}
    A_\varepsilon^d &\subseteq \bc{\absx{x_i} > \bp{\frac{J_1}{h_*d}}^{1/3}} \cup \bc{\absx{x_i} > \bp{\frac{J_2}{h_*d}}^{1/3}} \notag \\
    &\subseteq \bc{\absx{x_i} > \bp{\frac{J}{d^{1/3}}}}
\end{align}
Hence, we have, with $M$ as defined in \eqref{eq:polyT}, $g^p\in \mathcal S_0$ such that $g^p\of{z} = h^{\prime \prime \prime}\of{0_{\sgn{z}}}^p$, and $p \in \{0,1,2\}$
\begin{align}
\EE{x_i^k g^p\of{x_i}\1{A_{\varepsilon}^d}}
&\leq
\frac{1}{N\of{\beta}}
\Bigl[
I_k^{g^p}\of{Jd^{-1/3}, M}
+
I_k^{g^p}\of{M, +\infty}
\notag \\
&\qquad\qquad
+
I_k^{g^p}\of{-M,-Jd^{-1/3}}
+
I_k^{g^p}\of{-\infty,-M}
\Bigr].
\end{align}
Then, by propositions \cref{proposition:tail} and \cref{proposition:inter} there exist $W_1, W_2, W_3 > 0$ such that 
\begin{align*}
    \EE{x_i^6 h^{\prime\prime\prime}\of{0_{\sgn{x_i}}}^2 \1{A_{\varepsilon}^d} } &= O\of{e^{-W_1 d^{1/3}}} \\
    \mu_\beta \EE{x_i^3 h^{\prime\prime\prime}\of{0_{\sgn{x_i}}} \1{A_{\varepsilon}^d} } &= O\of{e^{-W_2 d^{1/3}}} \\
    \mu_\beta^2 \EE{\1{A_{\varepsilon}^d} } &= O\of{e^{-W_3 d^{1/3}}}
\end{align*}
Setting $W = \min\of{W_1, W_2, W_3}$, we have :
\begin{align}
    \EE{x_i^6 h^{\prime\prime\prime}\of{\xi_{x_i}}^2 \1{A_\varepsilon^d}} -  2 \mu_\beta \EE{x_i^3 h^{\prime\prime\prime}\of{\xi_{x_i}} \1{A_\varepsilon^d}} + \mu_\beta^2 \PP\of{A_\varepsilon^d} = O\of{e^{-W d^{1/3}}} \label{LF_asymptotic}
\end{align}
Consequently, by \eqref{LF} and \eqref{LF_asymptotic}, we have :
\begin{align}
    \lim _{d \rightarrow \infty} \sum_{i=1}^{d} \mathbb{E}\left(\left|X_{(d, i)}\right|^{2} \1{A_{\varepsilon}^d}\right)=0 
\end{align}
Thus, both Lindeberg-Feller conditions hold and the conclusion follows : 
\begin{align}
    - \sum_{i=1}^d X_{\bp{d,i}} \Rightarrow \cN\of{0, \frac{1}{72 l \absx{H}^3} \bs{ 15 \bp{h^{\prime \prime \prime}\of{0^+}^2 + h^{\prime \prime \prime}\of{0^-}^2} - \frac{4}{\pi} \bp{h^{\prime \prime \prime}\of{0^+} - h^{\prime \prime \prime}\of{0^-}}^2  }} \label{weak1}
\end{align}

\subsubsection*{Behaviour of the $y_i$'s}
This time, we consider the following collection $\bp{Y_{\bp{d,i}}}_{1\leq i \leq d}$ of random variables for $d \in \NN^*$, such that
\begin{align*}
    Y_{\bp{d,i}} := \frac{\beta}{6} \bp{y_i^3 h^{\prime\prime\prime}\of{0_{\sgn{y_i}}} - \mu_y}
\end{align*}
where, from standard computations using the known and standard distribution of the $y_i$'s, 
\begin{align}
    \mu_y := \EE{y_i^3 h^{\prime\prime\prime}\of{0_{\sgn{y_i}}}} &= \sqrt{\frac{8}{\pi}}\sqrt{\frac{1}{\bp{\beta \absx{H}}^{3}}} \frac{h^{\prime \prime \prime}\of{0_+} - h^{\prime \prime \prime}\of{0_-}}{2}  \\
    \sigma_y := \EE{y_i^6h^{\prime\prime\prime}\of{0_{\sgn{y_i}}}^2} &= 15 \frac{1}{\bp{\beta \absx{H}}^{3}} \frac{h^{\prime \prime \prime}\of{0_+}^2 + h^{\prime \prime \prime}\of{0_-}^2}{2} 
\end{align}
one can follow an analogous procedure to comply with the Lindeberg-Feller theorem requirements and show that:
\begin{align}
    \sum_{i=1}^d Y_{\bp{d,i}} \Rightarrow \cN\of{0, \frac{1}{72 l \absx{H}^3} \bs{ 15 \bp{h^{\prime \prime \prime}\of{0^+}^2 + h^{\prime \prime \prime}\of{0^-}^2} - \frac{4}{\pi} \bp{h^{\prime \prime \prime}\of{0^+} - h^{\prime \prime \prime}\of{0^-}}^2  }} \label{weak2}
\end{align}
Using the fact that the $x_{i}$ 's and $y_{i}$ 's are independent along with the weak convergences established in \eqref{weak1} and \eqref{weak2}, we get 
\begin{align}
\sum_{i=1}^d \bs{Y_{\bp{d,i}} - X_{\bp{d,i}}}
&\Rightarrow
\cN\Biggl(
0,
\frac{1}{36l\absx{H}^3}
\Biggl[
15\bp{
h^{\prime\prime\prime}\of{0^+}^2
+
h^{\prime\prime\prime}\of{0^-}^2
}
\notag \\
&\qquad\qquad\qquad
-
\frac{4}{\pi}
\bp{
h^{\prime\prime\prime}\of{0^+}
-
h^{\prime\prime\prime}\of{0^-}
}^2
\Biggr]
\Biggr).
\label{weak3}
\end{align}
Then, one can observe that :
\begin{align}
\sum_{i=1}^d \bs{Y_{\bp{d,i}} - X_{\bp{d,i}}}
&=
\frac{\beta}{6}
\sum_{i=1}^d
\Bigl[
y_i^3 h^{\prime\prime\prime}\of{0_{\sgn{y_i}}}
-
x_i^3 h^{\prime\prime\prime}\of{0_{\sgn{x_i}}}
\notag \\
&\qquad\qquad
-
\EE{
y_i^3 h^{\prime\prime\prime}\of{0_{\sgn{y_i}}}
-
x_i^3 h^{\prime\prime\prime}\of{0_{\sgn{x_i}}}
}
\Bigr]
\notag \\
&=
\frac{\beta}{6}
\sum_{i=1}^d
\Bigl[
y_i^3 h^{\prime\prime\prime}\of{0_{\sgn{y_i}}}
-
x_i^3 h^{\prime\prime\prime}\of{0_{\sgn{x_i}}}
\Bigr]
\notag \\
&\quad
-
\frac{\ell d^2}{6}
\EE{
y_i^3 h^{\prime\prime\prime}\of{0_{\sgn{y_i}}}
-
x_i^3 h^{\prime\prime\prime}\of{0_{\sgn{x_i}}}
}
\notag \\
&=
B\of{d}
-
\frac{\ell d^2}{6}
\EE{
y_i^3 h^{\prime\prime\prime}\of{0_{\sgn{y_i}}}
-
x_i^3 h^{\prime\prime\prime}\of{0_{\sgn{x_i}}}
}.
\label{back_to_acceptance}
\end{align}
Now, using the asymptotic formula recalled in \eqref{mu},
\begin{align}
\EE{
y_i^3 h^{\prime\prime\prime}\of{0_{\sgn{y_i}}}
-
x_i^3 h^{\prime\prime\prime}\of{0_{\sgn{x_i}}}
}
&=
-
\frac{1}{12\beta^2\absx{H}^{3}}
\Biggl[
15\bp{
h^{\prime\prime\prime}\of{0_+}^2
+
h^{\prime\prime\prime}\of{0_-}^2
}
\notag \\
&\qquad\qquad
-
\frac{4}{\pi}
\bp{
h^{\prime\prime\prime}\of{0_+}
-
h^{\prime\prime\prime}\of{0_-}
}^2
\Biggr]
+
O\of{\beta^{-5/2}}.
\label{expectation}
\end{align}
Finally, by plugging \eqref{expectation} into \eqref{back_to_acceptance} and using \eqref{weak3}, we show :
\begin{align}
    W\of{d} \Rightarrow \cN\of{- \frac{\sigma^2}{2}, \sigma^2}
\end{align}
with 
\begin{align}
    \sigma^2 = \frac{15 \bp{h^{\prime \prime \prime}\of{0_+}^2 + h^{\prime \prime \prime}\of{0_-}^2} - \frac{4}{\pi} \bp{h^{\prime \prime \prime}\of{0_+}- h^{\prime \prime \prime}\of{0_-}}^2}{36 l \absx{H}^3} 
\end{align}
\end{proof}

\begin{proof}[Proof of \cref{lemmaT1}]
    Recall that we want to study the asymptotic behaviour of the following quantity \eqref{t1}:
\begin{align*}
    T_1\of{d} &:= \sum_{i=1}^d \frac{\beta}{24}\bp{y_i^4 h^{\prime \prime \prime \prime}\of{0_{\sgn{y_i}}}- x_i^4 h^{\prime \prime \prime \prime}\of{0_{\sgn{x_i}}}}
\end{align*}
Through our findings of \cref{lemma}, we know:
\begin{align}
    \EE{ x_i^4 h^{\prime \prime \prime \prime}\of{0_{\sgn{x_i}}}} &= 3 \frac{h^{\prime \prime \prime \prime}\of{0_+} + h^{\prime \prime \prime \prime}\of{0_-}}{2 \beta^2 \absx{H}^2} + O\of{\beta^{-5/2}} \label{exp_x_4} \\
    \EE{ x_i^8 h^{\prime \prime \prime \prime}\of{0_{\sgn{x_i}}}^2} &= 105 \frac{h^{\prime \prime \prime \prime}\of{0_+}^2 + h^{\prime \prime \prime \prime}\of{0_-}^2}{2 \beta^4 \absx{H}^4} + O\of{\beta^{-9/2}} \label{exp_x_8}
\end{align}
and using the fact that $y_i$'s follow a Gaussian distribution,
\begin{align}
    \EE{ y_i^4 h^{\prime \prime \prime \prime}\of{0_{\sgn{y_i}}}} &= 3 \frac{h^{\prime \prime \prime \prime}\of{0_+} + h^{\prime \prime \prime \prime}\of{0_-}}{2 \beta^2 \absx{H}^2}  \label{exp_y_4} \\
    \EE{ y_i^8 h^{\prime \prime \prime \prime}\of{0_{\sgn{y_i}}}^2} &= 105 \frac{h^{\prime \prime \prime \prime}\of{0_+}^2 + h^{\prime \prime \prime \prime}\of{0_-}^2}{2 \beta^4 \absx{H}^4}  \label{exp_y_8}
\end{align}

Defining $\mu_{T_{1}(d)}:=\mathbb{E}\left(T_{1}(d)\right)$ then by \eqref{exp_x_4} and \eqref{exp_y_4}, as $d \rightarrow \infty$

\begin{align}
\mu_{T_{1}(d)} \rightarrow 0 \label{limit_exp}
\end{align}

Thus there exists $U>0$ such that for all $d>U$ then $\left|\mu_{T_{1}(d)}\right|<\frac{\varepsilon}{2}$. Furthermore, by \eqref{exp_x_4}, \eqref{exp_x_8}, \eqref{exp_y_4} and \eqref{exp_y_8} then

\begin{align}
\operatorname{Var}\left(T_{1}(d)\right) \rightarrow 0 \text { as } d \rightarrow 0 \label{var}
\end{align}

For $d>U$ and using Chebyshev's inequality and utilising \eqref{variance}

$$
\mathbb{P}\left(\left|T_{1}(d)\right|>\varepsilon\right) \leq \mathbb{P}\left(\left|T_{1}(d)-\mu_{T_{1}(d)}\right| \geq \frac{\varepsilon}{2}\right) \leq \frac{4 \operatorname{Var}\left(T_{1}(d)\right)}{\varepsilon^{2}} \rightarrow 0 \quad \text { as } \quad d \rightarrow 0
$$
Hence,

\begin{align}
    T_1\of{d}  \overset{\PP}{\longrightarrow} 0 \label{prob_conv_1_proof}
\end{align}
\end{proof}

\begin{proof}[Proof of \cref{lemmaT2}]
    Recall that we want to study the asymptotic behaviour of the following quantity \eqref{t2} :
\begin{align*}
    T_2\of{d} &:= \sum_{i=1}^d \frac{\beta}{120} \bp{y_i^5 h^{\prime \prime \prime \prime \prime}\of{\xi_{y_i}}- x_i^5  h^{\prime \prime \prime \prime \prime}\of{\xi_{x_i}}}
\end{align*}
To this end, using the assumption given in \eqref{eq:fextendedbycontinuity}
\begin{align}
\left|T_{2}(d)\right| \leq \sum_{i=1}^{d} L \frac{\beta}{120}\left(\left|y_{i}\right|^{5}+\left|x_{i}\right|^{5}\right)=: T_{2}^{r}(d)>0 \label{t2bis}
\end{align}

Once again using the conclusions of \cref{lemma}

\begin{align}
    \EE{\absx{x_i}^5} = 8 \sqrt{\frac{2}{\pi \bp{\beta \absx{H}}^5}} + O\of{\beta^{-3}} \label{exp_x_5}
\end{align}
and through the Gaussian distributed $y_i$'s
\begin{align}
    \EE{\absx{y_i}^5} = 8 \sqrt{\frac{2}{\pi \bp{\beta \absx{H}}^5}} \label{exp_y_5}
\end{align}

Let $\varepsilon>0$, then by Markov inequality and then using \eqref{exp_x_5} and \eqref{exp_y_5}

\begin{align*}
\mathbb{P}\left(T_{r}(d)>\varepsilon\right) & \leq \frac{L \ell d}{24 \varepsilon} \sum_{i=1}^{d}\left[\mathbb{E}\left(\left|y_{i}\right|^{5}\right)+\mathbb{E}\left(\left|x_{i}\right|^{5}\right)\right] \\
& =\frac{L \ell d^{2}}{24 \varepsilon}\left[\mathbb{E}\left(\left|y_{1}\right|^{5}\right)+\mathbb{E}\left(\left|x_{1}\right|^{5}\right)\right] \\
& = \frac{L}{3\varepsilon} \sqrt{\frac{2}{\pi \ell^3 d \absx{H}^5}} \\
\underset{d \rightarrow \infty}{\rightarrow} 0
\end{align*}
Then using \eqref{t2bis}, we have that 
\begin{align}
    T_2\of{d} \overset{\PP}{\longrightarrow} 0 \label{prob_conv_2_proof}
\end{align}
\end{proof}

\subsubsection{Deriving the final result}
\begin{proof}[Proof of Theorem~\ref{Thm:scalingRelaxed}]
Merging the findings from \cref{lemmaW,lemmaT1,lemmaT2} and applying Slutsky's Theorem gives us asymptotically, when $d \to \infty$,
\begin{align}
    B\of{d} = W\of{d} + T_1\of{d} + T_2\of{d} \Rightarrow B \sim \cN\of{- \frac{\sigma^2}{2}, \sigma^2} \label{conv}
\end{align}
Leveraging the convergence in distribution and $z \mapsto \min\of{1, e^z}$ being a bounded continuous function we have,
\begin{align}
    \lim_{d \to \infty } \EE{\min\of{1, e^{B\of{d}}}} = \EE{\min\of{1, e^{B}}}
\end{align}
\begin{proposition}
    \label{proposition:sigma}
    Take $Z \sim \cN\of{\mu, \sigma^2}$, then
    \begin{align}
        \EE{\min\of{1, e^{Z}}} = \mathbf \Phi\of{\frac{\mu}{\sigma}} + \exp\of{\mu + \sigma^2 /2} \mathbf \Phi\of{-\sigma - \frac{\mu}{\sigma}}
    \end{align}
\end{proposition}
\begin{proof}
    Simple computation \citep[][see, e.g.,]{roberts1997weak}.
\end{proof}
Using \cref{proposition:sigma} with the result established in \cref{conv} and with $\sigma$ as in \eqref{sigma}, we obtain the final asymptotic result presented in Theorem~\ref{Thm:scalingRelaxed}.
\end{proof}

\clearpage

\section{The HAT distributions}
\subsection{Weight-Stabilised}
\label{sec:Weightstabilised}
In Section~\ref{sec:3.1} we introduced the HAT target distribution that we use as the target distributions for the annealed densities in the PA component of ALPS. This section explains why (at least to first order) these approximately preserve the modal weight upon annealing.

Much of this follows the motivation and arguments used in \cite{GarethJeffNick}. We start by introducing a specific mixture distribution where all components consist of a scale and location transformation of some common density.

 \begin{definition}[Componentwise Rescaled Mixture Distribution (CRMD)] \label{def:cwrd}
 
 $\forall j \in \{1,\ldots,J\}$  let $\mu_j \in \mathbb{R}^d$ and $\Sigma_j \in \mathbb{R}^{d\times d}$ be  positive definite matrices.
 
Define a component-wise rescaled mixture distribution to be a mixture distribution with $J$ components but where the component distributions are homogeneous up to a scale and location reparametrisation. Specifically, for $J \in \mathbb{N}$, $s_j(x):= \Sigma_j^{-1/2}(x-\mu_j)$ and for $g(\cdot)$, a $d$-dimensional density function with a global maximum at 0,    the CRMD is defined as
\begin{equation}
	\gamma(x) \propto \sum_{j=1}^J a_j(x) = \sum_{j=1}^J w_j |\Sigma_j|^{-1/2}g\left(s_j(x))\right)
\label{eq:themixer}
\end{equation}
where $\forall j \in \{1,\ldots,J\}$ $\mu_j \in \mathbb{R}^d$ and $\Sigma_j \in \mathbb{R}^{d\times d}$ is a positive definite matrix.
\end{definition}

 In what follows we make the added assumption that the modes are well-separated which in this context means that there is sufficient distance between the location parameters, the $\mu_j$'s, such that there is negligible tail-overlap from the components. This is more reasonable assumption in high dimensional settings or indeed in cases where the distribution becomes annealed.
 
Upon annealing/tempering, either using some local or global scheme, then due to the assumed large separation the resulting distribution will still resemble a mixture distribution as described in Definition~\ref{def:CTRMD}.

 \begin{definition}[Componentwise Tempered Rescaled Mixture Distribution (CTRMD)] \label{def:CTRMD}
Under the setting of Definition~\ref{def:cwrd}, the component-wise tempered rescaled mixture distribution at inverse-temperature level $\beta$, $\gamma_\beta$,  is defined as the CRMD given by
\begin{equation}
	\gamma_\beta(x) \propto \sum_{j=1}^J b_j(x,\beta) = \sum_{j=1}^J W_{(j,\beta)} \frac{\left[g\left(s_j(x)\right)\right]^\beta}{\int\left[g\left(s_j(u)\right)\right]^\beta du}. \label{eq:themixerbeta}
\end{equation}
where $s_j(x):= \Sigma_j^{-1/2}(x-\mu_j)$.
\end{definition}

We know that if traditional tempering is used then the effective weights of the components become increasingly inconsistent with the original mixture. The following proposition shows what happens under traditional tempering vs a scheme that re-weights the component according to its (local) maxima. 

\begin{proposition}[CTRMD Equivalences] 
Consider the setting of \eqref{eq:themixer} and \eqref{eq:themixerbeta}. Then $\forall j \in 1,\ldots, J$
\begin{itemize}
\item Standard power-tempering\\
If $b_j(x,\beta) = [a_j(x)]^\beta$ then 
	\begin{equation}
		W_{(j,\beta)} \propto w_j^\beta |\Sigma_j|^{\frac{1-\beta}{2}}.
		\nonumber
	\end{equation}
\item Weight-preserving tempering\\If 
\begin{equation}
	b_j(x,\beta)= a_j(x)^\beta a_j(\mu_j)^{(1-\beta)} \nonumber
\end{equation}
then  $W_{(j,\beta)} \propto w_j$.
\end{itemize}
\label{Theorem:equivalence}
\end{proposition}

\begin{proof}
The proof of this result is a routine extension of the proof given for the specific case of Gaussian mixture targets given in \cite{GarethJeffNick}.
\end{proof}

The significance of this is that in cases where the target distribution can be approximated by a mixture of the form given in \eqref{eq:themixer} then the HAT densities prevent the severe weight degeneracy issues that arise from traditional power-tempering methods.

\subsection{Truncated HAT distributions}
\label{sec:Truncated}

In Section~\ref{sec:SURexample}, the ALPS procedure was used to sample from a posterior distribution that resulted from an SUR model. The condition number of the Hessian evaluations showed that the problem was unstable. The issue is that the problem is ill-posed in that the parameters are essentially unidentifiable. The posterior distribution seems to exhibit narrow ridges with tail-structures in those directions appearing to be significantly heavier than Gaussian.

As such there would need to be significant localised annealing before the target distribution along the ridge directions began to appear sufficiently Gaussian so that the independence sampler could provide reasonable inter-modal jump rates. However, this led to issues with machine precision. To overcome this issue, an ad-hoc adjustment to the HAT targets was proposed. The idea was to introduce a truncated version of the annealed HAT targets that would be exactly as described in Definition~\ref{def:HAT} if the location was within a prescribed distance from the mode and would be set to zero outside of that. The idea being to eliminate some of the instability of the heavy tailed ridge directions at the annealed temperature levels so that the algorithm works without using annealed temperature levels that would create computational problems due to machine precision issues.

\begin{definition}[Truncated HAT distributions at quantile $q$]
\label{def:THAT}
Consider the definition of the HAT distribution at inverse temperature level $\beta$ from Definition~\ref{def:HAT}. Then using the same notation as Definition~\ref{def:HAT}, and letting $A= A_{x,\beta}$ for ease of notation, set
\begin{equation}
	Q_{x,\beta} := \mathbbm{1}_{\{(x-\mu_A)^T \Sigma_A^{-1} (x-\mu_A)<q\}}
\label{eq:quanif}
\end{equation}

The truncated HAT distribution at inverse temperature level $\beta$ and quantile $q$ is then defined as:
\begin{equation}
	\pi_{\beta}(x)  \propto
  \begin{cases}
    \pi(x)^{\beta} \pi(\mu_{A_{x,\beta}})^{1-\beta} Q_{x,\beta} & \text{if $A_{x,\beta}=A_{x,1} $} \\
    G(x,\beta)Q_{x,\beta} & \text{if $A_{x,\beta} \neq A_{x,1} $}
  \end{cases}  \label{eq:robadj2}
\end{equation}
where,
\[
G(x,\beta)= \frac{\pi(\mu_{A})\left((2\pi)^{d}\absx{\Sigma_{A}}\right)^{1/2}\phi \left(x|\mu_{A},\frac{\Sigma_{A}}{\beta}\right)}{\beta^{d/2}}.
\] 
\end{definition}

The user would choose $q$ as some pre-specified quantile of a Chi-squared distribution with degrees of freedom $d$-the dimensionality of the problem. The intuition is that the term in the indicator for the object $Q_{x,\beta}$ in Definition~\ref{def:THAT} would be Chi-squared $d$ distributed if the local mode was Gaussian with mean $\mu_A$ and covariance matrix $\Sigma_A$. Therefore, the distribution would eliminate the distribution for values inconsistent with such a null hypothesis. 

This overcame the issues for the  SUR model but there is further work needed to explore if this ad-hoc solution is robust for similarly computationally unstable problems.

\clearpage

\section{Pseudo-code implementation of ALPS}
\label{sec:PseudoCode}
\begin{algorithm}[H]
\caption{Annealed Leap-Point Sampler (ALPS)}\label{alg:ALPS}
\begin{algorithmic}[1]

\Require Temperature schedule $\Delta = \{\beta_0,\beta_1,\dots,\beta_n\}$
\Require Initial mode collections $\{\hat{W},\hat{M},\hat{C}\}$
\Require Number of temperature swaps $s$
\Require Initial sample $X^0 = (x_0^0,\dots,x_n^0)$
\Require Total number of samples $T$

\vspace{0.25em}

\For{$t = 0$ to $T-1$}

    \Statex \textbf{Within-temperature moves:}

    \For{$j = 0, \dots, n-1$}
        \State Sample $x_j^{t+1} \sim P_{\beta_j}(x_j^t, \cdot)$
        \Statex \hspace{2em}\textit{(local exploration at intermediate temperatures; see \eqref{sec:3.1}, \eqref{sec:3.4})}
    \EndFor

    \vspace{0.25em}

    \State Sample $x_n^{t+1} \sim Q(x_n^t,\cdot)$
    \Statex \hspace{2em}\textit{(global exploration at the coldest temperature $\beta_n$; see \eqref{sec:3.5})}

    \vspace{0.25em}

    \Statex \textbf{Temperature swap moves:}

    \For{$i = 1, \dots, s$}
        \State Perform QuanTA-aided temperature swap
        \Statex \hspace{2em}\textit{(see \eqref{sec:3.3})}
    \EndFor

    \vspace{0.25em}

    \State Set $X^{t+1} \leftarrow (x_0^{t+1}, x_1^{t+1}, \ldots, x_n^{t+1})$

    \vspace{0.25em}
    \Statex \textbf{(Optional) Update mode collections:}
    \State Update $\{\hat{W},\hat{M},\hat{C}\}$ through online running of the Exploration Component
    \Statex \hspace{2em}\textit{(see \eqref{sec:3.6})}

\EndFor

\end{algorithmic}
\end{algorithm}

\clearpage

\section{Example 4: 4-Modal 20-Dimensional Skew-Normal}
\label{sec:Ex3}

The target distribution is given by a mixture of four evenly-weighted and well-separated twenty-dimensional skew-normal distributions with heterogeneous scaling of the modes and homogeneous skew:
\begin{equation*}
    \pi(x) \propto \sum_{k=1}^4 \prod_{j=1}^{20}  \frac{2}{\omega_k}\phi\left( \frac{x_j - (\mu_{k})_j }{\omega_k}\right)\Phi\left( \alpha\frac{x_j - (\mu_{k})_j }{\omega_k}\right)
\end{equation*}
where $\alpha =10$, $\mu_1 = (20,20,\ldots,20) = -\mu_2$, $\mu_3=(-10-10,\ldots,-10,10,10,\ldots,10)=-\mu_4$, $\omega_1=\omega_2 = 1$ and $\omega_3=\omega_4=2$ and where $\phi(\cdot)$ and $\Phi(\cdot)$ denote the density and CDF of a standard normal distribution. This is similar to Example 1 in Section~\ref{sec:Ex1}, but with 4 modes instead of 13.

We compare the performance of ALPS against PT and LAIS. LAIS only uses a single temperature level which is the target level. For ALPS, 7 temperature levels were used for the annealing part of the algorithm with a cooling schedule given by powers in the vector  $(1  ,  4  , 16  , 64 , 256, 1024 , 4096)$. In contrast, PT required 14 temperature levels which were geometrically spaced with common ratio of 0.6. This gave the suggested optimal temperature swap rates of 0.234 between consecutive temperature  level pairs along with a temperature that appeared hot enough to explore the entire state space. Even so, the PT algorithm fails in this example rather spectacularly because the hot state target distribution is significantly inconsistent with the target distribution.

In all three algorithms random walk Metropolis was utilised for the within temperature moves; in the case of ALPS and LAIS then the random walk Metropolis moves used pre-conditioned  covariance structure informed by the covariance matrix estimated at the local mode by the Laplace approximation. In all cases the within temperature moves were tuned to have the suggested optimal acceptance rate of 0.234.

The hardest aspect of tuning ALPS was regarding the hot-state mode finder. As mentioned previously, finding ``narrow modes'' is never guaranteed in the finite run of the algorithm if their location is unknown \citep{Fong2019}. In both LAIS and ALPS where the hot state mode finder is required, it was found that an inverse temperature of 0.000005 worked very well and rapidly found the modes; typically within a few thousand iterations of the algorithm. As one would expect the performance of the mode finder is sensitive to its tuning and therein the algorithm's success also. The observations made in these empirical studies has suggested that the temperature of the hot state should also move dynamically to increase the robustness of this part of the algorithm. With the tuning of the hot state tuned as described, the hot state mode finder typically discovered all four mode-points within the first 4000 iterations of the algorithm.

Each of the three algorithms was run a total of 10 times, where for fair comparability the chains were all initiated from the first mode, $\mu_1$ in each instance. Furthermore, all chains were run long enough to generate 200,000 samples from the target temperature (including the burn-in samples). Figure~\ref{fig:TracePlotComp4} shows trace plots for the first component of the Markov chain in each of the three cases after a burn-in is removed. A successful outcome would be a plot where the trace line jumps between the marginal component's modes centred at values -20,-10,10 and 20. Figure~\ref{fig:TracePlotComp4} shows  ALPS was the only algorithm to successfully jump regularly between the different modes whilst in contrast the other two algorithms fail dramatically. LAIS makes only a couple of moves, due to the failure of the Laplace approximation to accurately capture the shape of such heavily skewed modes. PT performs even worse, with the chain remaining trapped in the (very spiky) first mode and demonstrating clearly the weakness of PT in high dimensions when modes have heterogeneous scales \cite{woodard2009conditions}.
\begin{figure}
  \includegraphics[width=0.9\linewidth]{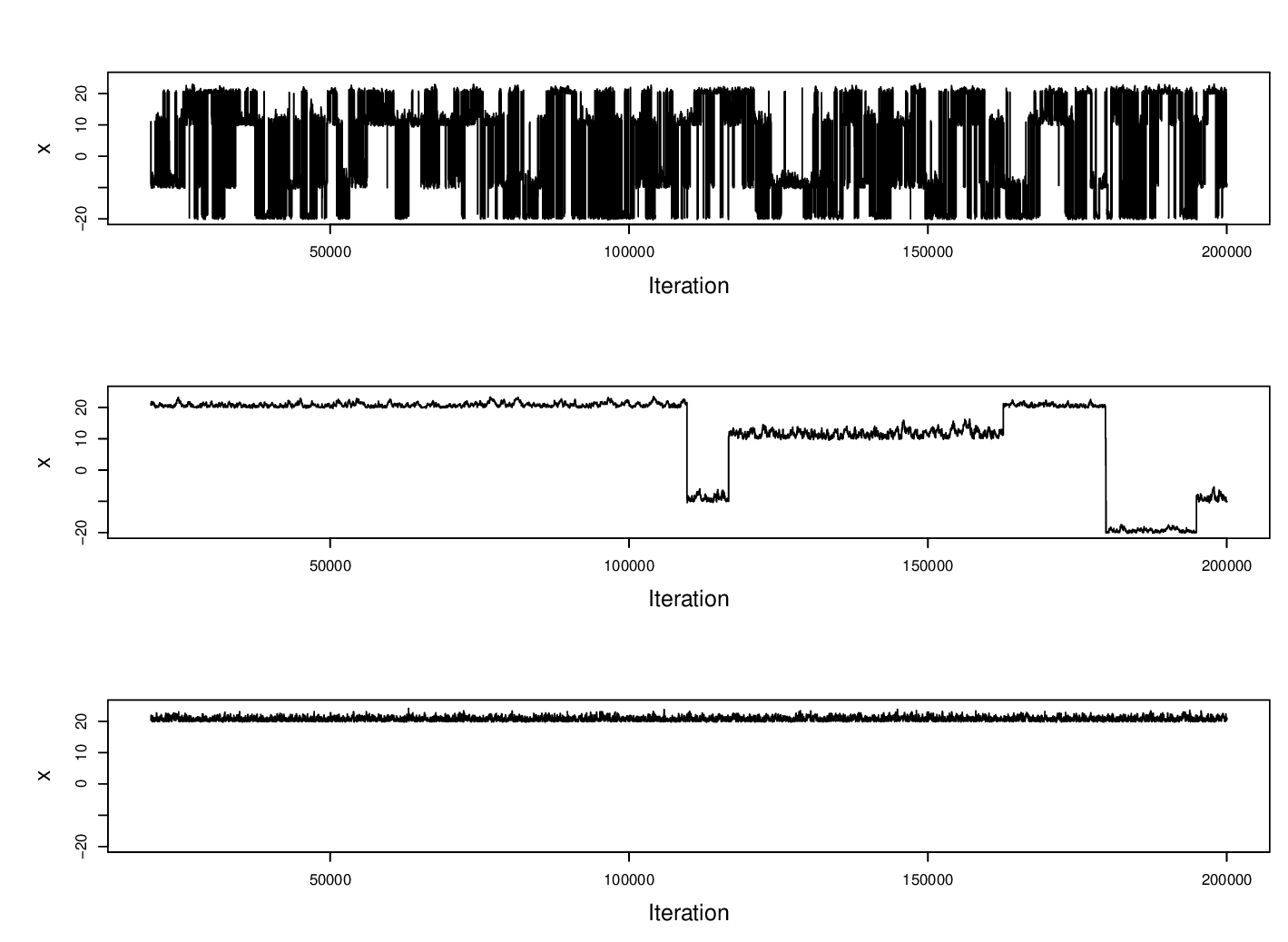}
  \caption{For the algorithms ALPS, LAIS and PT respectively, the trace plots of the first component of the Markov chain in a representative run from each of the 10 repetitions. ALPS (top) successfully explores the 4 modes centred on -20,-10,10 and 20 with rapid inter-modal mixing. In contrast LAIS only manages to make five jumps to another mode whilst the chain in PT remains trapped in the initialisation mode.}
  \label{fig:TracePlotComp4}
\end{figure}

The trace plots alone show that ALPS has vastly outperformed LAIS and PT.   Figure~\ref{fig:FuncCon} shows the running estimate of the probability that the first marginal component of the target distribution is less than $1/2$ for each of the 10 repetitions for each of the algorithms. The entity being estimated is $\mathbb{P}(X_1<1/2)$, which equals almost exactly 0.5. Therefore, for the $i^{th}$ iteration of the Markov chain with burn-in $b$, then the following functional is being plotted for each repetition
\begin{equation}
     f(i) = \frac{1}{i-b-1}\sum_{j=b+1}^i \mathbbm{1}_{\{x_1^j <1/2 \}}
		\label{eq:runningprobest}
\end{equation}

\begin{figure}
  \includegraphics[width=0.99\linewidth]{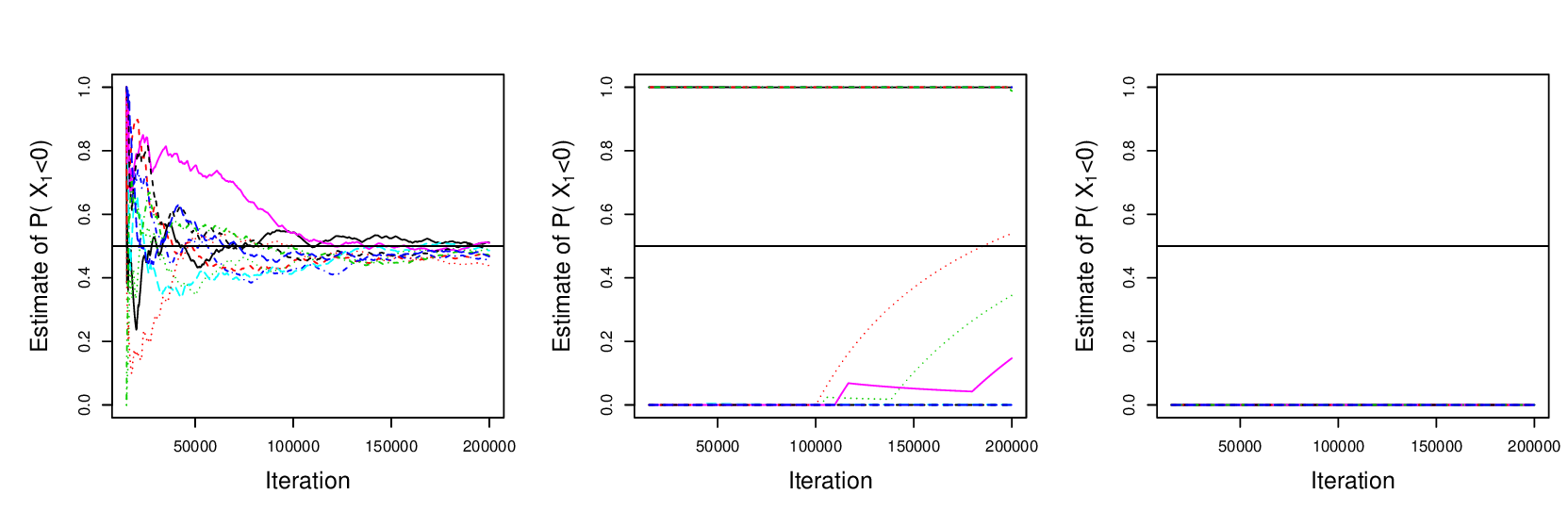}
  \caption{For the algorithms ALPS, LAIS and PT respectively, the function in equation~\ref{eq:runningprobest} with a burn-in $b=15,000$ is plotted for each of the 10 repetitions. For ALPS (leftmost plot) all trajectories concentrate around the truth by the end of their runs, whilst LAIS and PT fail dramatically in estimating this basic quantity.}
  \label{fig:FuncCon}
\end{figure}
The plot for the runs of ALPS show that all  repetitions converge successfully towards the true value of 0.5 whilst PT and LAIS fail to estimate this probability successfully.

All three algorithms have been run to give the same output quantity at the target state. However, it takes different amounts of time to run each algorithm per iteration. Table~\ref{tab:timecomp} shows the average run time for each of the three algorithms across each of the 10 repetitions. 
\begin{table}
\caption{\label{tab:timecomp}Comparison of the observed average run times for each of the three algorithms.}
\centering
\fbox{%
\begin{tabular}{c | c c c}
    Algorithm & ALPS & LAIS &  PT   \\ [0.5ex] 
    \hline 
   Run Time (sec) & 276.7 & 57.2 & 529.6 \\ [0.5ex] 
\end{tabular}}
\end{table}
The results in Table~\ref{tab:timecomp} are very complementary to ALPS since there appears to be little extra cost in running ALPS for what has been shown to be a dramatic gain in sample output. Indeed ALPS appears less expensive than PT and this is simply due to the vast number of extra temperatures that PT is updating per step of the chain in this example. However, in this example the number of modes was known and so the mode finding could be optimised for ALPS and LAIS; something not possible in a real data example where one may continue mode searching throughout the duration or even doing some pre-computation dedicated to mode searches (as was the case in the following real data example in Section~\ref{sec:SURexample}).

\clearpage

\section{Seemingly-Unrelated Regression Model}
\label{appendix:SUR}
SUR models involve a system of $M$ linear regression equations, one for each response, $m = 1, \dots, M$:
\begin{equation}
\vec{y}_m = \mathbf{X}_m \vec\theta_m + \vec\epsilon_m ,
\end{equation}
where $\vec{y}_m$ is a column vector of $N_m$ observations for the $m$th response, $\vec\theta_m$ is a vector of $J_m$ regression coefficients, $\mathbf{X}_m$ is an $N_m \times J_m$ design matrix of covariates, and $\vec\epsilon_m$ is a vector of residual errors. 

For notational convenience, we assume that $N_1 = N_2 = \dots = N_M = N$ and likewise that $J_1 = J_2 = \dots = J_M = J$, although in full generality that is not always the case. For example, \citet[p. 687]{Kmenta1986} included a SUR model with $J_1 = 2$ and $J_2 = 3$ covariates.

The vectors of observations can be stacked on top of each other to form a single, longer vector $\left(\vec{y}_1^T, \vec{y}_2^T, \dots, \vec{y}_M^T\right)^T$ and similarly with the regression coefficients  $\left(\vec{\theta}_1^T, \vec{\theta}_2^T, \dots, \vec{\theta}_M^T\right)^T$. The covariates $\mathbf{X}_m$ are combined to form a $MN \times MP$ block-diagonal matrix, resulting in the regression equation:
\begin{equation}\label{eq:surEq}
\left[\begin{array}{c}
\vec{y}_1 \\
\vec{y}_2 \\
\vdots \\
\vec{y}_M
\end{array}\right] = \left[\begin{array}{cccc}
\mathbf{X}_1 & & \dots & \mathbf{0} \\
& \mathbf{X}_2 & & \vdots \\
\vdots & & \ddots & \\
\mathbf{0} & \dots & & \mathbf{X}_M
\end{array}\right] \left[\begin{array}{c}
\vec\theta_1 \\
\vec\theta_2 \\
\vdots \\
\vec\theta_M
\end{array}\right]  + \left[\begin{array}{c}
\vec\epsilon_1 \\
\vec\epsilon_2 \\
\vdots \\
\vec\epsilon_M
\end{array}\right] .
\end{equation}
The errors $\vec\epsilon_m$ are assumed to be jointly multivariate normal,
$$\boldsymbol\epsilon \sim \mathcal{N}\left(\vec{0},\, \boldsymbol\Sigma_\epsilon \otimes \mathbf{I}\right) ,
$$
where $\mathbf{I}$ is a $N \times N$ identity matrix and $\boldsymbol\Sigma_\epsilon$ is a $M \times M$ variance-covariance matrix with entries $\sigma^2_{\ell,m}$. 
\cite{Zellner1962} introduced an asymptotically-unbiased estimator for $\boldsymbol\Sigma_\epsilon$, conditional on estimates $\hat{\boldsymbol\theta}$ of the regression coefficients:
$$\hat\sigma^2_{\ell,m} = \frac{1}{N} \left(\vec{y}_\ell - \mathbf{X}_\ell \hat\theta_\ell \right)^T\left(\vec{y}_m - \mathbf{X}_m \hat\theta_m \right) .
$$
Given an estimate of $\widehat{\boldsymbol\Sigma}_\epsilon$, $\hat{\boldsymbol\theta}$ can be estimated using generalised least squares \citep{Aitken1936}. Let $\widehat{\boldsymbol\Omega} = ( \widehat{\boldsymbol\Sigma}_\epsilon \otimes \mathbf{I} )^{-1}$, then
$$\hat{\boldsymbol\theta} = \left( \mathbf{X}^T \widehat{\boldsymbol\Omega} \mathbf{X} \right)^{-1} \mathbf{X}^T \widehat{\boldsymbol\Omega} \mathbf{y} .
$$
Alternating between these two estimators results in an iterative procedure that converges to a solution to the simultaneous regression equations.

The profile log-likelihood of the SUR model is thus:
\begin{equation}\label{eq:surLike}
\ell(\mathbf{Y} \mid \hat{\boldsymbol\theta}) = -N \log\{2\pi\} - \frac{N}{2}\log\left\{\left| \widehat{\boldsymbol\Sigma}_\epsilon \right|\right\} - N
\end{equation}

Due to the quadratic form of the SUR model, it was long assumed that its likelihood was log-concave and therefore the iterative estimation procedure of \cite{Zellner1962} was guaranteed to converge to the global maximum. For example, see \citet[p. 157]{Srivastava1987} and \citet[p. 686]{Greene1997}. However, in fact these models are an example of a curved exponential family \citep{Sundberg2010}.
\cite{Drton2004} showed in the bivariate case that maximising \eqref{eq:surLike} is equivalent to finding the real roots of a fifth-degree polynomial. The profile log-likelihood preserves all stationary points of the log-likelihood, so if the likelihood is multi-modal then so will $\ell(\mathbf{Y} \mid \boldsymbol\theta)$ \citep[Lemma 1]{Drton2004}. This means that the likelihood could have one, two, or three local maxima.

\subsection{Further results}\label{appendix:Grunfeld}

Additional results for ALPS with the SUR dataset from \cite{Grunfeld1958} with $N=15$ years of observations for $M=5$ manufacturing firms and with $J=3$ covariates per firm, are shown in Figures~\ref{f:sur15param1}--\ref{f:sur15param15}. For each parameter, panels~(a) and~(b) show marginal traceplots at the target inverse temperature $\beta=1$ and coldest inverse temperature $\beta_{\max} = 32$; panel~(c) shows the marginal histogram at the coldest inverse temperature. Red dashed lines indicate the marginal locations of the modes discovered by EC.

\begin{figure}[H]\centering
  \begin{tabular}[b]{c}\includegraphics[width=0.25\linewidth]{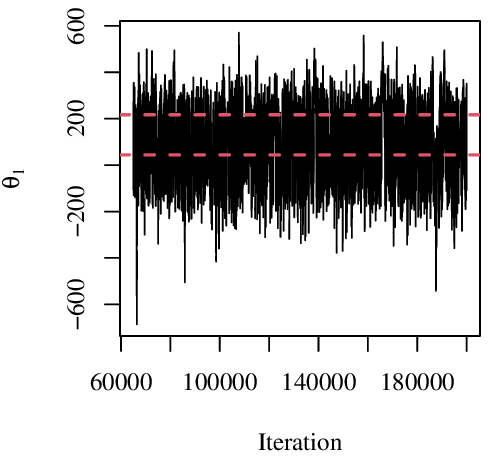} \\
    \small \textsf{(a) Traceplot, $\theta_1$, $\beta=1$.}\end{tabular}
  \begin{tabular}[b]{c}\includegraphics[width=0.25\linewidth]{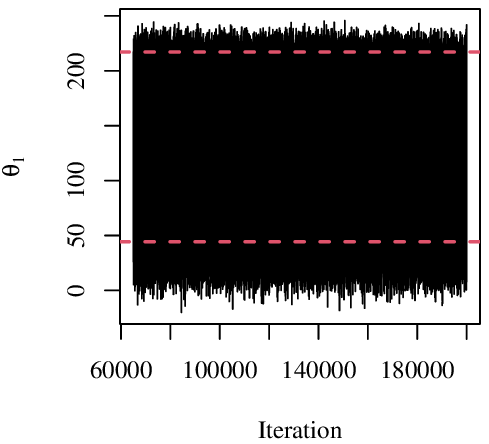} \\
    \small \textsf{(b) Traceplot, $\theta_1$, $\beta_{\max}=32$.}\end{tabular}
  \begin{tabular}[b]{c}\includegraphics[width=0.25\linewidth]{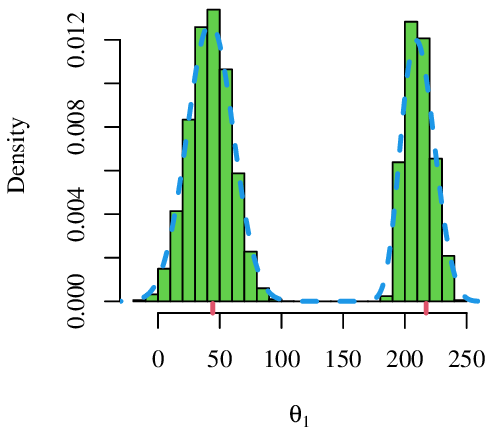} \\
    \small \textsf{(c) Histogram, $\theta_1$, $\beta_{\max}=32$.}\end{tabular}
\caption{Results for the marginal $\theta_1$ chain.}\label{f:sur15param1}
\end{figure}

\begin{figure}[H]\centering
  \begin{tabular}[b]{c}\includegraphics[width=0.25\linewidth]{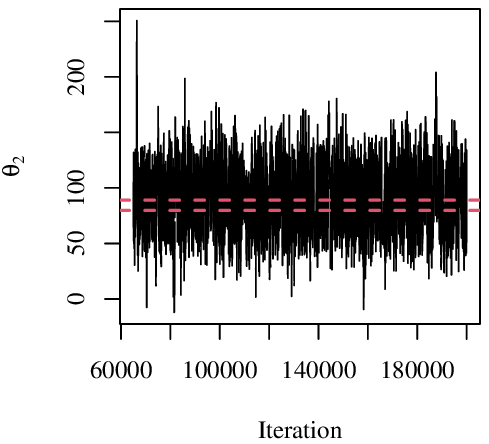} \\
    \small \textsf{(a) Traceplot, $\theta_2$, $\beta=1$.}\end{tabular}
  \begin{tabular}[b]{c}\includegraphics[width=0.25\linewidth]{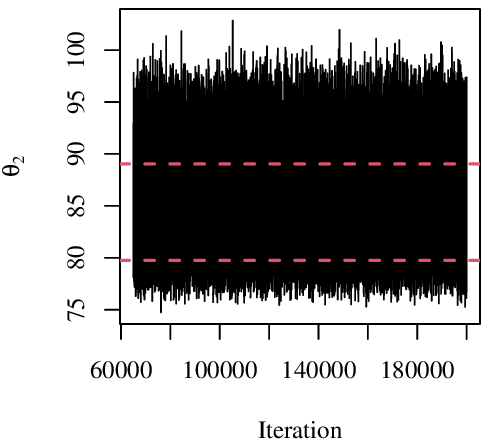} \\
    \small \textsf{(b) Traceplot, $\theta_2$, $\beta_{\max}=32$.}\end{tabular}
  \begin{tabular}[b]{c}\includegraphics[width=0.25\linewidth]{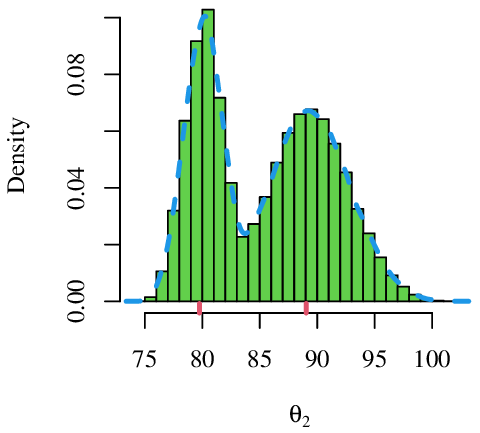} \\
    \small \textsf{(c) Histogram, $\theta_2$, $\beta_{\max}=32$.}\end{tabular}
\caption{Results for the marginal $\theta_2$ chain.}\label{f:sur15param2}
\end{figure}

\begin{figure}[H]\centering
  \begin{tabular}[b]{c}\includegraphics[width=0.25\linewidth]{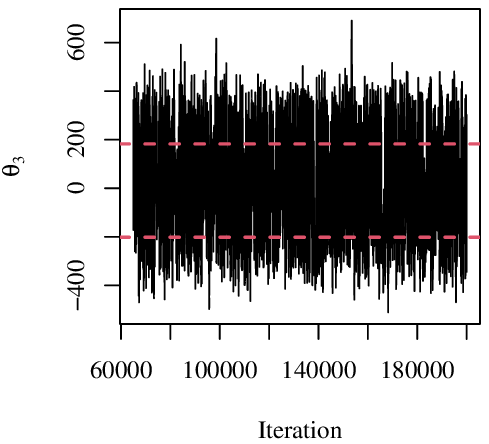} \\
    \small \textsf{(a) Traceplot, $\theta_3$, $\beta=1$.}\end{tabular}
  \begin{tabular}[b]{c}\includegraphics[width=0.25\linewidth]{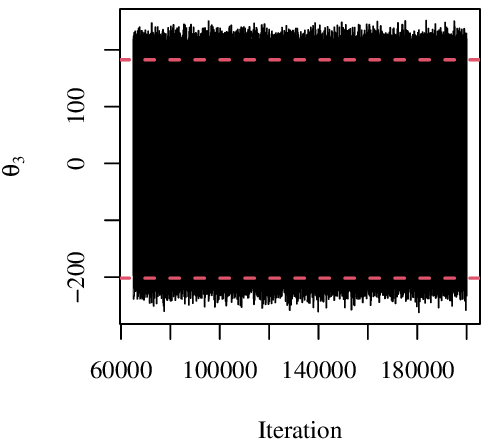} \\
    \small \textsf{(b) Traceplot, $\theta_3$, $\beta_{\max}=32$.}\end{tabular}
  \begin{tabular}[b]{c}\includegraphics[width=0.25\linewidth]{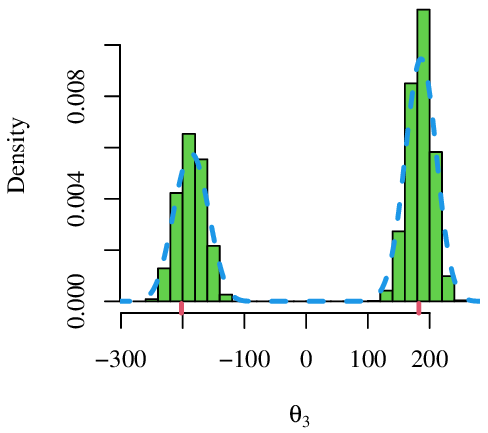} \\
    \small \textsf{(c) Histogram, $\theta_3$, $\beta_{\max}=32$.}\end{tabular}
\caption{Results for the marginal $\theta_3$ chain.}\label{f:sur15param3}
\end{figure}

\begin{figure}[H]\centering
  \begin{tabular}[b]{c}\includegraphics[width=0.25\linewidth]{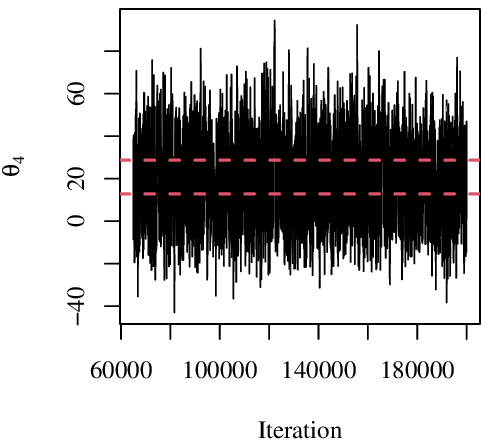} \\
    \small \textsf{(a) Traceplot, $\theta_4$, $\beta=1$.}\end{tabular}
  \begin{tabular}[b]{c}\includegraphics[width=0.25\linewidth]{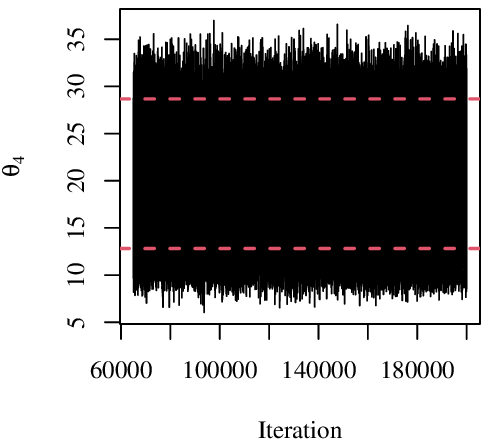} \\
    \small \textsf{(b) Traceplot, $\theta_4$, $\beta_{\max}=32$.}\end{tabular}
  \begin{tabular}[b]{c}\includegraphics[width=0.25\linewidth]{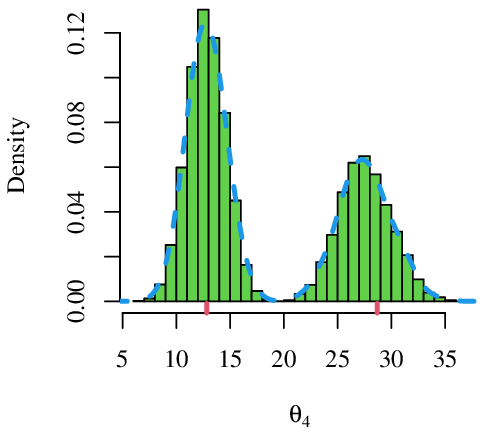} \\
    \small \textsf{(c) Histogram, $\theta_4$, $\beta_{\max}=32$.}\end{tabular}
\caption{Results for the marginal $\theta_4$ chain.}\label{f:sur15param4}
\end{figure}

\begin{figure}[H]\centering
  \begin{tabular}[b]{c}\includegraphics[width=0.25\linewidth]{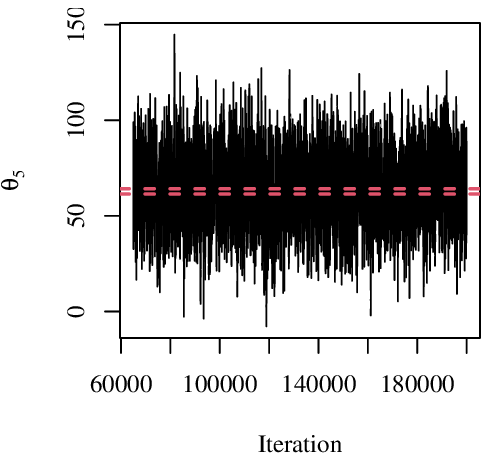} \\
    \small \textsf{(a) Traceplot, $\theta_5$, $\beta=1$.}\end{tabular}
  \begin{tabular}[b]{c}\includegraphics[width=0.25\linewidth]{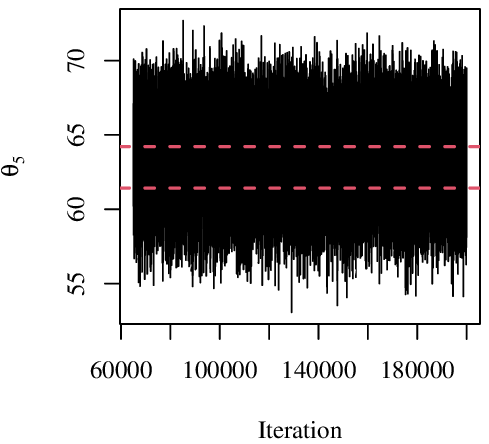} \\
    \small \textsf{(b) Traceplot, $\theta_5$, $\beta_{\max}=32$.}\end{tabular}
  \begin{tabular}[b]{c}\includegraphics[width=0.25\linewidth]{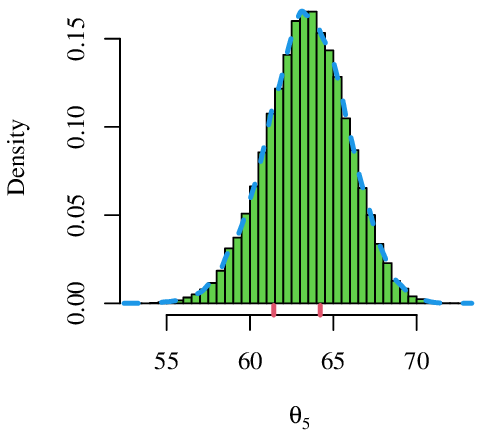} \\
    \small \textsf{(c) Histogram, $\theta_5$, $\beta_{\max}=32$.}\end{tabular}
\caption{Results for the marginal $\theta_5$ chain.}\label{f:sur15param5}
\end{figure}

\begin{figure}[H]\centering
  \begin{tabular}[b]{c}\includegraphics[width=0.25\linewidth]{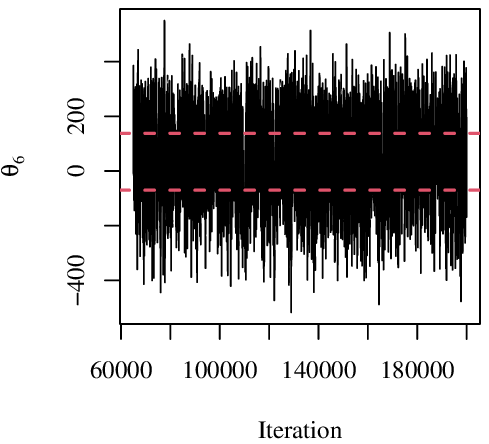} \\
    \small \textsf{(a) Traceplot, $\theta_6$, $\beta=1$.}\end{tabular}
  \begin{tabular}[b]{c}\includegraphics[width=0.25\linewidth]{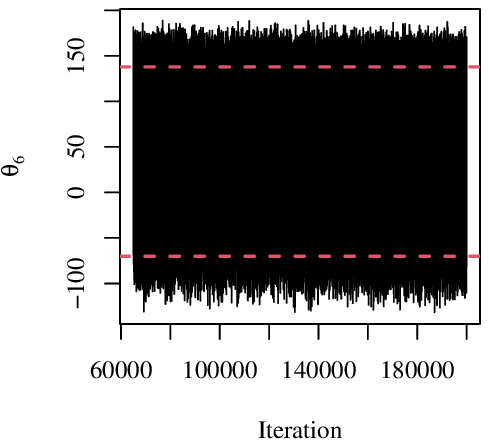} \\
    \small \textsf{(b) Traceplot, $\theta_6$, $\beta_{\max}=32$.}\end{tabular}
  \begin{tabular}[b]{c}\includegraphics[width=0.25\linewidth]{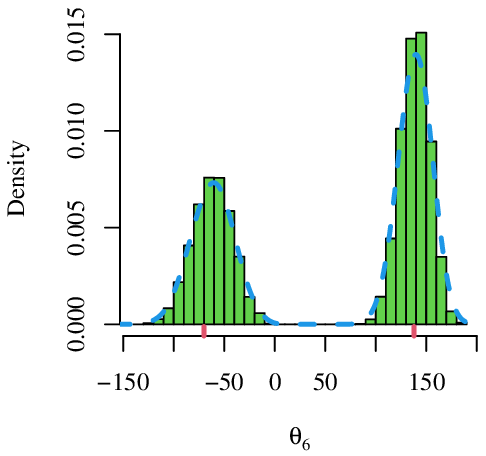} \\
    \small \textsf{(c) Histogram, $\theta_6$, $\beta_{\max}=32$.}\end{tabular}
\caption{Results for the marginal $\theta_6$ chain.}\label{f:sur15param6}
\end{figure}

\begin{figure}[H]\centering
  \begin{tabular}[b]{c}\includegraphics[width=0.25\linewidth]{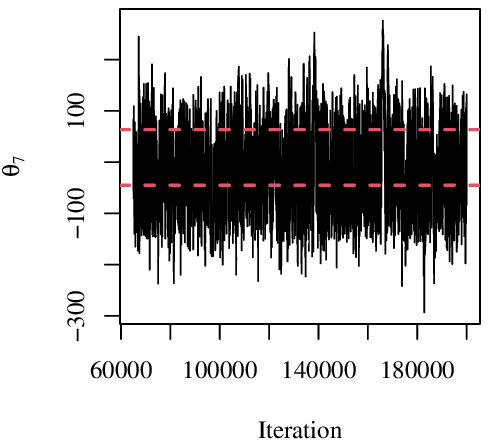} \\
    \small \textsf{(a) Traceplot, $\theta_7$, $\beta=1$.}\end{tabular}
  \begin{tabular}[b]{c}\includegraphics[width=0.25\linewidth]{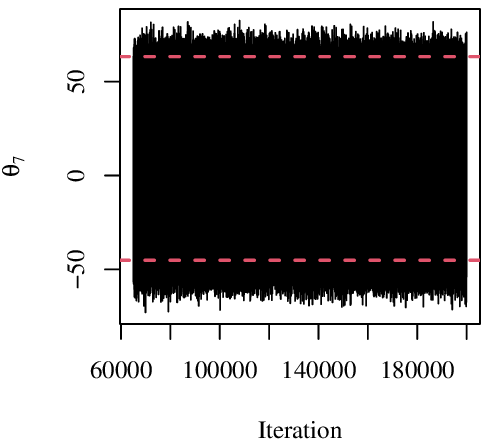} \\
    \small \textsf{(b) Traceplot, $\theta_7$, $\beta_{\max}=32$.}\end{tabular}
  \begin{tabular}[b]{c}\includegraphics[width=0.25\linewidth]{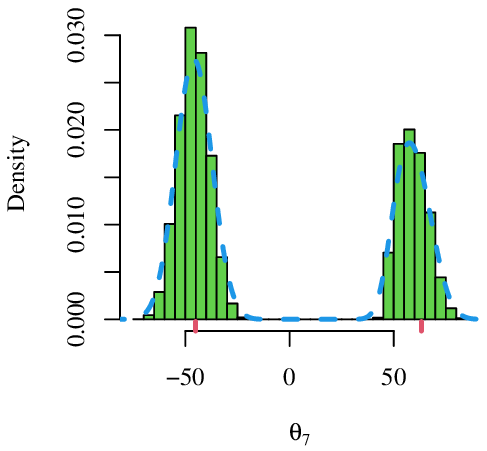} \\
    \small \textsf{(c) Histogram, $\theta_7$, $\beta_{\max}=32$.}\end{tabular}
\caption{Results for the marginal $\theta_7$ chain.}\label{f:sur15param7}
\end{figure}

\begin{figure}[H]\centering
  \begin{tabular}[b]{c}\includegraphics[width=0.25\linewidth]{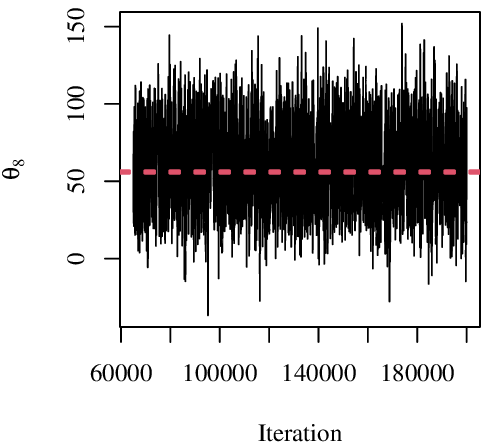} \\
    \small \textsf{(a) Traceplot, $\theta_8$, $\beta=1$.}\end{tabular}
  \begin{tabular}[b]{c}\includegraphics[width=0.25\linewidth]{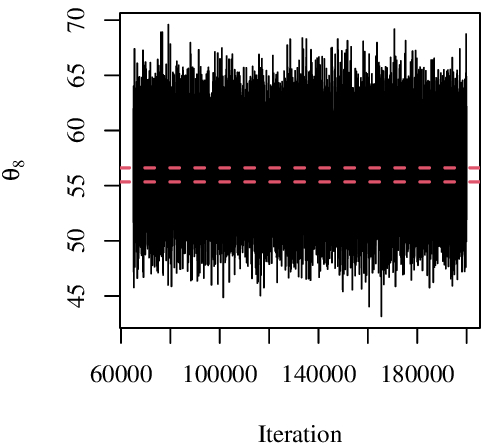} \\
    \small \textsf{(b) Traceplot, $\theta_8$, $\beta_{\max}=32$.}\end{tabular}
  \begin{tabular}[b]{c}\includegraphics[width=0.25\linewidth]{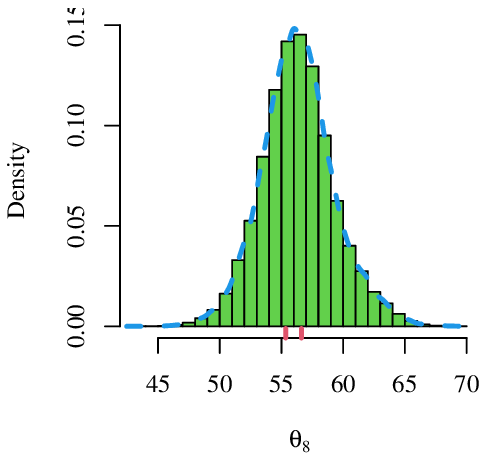} \\
    \small \textsf{(c) Histogram, $\theta_8$, $\beta_{\max}=32$.}\end{tabular}
\caption{Results for the marginal $\theta_8$ chain.}\label{f:sur15param8}
\end{figure}

\begin{figure}[H]\centering
  \begin{tabular}[b]{c}\includegraphics[width=0.25\linewidth]{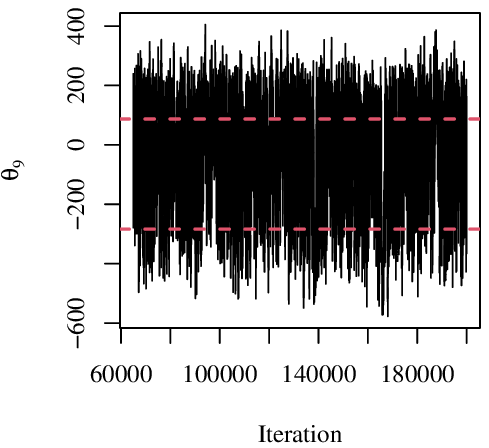} \\
    \small \textsf{(a) Traceplot, $\theta_9$, $\beta=1$.}\end{tabular}
  \begin{tabular}[b]{c}\includegraphics[width=0.25\linewidth]{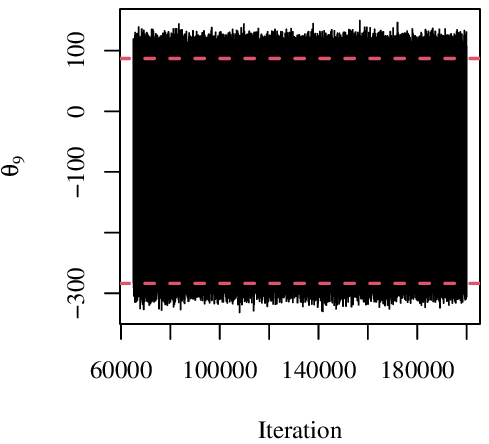} \\
    \small \textsf{(b) Traceplot, $\theta_9$, $\beta_{\max}=32$.}\end{tabular}
  \begin{tabular}[b]{c}\includegraphics[width=0.25\linewidth]{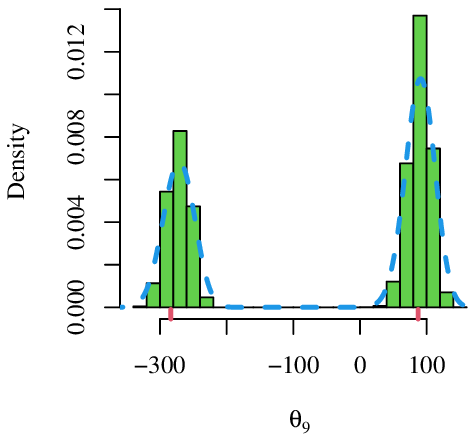} \\
    \small \textsf{(c) Histogram, $\theta_9$, $\beta_{\max}=32$.}\end{tabular}
\caption{Results for the marginal $\theta_9$ chain.}\label{f:sur15param9}
\end{figure}

\begin{figure}[H]\centering
  \begin{tabular}[b]{c}\includegraphics[width=0.25\linewidth]{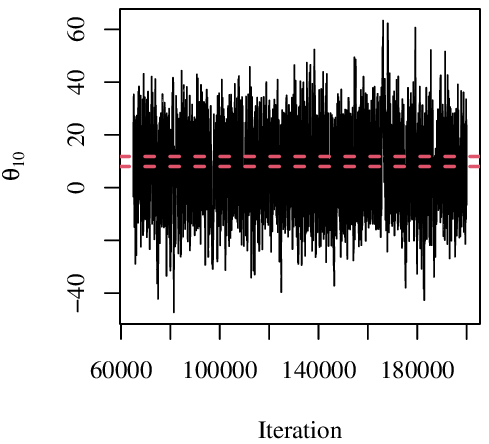} \\
    \small \textsf{(a) Traceplot, $\theta_{10}$, $\beta=1$.}\end{tabular}
  \begin{tabular}[b]{c}\includegraphics[width=0.25\linewidth]{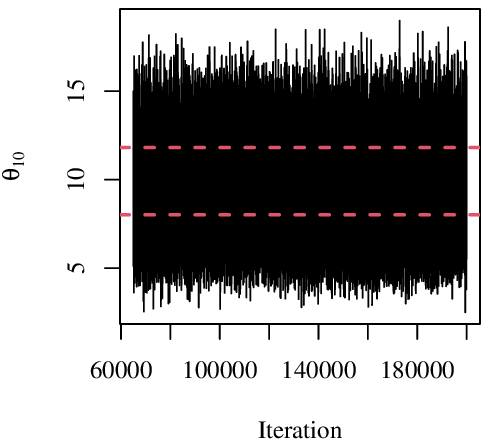} \\
    \small \textsf{(b) Traceplot, $\theta_{10}$, $\beta_{\max}=32$.}\end{tabular}
  \begin{tabular}[b]{c}\includegraphics[width=0.25\linewidth]{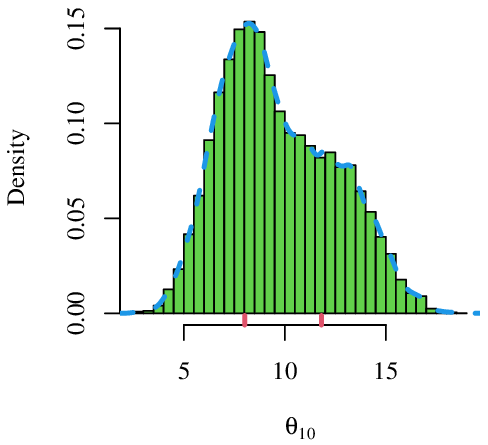} \\
    \small \textsf{(c) Histogram, $\theta_{10}$, $\beta_{\max}=32$.}\end{tabular}
\caption{Results for the marginal $\theta_{10}$ chain.}\label{f:sur15param10}
\end{figure}

\begin{figure}[H]\centering
  \begin{tabular}[b]{c}\includegraphics[width=0.25\linewidth]{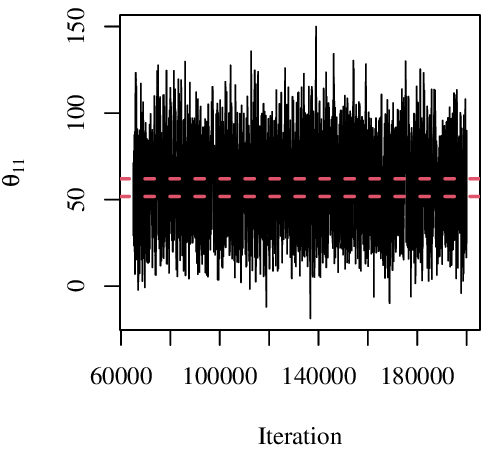} \\
    \small \textsf{(a) Traceplot, $\theta_{11}$, $\beta=1$.}\end{tabular}
  \begin{tabular}[b]{c}\includegraphics[width=0.25\linewidth]{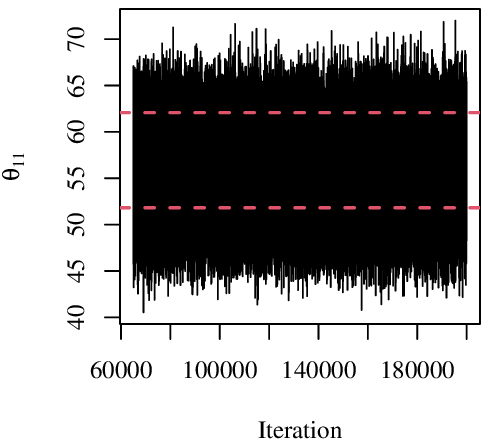} \\
    \small \textsf{(b) Traceplot, $\theta_{11}$, $\beta_{\max}=32$.}\end{tabular}
  \begin{tabular}[b]{c}\includegraphics[width=0.25\linewidth]{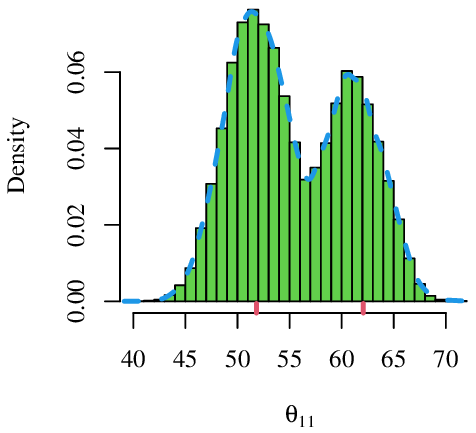} \\
    \small \textsf{(c) Histogram, $\theta_{11}$, $\beta_{\max}=32$.}\end{tabular}
\caption{Results for the marginal $\theta_{11}$ chain.}\label{f:sur15param11}
\end{figure}

\begin{figure}[H]\centering
  \begin{tabular}[b]{c}\includegraphics[width=0.25\linewidth]{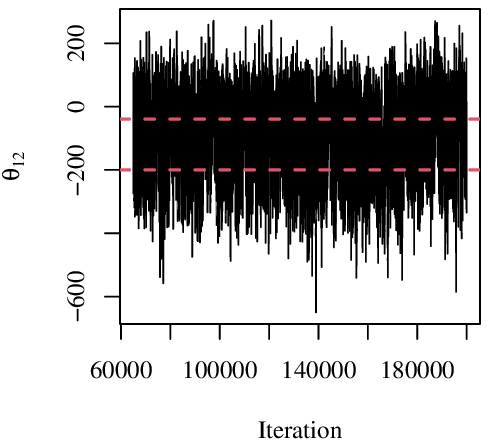} \\
    \small \textsf{(a) Traceplot, $\theta_{12}$, $\beta=1$.}\end{tabular}
  \begin{tabular}[b]{c}\includegraphics[width=0.25\linewidth]{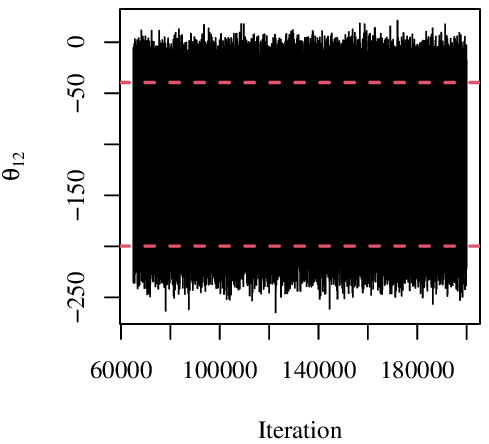} \\
    \small \textsf{(b) Traceplot, $\theta_{12}$, $\beta_{\max}=32$.}\end{tabular}
  \begin{tabular}[b]{c}\includegraphics[width=0.25\linewidth]{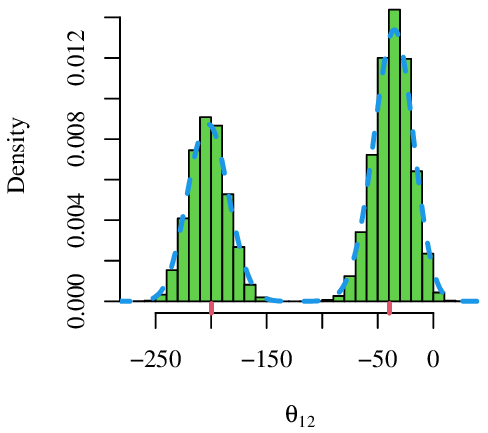} \\
    \small \textsf{(c) Histogram, $\theta_{12}$, $\beta_{\max}=32$.}\end{tabular}
\caption{Results for the marginal $\theta_{12}$ chain.}\label{f:sur15param12}
\end{figure}

\begin{figure}[H]\centering
  \begin{tabular}[b]{c}\includegraphics[width=0.25\linewidth]{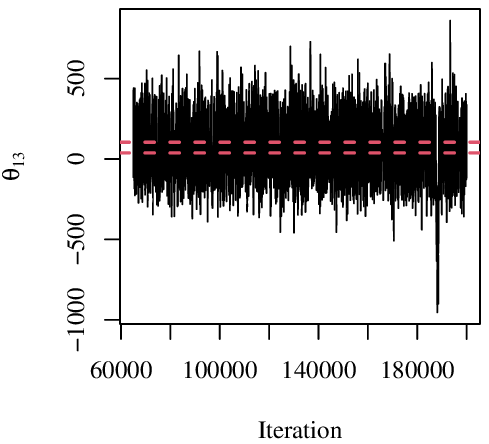} \\
    \small \textsf{(a) Traceplot, $\theta_{13}$, $\beta=1$.}\end{tabular}
  \begin{tabular}[b]{c}\includegraphics[width=0.25\linewidth]{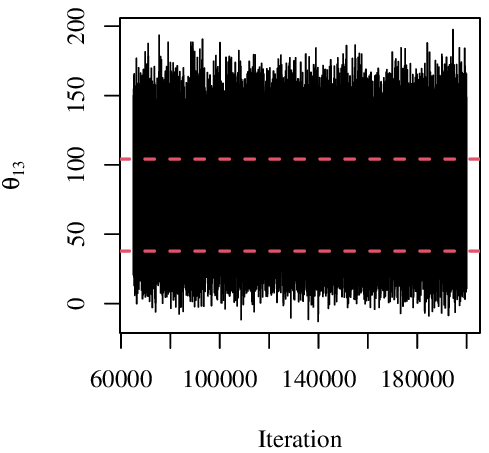} \\
    \small \textsf{(b) Traceplot, $\theta_{13}$, $\beta_{\max}=32$.}\end{tabular}
  \begin{tabular}[b]{c}\includegraphics[width=0.25\linewidth]{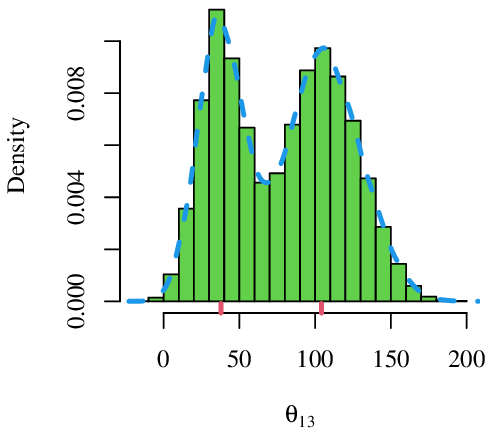} \\
    \small \textsf{(c) Histogram, $\theta_{13}$, $\beta_{\max}=32$.}\end{tabular}
\caption{Results for the marginal $\theta_{13}$ chain.}\label{f:sur15param13}
\end{figure}

\begin{figure}[H]\centering
  \begin{tabular}[b]{c}\includegraphics[width=0.25\linewidth]{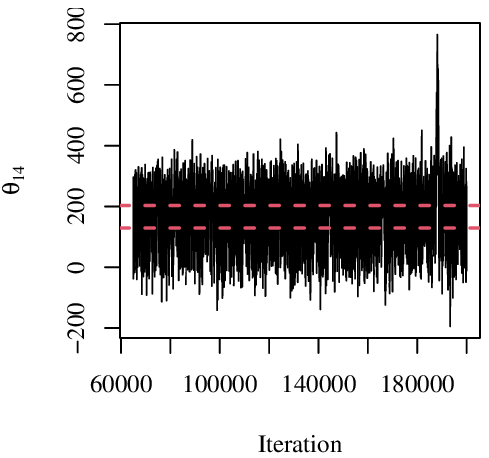} \\
    \small \textsf{(a) Traceplot, $\theta_{14}$, $\beta=1$.}\end{tabular}
  \begin{tabular}[b]{c}\includegraphics[width=0.25\linewidth]{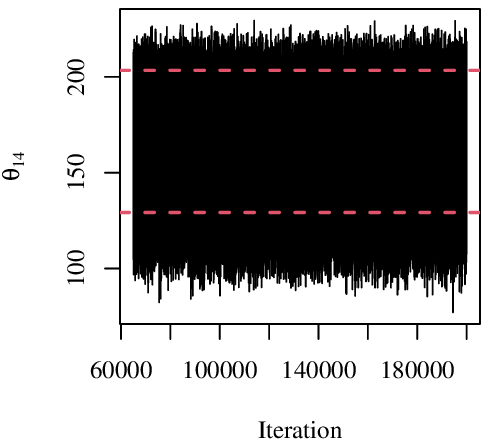} \\
    \small \textsf{(b) Traceplot, $\theta_{14}$, $\beta_{\max}=32$.}\end{tabular}
  \begin{tabular}[b]{c}\includegraphics[width=0.25\linewidth]{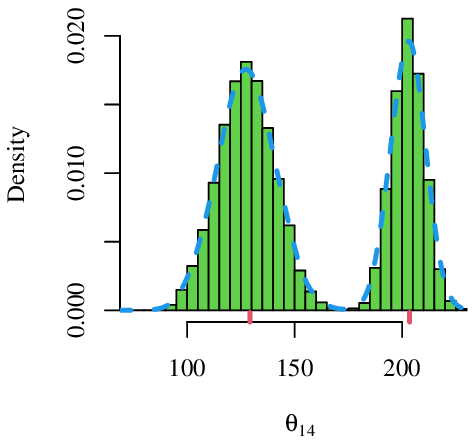} \\
    \small \textsf{(c) Histogram, $\theta_{14}$, $\beta_{\max}=32$.}\end{tabular}
\caption{Results for the marginal $\theta_{14}$ chain.}\label{f:sur15param14}
\end{figure}

\begin{figure}[H]\centering
  \begin{tabular}[b]{c}\includegraphics[width=0.25\linewidth]{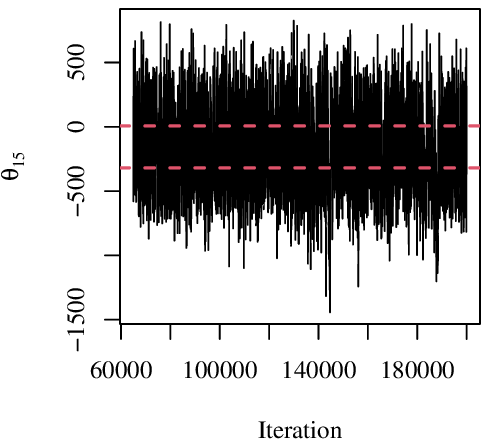} \\
    \small \textsf{(a) Traceplot, $\theta_{15}$, $\beta=1$.}\end{tabular}
  \begin{tabular}[b]{c}\includegraphics[width=0.25\linewidth]{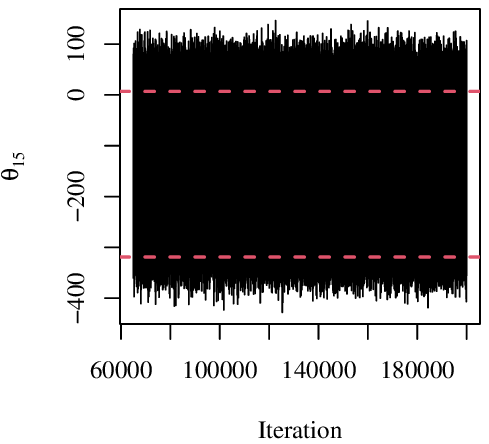} \\
    \small \textsf{(b) Traceplot, $\theta_{15}$, $\beta_{\max}=32$.}\end{tabular}
  \begin{tabular}[b]{c}\includegraphics[width=0.25\linewidth]{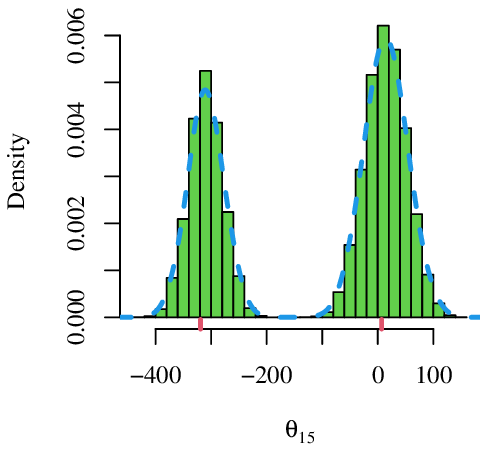} \\
    \small \textsf{(c) Histogram, $\theta_{15}$, $\beta_{\max}=32$.}\end{tabular}
\caption{Results for the marginal $\theta_{15}$ chain.}\label{f:sur15param15}
\end{figure}

\clearpage

\bibliographystyle{plainnat}
\bibliography{biblisim}
\end{document}